\documentclass[10pt,english]{article}
\usepackage[T1]{fontenc}
\usepackage[utf8]{inputenc}
\usepackage{geometry}
\usepackage{float}
\usepackage{amsmath}
\usepackage{amsthm}
\usepackage{amssymb}
\usepackage{graphicx}
\usepackage{setspace}
\usepackage{esint}
\usepackage[colorlinks=true]{hyperref}
\hypersetup{urlcolor=blue, citecolor=red}

\makeatletter

\numberwithin{equation}{section}
\numberwithin{figure}{section}
\theoremstyle{plain}
\newtheorem{thm}{\protect\theoremname}[section]
  \theoremstyle{plain}
  \newtheorem{lem}[thm]{\protect\lemmaname}
  \theoremstyle{remark}
  \newtheorem{rem}[thm]{\protect\remarkname}

\@ifundefined{date}{}{\date{}}
\usepackage{pgf,tikz}
\usepackage{mathrsfs}
\usetikzlibrary{arrows}
\usepackage{authblk}

\makeatother

\usepackage{babel}
  \providecommand{\lemmaname}{Lemma}
  \providecommand{\remarkname}{Remark}
\providecommand{\theoremname}{Theorem}

\begin{document}
\title{The rigorous derivation of the Linear Landau equation from a particle system in a weak-coupling limit} 
\author[*]{N. Catapano}
\affil[*]{Dipartimento di Matematica, Università di Roma “La Sapienza”, P.le A. Moro, 5, 00185 Roma, Italy. \authorcr E-mail: catapano@mat.uniroma1.it} 
\maketitle
\begin{abstract}
We consider a system of N particles interacting via a short-range
smooth potential, in a weak-coupling regime. This means that the number
of particles $N$ goes to infinity and the range of the potential
$\epsilon$ goes to zero in such a way that $N\epsilon^{2}=\alpha$,
with $\alpha$ diverging in a suitable way. We provide a rigorous
derivation of the Linear Landau equation from this particle system.
The strategy of the proof consists in showing the asymptotic equivalence
between the one-particle marginal and the solution of the linear Boltzmann
equation with vanishing mean free path. This point follows \cite{Bodineau2015}
and makes use of technicalities developed in \cite{Pulvirenti2014}.
Then, following the ideas of Landau, we prove the asympotic equivalence
between the solutions of the Boltzmann and Landau linear equation
in the grazing collision limit.
\end{abstract}
\tableofcontents{}\pagebreak{}

\section{Introduction}

\subsection{The Boltzmann-Grad limit}

In kinetic theory a gas is described by a system of small indistinguishable
interacting particles. The evolution of this system is quite complicated
since the order of particles involved is quite large. For this reason
it is interesting to consider the system from a statistical point
of view. The starting point is a system of $N$ particles having unitary
mass and moving in a domain $D\subseteq\mathbb{R}^{3}$. These particles
can interact by means of a short-range radial potential $\Phi$. The
microscopic state of the system is given by the position and velocity
variables denoted by $\boldsymbol{q}_{N}=\left(q_{1},q_{2},...,q_{N}\right)$
and $\boldsymbol{v}_{N}=\left(v_{1},v_{2},...,v_{N}\right)$, where
$q_{i},v_{i}$ are respectively position and velocity of the i-th
particle. The time is denoted by $\tau$. Throughout the paper we
will use bold letters for vectors of variables.

Let $\epsilon>0$ be a parameter denoting the ratio between typical
macroscopic and microscopic scales, say the inverse of the number
of atomic diameters necessary to fill a centimeter. If we want a macroscopid
description of the system it is natural to introduce macroscopic variables
defined by
\begin{equation}
\boldsymbol{x}_{N}=\boldsymbol{q}_{N}\epsilon\,\,\,\,t=\tau\epsilon\label{eq:scalingmic}
\end{equation}
where $\boldsymbol{x}_{N}=\left(x_{1},x_{2},...,x_{N}\right)$ are
the macroscopic position and $t$ is the macroscopic time variable.
Notice that the velocities are unscaled. From the Liouville equation
for the particle dynamic it is possible to derive a hierarchy of equations
for the j-particles marginal probability density function, with $j\leq N$.
In the case of hard spheres we found the following BBGKY hierarchy
\[
\left(\partial_{t}+\boldsymbol{v}_{j}\cdot\nabla_{\boldsymbol{x}_{j}}\right)f_{j}^{N}=(N-j)\epsilon^{2}\sum_{k=1}^{j}\intop_{\mathbb{R}^{3}}dv_{j+1}\intop_{\nu\cdot(v_{k}-v_{j+1})\geq0}
\]
\begin{equation}
\vert\nu\cdot(v_{k}-v_{j+1})\vert\left[f_{j+1}^{N}(x_{1},v_{1},...,x_{k},v_{k}^{'},...,x_{j},v_{j},x_{k}-\eta\epsilon,v_{j+1}^{'})-f_{j+1}^{N}(x_{1},v_{1},...,x_{j},v_{j},x_{k}+\eta\epsilon,v_{j+1})\right]\label{eq:bbgky}
\end{equation}
where $\nu=\frac{x_{j_{+1}}-x_{k}}{\vert x_{j+1}-x_{k}\vert}$ and
$v_{k}^{'}=v_{k}-\nu\left[\nu\cdot\left(v_{k}-v_{j+1}\right)\right]$,
$v_{j+1}^{'}=v_{j+1}+\nu\left[\nu\cdot\left(v_{k}-v_{j+1}\right)\right]$
. Equations \ref{eq:bbgky} were first formally derived by \cite{Cercignani1994},
then a rigorous analysis has been done by \cite{Uchiyama1988,Spohn2006,Simonella2014,Pulvirenti2014a}.

Scaling according to $N\rightarrow\infty$ and $\epsilon\rightarrow0$,
in such a way that $N\epsilon^{2}\cong1$, we are in a low-density
regime suitable for the description of a rarified gas. This kind of
scaling is usually called the Boltzmann-Grad limit. The formal Boltzmann-Grad
limit in the BBGKY gives a new hierarchy of equations called the Boltzmann
hierarchy. The central idea in kinetic theory is the concept of propagation
of chaos, namely, if the initial datum factorizes, i.e. $f_{0,j}(\boldsymbol{x}_{j},\boldsymbol{v}_{j})=\prod_{i=1}^{j}f_{0,1}(x_{i},v_{i})$,
then also the solution at time $t$ factorizes:

\begin{equation}
f_{j}(\boldsymbol{x}_{j},\boldsymbol{v}_{j})=\prod_{i=1}^{j}f_{1}(x_{i},v_{i}).
\end{equation}
Actually the Boltzmann hierarchy admits factorized solutions so that
it is compatible with the propagation of chaso and under this hypothesis,
which however must be proved froma rigorous view point, the first
equation of this hierarchy is the Boltzmann equation 
\begin{equation}
\partial_{t}f+v\cdot\nabla_{x}f=\intop d\nu dv_{1}B(\nu,v-v_{1})\left[f(x,v^{'})f(x,v_{1}^{'})-f(x,v)f(x,v_{1})\right].\label{eq:Bzeq}
\end{equation}
However, as soon as $\epsilon>0$ propagation of chaos does not hold
because the evolution creates correlation between particles so that
we cannot describe the system in terms of a single equation for the
one-particle marginal and this is the reason why the Boltzmann equation
can describe in a more handable way the statistical evolution of a
gas. 

The validity of the Boltzmann equation is a fundamental problem in
kinetic theory. It consists in proving that the solution of the BBGKY
hierarchy for hard spheres converge in the Boltzmann-Grad limit to
the solution of the Boltzmann hierarchy. This means that the propagation
of chaos is recovered in the limit.

The rigorous derivation of the Boltzmann equation was first proved
by Lanford in 1975 \cite{Lanford1975} in the case of an hard spheres
system for a small time. The main idea of the Lanford work is to write
the solution of the BBGKY hierarchy for hard spheres and of the Boltzmann
hierarchy as a perturbative series of the free evolution and then
prove that the series solution of the BBGKY converge to the series
solution of the Boltzmann hierarchy.

More recently Gallagher, Saint-Raymond and Texier \cite{Gallagher}
and Pulvirenti, Saffirio and Simonella \cite{Pulvirenti2014} proved
the rigorous derivation of the Boltzmann equation, for a small time,
starting from a system of particle interacting by means of a short-range
potential providing an explicit rate of convergence. In the case of
a short-range potential the starting hierarchy is no more the BBGKY
hierarchy but the Grad hierarchy, that was developed by Grad in \cite{Grad1958}.

\subsection{The linear case}

The linear Boltzmann equation describes the evolution of a tagged
particle in a random stationary background at equlibrium and reads
as follows
\begin{equation}
\partial_{t}g^{\alpha}+v\cdot\nabla_{x}g^{\alpha}=\alpha\intop dv_{1}M_{\beta}(v_{1})\intop d\nu B(\nu,v-v_{1})\left[g^{\alpha}(x,v^{'})-g^{\alpha}(x,v)\right]
\end{equation}
where $M_{\beta}(v_{1})=\frac{1}{C_{\beta}}e^{-\beta\vert v_{1}\vert^{2}}$
and $C_{\beta}$ is chosen in such a way that $\intop dv_{1}M_{\beta}(v_{1})=1$.
The linear Boltzmann equation can be obtained from the equation (\ref{eq:Bzeq})
setting $f(x,v)=g^{\alpha}(x,v)M_{\beta}(v)$ and $f(x,v_{1})=M_{\beta}(v_{1})$,
$g^{\alpha}$ is the evolution of the perturbation in the stationary
background given by $M_{\beta}(v)$.

The derivation of the linear Boltzmann equation from an hard spheres
system has been proved for an arbitrary time by Spohn, Lebowitz \cite{J.Lebowitz1982}
and more recently quantitative estimates on the rate of convergence
have been obtained by Bodineau, Gallagher and Saint-Raymond \cite{Bodineau2015}.
A different type of linear Boltzmann equation has been derived in
the case of a Lorentz gas in Ref.s \cite{Gallavotti1972,Basile2014}.

\subsection{A different scaling}

A different scaling can be used to study a different regime from the
low density. In case of particles interacting by means of a short-range
radial potential $\Phi,$ we rescale position and time as in (\ref{eq:scalingmic})
but we set $N\epsilon^{2}\cong\epsilon^{-1}$ and $\Phi(q)=\epsilon^{-\frac{1}{2}}\Phi(\frac{x}{\epsilon})$
. This scaling is called the weak-coupling limit since the density
of the particle is diverging in the limit but this is balanced by
the interaction that becomes weaker. This weak interaction between
particles is called also a ``grazing collision'' since it changes
only slightly the velocity of a particle. The kinetic equation derived
from this scaling is the Landau equation 
\begin{equation}
\partial_{t}f+v\cdot\nabla_{x}f==\intop dv_{1}\nabla_{v}\cdot\left[\frac{A}{\vert v-v_{1}\vert}P_{(v-v_{1})}^{\perp}\left(\nabla_{v}-\nabla_{v_{1}}\right)f(v)f(v_{1})\right]
\end{equation}
where $A$ is a suitable constant and $P_{(v-v_{1})}^{\perp}$ is
the projector on the orthogonal subspace to that generated by $v-v_{1}$. 

The Landau equation was derived in a formal way by Landau in \cite{Landau1936}
starting from the Boltzmann equation in the so-called grazing collision
limit. It rules the dynamics of a dense gas with weak interaction
between particles. Recently Boblylev, Pulvirenti and Saffirio proved
in \cite{Boblylev2013} a result of consistency, but the problem of
the rigorous derivation of Landau equation is still open even for
short times. 

Also in the case of the Landau equation it is possible to consider
the evolution of a perturbation of the stationary solution. This evolution
is given by the following linear Landau equation
\[
\partial_{t}g+v\cdot\nabla_{x}g=
\]
\begin{equation}
A\intop dv_{1}M_{\beta}(v_{1})\frac{1}{\vert V\vert^{3}}\left[\vert V\vert^{2}\triangle g(v)-\left(V,D^{2}(g)V\right)-4V\cdot\nabla_{v}g(v)\right]
\end{equation}
where $D^{2}(g)$ is the hessian matrix of $g$ with respect to the
velocity variables and $A$ is a suitable constant.

Recently Desvillettes and Ricci \cite{Desvillettes2001} and Kirkpatrick
\cite{Kirkpatrick2009}proved a rigorous derivation for a type of
linear Landau equation in two dimensions starting from a Lorentz gas.
In this case the velocity of the test particles does not change and
the equation obtained is a diffusion of the velocity on the unitary
sphere.

\subsection{Main theorem}

In this paper we prove the rigorous derivation of the linear Landau
equation starting from a system of particles. These particles interact
by means of a two body short-range smooth potential and we consider
an initial datum which is a perturbation of the equlibrium. We rescale
the variables describing the particles system according to (\ref{eq:scalingmic}).
Simultanously we set $N\epsilon^{2}\cong\alpha$ and $\Phi(q)=\frac{1}{\sqrt{\alpha}}\Phi(\frac{x}{\epsilon})$.
This gives us an intermediate scaling between the low density and
the weak-coupling and allows us to use the properties of both. Thanks
to the low density properties of the scaling as first step we prove
that the dynamics of the particles system is near to the solution
of the linear Boltzmann equation. In a second step using the weak-coupling
properties of the scaling we show that the solution of the linear
Landau equation is near to the solution of the linear Boltzmann equation.
More precisely let $\overline{f_{1}^{N}}$ be the one particle marginal
distribution and let $g^{\alpha}$ be the solution of the linear Boltzmann
equation, then we are able to prove that
\begin{equation}
\Vert\overline{f_{1}^{N}}(x,v)-g^{\alpha}(x,v)M_{\beta}(v)\Vert_{\infty}\rightarrow0.
\end{equation}
Then, denoting with $g$ the solution of the linear Landau equation,
it results that 
\begin{equation}
\Vert g(x,v)-g^{\alpha}(x,v)\Vert_{\mathbf{H}}\rightarrow0
\end{equation}
where $\mathbf{H}=L^{2}\left(\Gamma\times\mathbb{R}^{3},dxd\mu\right)$,
with $d\mu=M_{\beta}(v)dv$. 

\pagebreak{}

\section{Dynamics and statistical description of the motion}

\subsection{Hamiltonian system}

We consider a system of $N$ indistinguishable particles with unitary
mass moving in a torus $\Gamma_{\epsilon}=[0,\frac{1}{\epsilon})^{3}\subset\mathbb{R}^{3}$
with $\epsilon>0$. The particles interact by means of a two body
positive, radial and not increasing potential $\Phi:\,\mathbb{R}^{3}\rightarrow\mathbb{R}$.
We assume also that $\Phi$ is short-range, namely $\Phi(q)=0$ if
$\vert q\vert>1$, moreover $\Phi\in C^{2}(\mathbb{R}^{3})$ . The
Hamiltonian of the system is given by 
\begin{equation}
H=\frac{1}{2}\sum_{i=1}^{N}\vert v_{i}\vert^{2}+\frac{1}{2}\sum_{i,j=1,\,i\neq j}^{N}\Phi(q_{i}-q_{j})
\end{equation}
where $q_{i},v_{i}$ are respectively position and velocity of the
i-th particle. 

The Newton equations are the following
\begin{equation}
\frac{d^{2}q_{i}}{d\tau^{2}}(\tau)=\sum_{i\neq j}F\left(q_{i}(\tau)-q_{j}(\tau)\right)\label{eq:Newtoneq}
\end{equation}
for $i=1,...,N$ , where $F\left(q_{i}-q_{j}\right)=-\nabla\Phi\left(q_{i}-q_{j}\right)$
and $\tau$ is the time variable. The hypothesis that we made on the
potential ensure the existence and uniqueness of the solution of the
(\ref{eq:Newtoneq}).

\subsection{Scaling}

We rescale the system from microscopic coordinates $(q,\tau)$ to
macroscopic ones in the following way. We set 

\begin{equation}
x=\epsilon q\,\,\,\,t=\epsilon\tau
\end{equation}
where $x,t$ are respectively the macroscopic position variable and
the macroscopic time variable. We set $N\epsilon^{2}\cong\alpha$,
with $\alpha\cong\left(\log\log N\right)^{\frac{1}{2}}$, and we also
assume that $\vert N\epsilon^{2}-\alpha\vert\rightarrow0$. With this
scaling the density of the gas and the inverse of the mean free path
are diverging in the limit. This means that a given particle experiences
an high number of interaction per unit time. To balance this divergence
we rescale also the potential in the following way
\begin{equation}
\Phi\rightarrow\alpha^{-\frac{1}{2}}\Phi
\end{equation}
In the microscopic variables the equations of motion read as
\begin{equation}
\frac{d^{2}x_{i}}{dt^{2}}(\tau)=\frac{1}{\epsilon\sqrt{\alpha}}\sum_{i\neq j}-\nabla\Phi\left(\frac{x_{i}(t)-x_{j}(t)}{\epsilon}\right).\label{eq:Motioneq}
\end{equation}
From now we shall work in macroscopic variables unless explicitely
indicated.

\subsection{The scattering of two particles}

In this section we want to give a picture of the scattering between
two particles. We turn back to microscopic variables where the potential
is assumed to have range one. Let $q_{1},v_{1},q_{2},v_{2}$ be positions
and velocities of two particles which are performing a collision.
This two-body problem can be reduced to a central-force problem if
we set the origin of the coordinates $c$ in the center of mass
\begin{equation}
c=\frac{q_{1}+q_{2}}{2}
\end{equation}

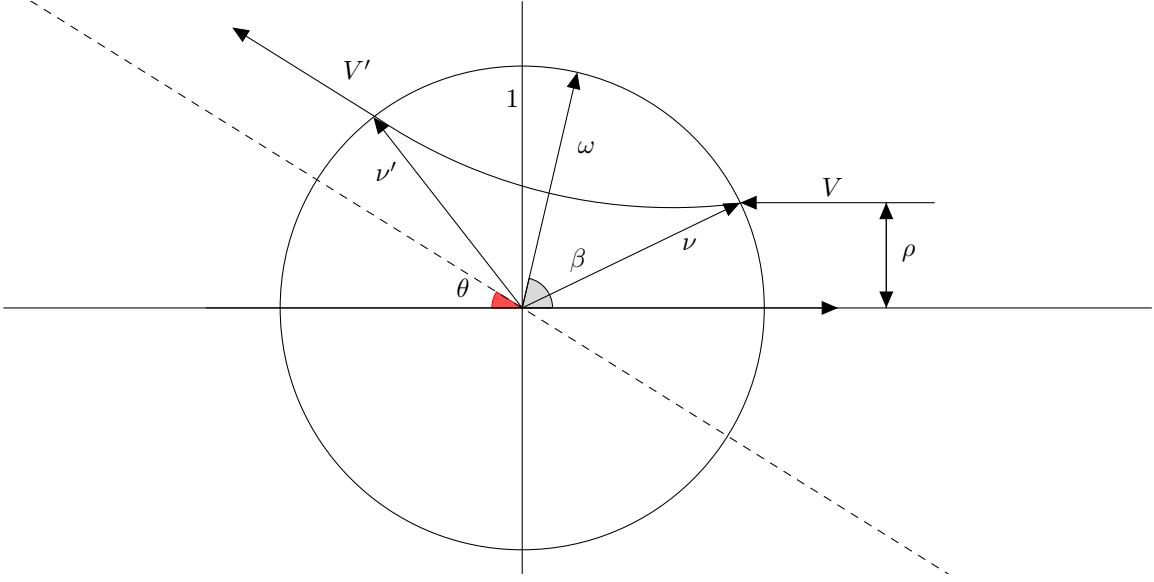
\begin{figure}[tb]
\definecolor{ffqqqq}{rgb}{1.,0.,0.}
\begin{tikzpicture}[line cap=round,line join=round,>=triangle 45,x=1.0cm,y=1.0cm]
\clip(-1.8539639814666136,-3.55917864118339) rectangle (13.405944342185185,4.057318811362377);
\draw [shift={(5.04,-0.02)},color=ffqqqq,fill=ffqqqq,fill opacity=0.7] (0,0) -- (147.99665839158172:0.40370127840348674) arc (147.99665839158172:180.:0.40370127840348674) -- cycle;
\draw [shift={(5.04,-0.02)},fill=black,fill opacity=0.15] (0,0) -- (0.:0.40370127840348674) arc (0.:76.86324293263583:0.40370127840348674) -- cycle;
\draw [->] (0.84,-0.02) -- (9.24,-0.02);
\draw (5.04,-3.55917864118339) -- (5.04,4.057318811362377);
\draw [domain=-1.8539639814666136:13.405944342185185] plot(\x,{(-0.168-0.*\x)/8.4});
\draw(5.04,-0.02) circle (3.217949657779002cm);
\draw [->] (10.52,1.38) -- (7.937447152235913,1.38);
\draw [->] (9.88,-0.02) -- (9.88,1.38);
\draw [->] (9.88,1.38) -- (9.88,-0.02);
\draw [->] (5.04,-0.02) -- (7.937447152235913,1.38);
\draw [->] (5.04,-0.02) -- (3.0616041978752038,2.5179420895944364);
\draw [->] (5.04,-0.02) -- (5.771363039709353,3.1137370828049833);
\draw (8.89794673334625,1.8369617801432045) node[anchor=north west] {$V$};
\draw (9.961026766475431,0.9488189676555355) node[anchor=north west] {$\rho$};
\draw (5.654879796838239,2.3214033142273878) node[anchor=north west] {$\omega$};
\draw (7.040920852690211,1.0295592233362327) node[anchor=north west] {$\nu$};
\draw (2.963537940814995,2.052269128625064) node[anchor=north west] {$\nu^\prime$};
\draw [->] (3.0762244565208556,2.529271585146075) -- (1.185454545454546,3.7109090909090883);
\draw [shift={(7.036153623666111,8.3712038403436)}] plot[domain=4.116686865126535:4.840999942122345,variable=\t]({1.*7.0575641324131775*cos(\t r)+0.*7.0575641324131775*sin(\t r)},{0.*7.0575641324131775*cos(\t r)+1.*7.0575641324131775*sin(\t r)});
\draw (2.5329232438512754,3.4113967659167996) node[anchor=north west] {$V^\prime$};
\draw (4.6994534379499875,3.0076954875133137) node[anchor=north west] {1};
\draw [dash pattern=on 3pt off 3pt,domain=-1.8539639814666136:13.405944342185185] plot(\x,{(-5.91763763082426--1.181637505763013*\x)/-1.8907699110663096});
\draw (4.040074683224293,0.4912908521315849) node[anchor=north west] {$\theta$};
\draw (5.547226122597309,0.8949921305350708) node[anchor=north west] {$\beta$};
\end{tikzpicture}

\caption{Here $\omega=\omega(\nu,V)$ is the unit vector bisecting the angle
between $-V$ and $V'$, $\nu$ is the unit vector pointing from the
particle with velocity $v_{1}$ to the particle with velocity $v_{2}$
when they are about to collide. We denote with $\beta$ the angle
between $-V$ and $\omega$, with $\varphi$ the angle between $-V$
and $\nu$, with $\rho=\sin\varphi$ the impact parameter and with
$\theta$ the deflection angle. It results that $\theta=\pi-2\beta$}
\end{figure}
Thanks to the conservation of the angular momentum we have that the
scattering takes place on a plane. We define $V=v_{1}-v_{2}$ as the
incoming relative velocity and $V^{'}=v_{1}^{'}-v_{2}^{'}$ as the
outgoing relative velocity with

\begin{equation}
\begin{cases}
v_{1}^{'}=v_{1}-\omega\left[\omega\cdot V\right]\\
v_{2}^{'}=v_{2}+\omega\left[\omega\cdot V\right]
\end{cases}
\end{equation}

Another useful way to represent the collision between two particles
is the so called $\sigma$-representation (Figure \ref{fig:sigma}).
With this notation the post collisional velocities can be written
as follow
\begin{equation}
\begin{cases}
v^{'}=\frac{v+v_{1}}{2}+\frac{\vert v-v_{1}\vert}{2}\sigma\\
v_{1}^{'}=\frac{v+v_{1}}{2}-\frac{\vert v-v_{1}\vert}{2}\sigma
\end{cases}\label{eq:sigmarapp}
\end{equation}

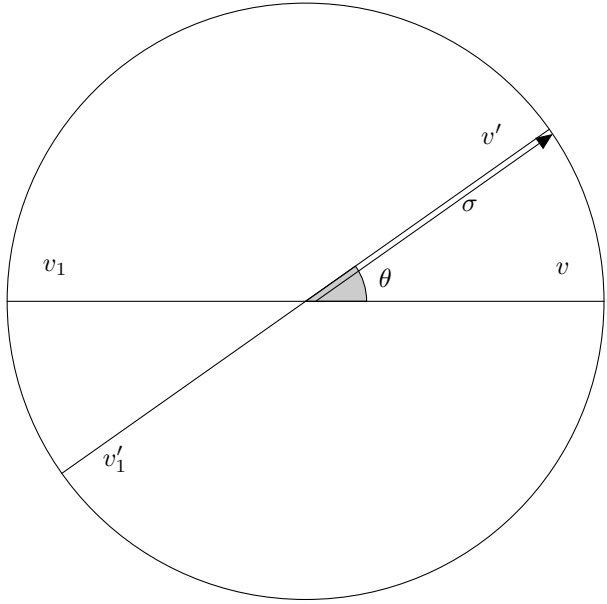
\begin{figure}[tb]
\begin{tikzpicture}[line cap=round,line join=round,>=triangle 45,x=0.8cm,y=0.8cm]
\clip(-3.881955492318414,-4.935403284178465) rectangle (19.130592082966732,6.550577251633976);
\draw [shift={(6.32,0.4)},fill=black,fill opacity=0.2] (0,0) -- (0.:1.0146625915028724) arc (0.:35.23996976045986:1.0146625915028724) -- cycle;
\draw(6.32,0.4) circle (3.9672257309107075cm);
\draw (1.3609678363616162,0.4)-- (11.279032163638384,0.4);
\draw (2.2697472858319774,-2.461372564238106)-- (10.370252714168023,3.261372564238106);
\draw (7.38079927336347,1.0916925093485579) node[anchor=north west] {$\theta$};
\draw (9.085432427088296,3.5065894771253787) node[anchor=north west] {$v^\prime$};
\draw (10.3233207887218,1.1931587684988445) node[anchor=north west] {$v$};
\draw (2.7945243597704867,-1.850829006009753) node[anchor=north west] {$v_1^\prime$};
\draw (1.8001550200976715,1.2540385239890166) node[anchor=north west] {$v_1$};
\draw (8.760740397807377,2.248407863661825) node[anchor=north west] {$\sigma$};
\draw [->] (6.488169014084508,0.4) -- (10.425466049105236,3.1815730297161364);
\end{tikzpicture}

\caption{\label{fig:sigma}We denote with $\sigma\in S^{2}\left(\frac{v_{1}+v_{2}}{2}\right)$
the direction of $V^{'}$ and with $\theta$ the angle between $V$
and $V^{'}$.}
\end{figure}

We can now define the scattering operator $I$, a map defined over
\begin{equation}
\left\{ \left(\nu,V\right)\in S^{2}\times\mathbb{R}^{3}\backslash\left\{ 0\right\} \,s.t.\,V\cdot\nu\leq0\right\} 
\end{equation}
by
\begin{equation}
I\left(\nu,V\right)=\left(\nu^{'},V^{'}\right)\label{eq:ChVar}
\end{equation}
\begin{equation}
\begin{cases}
V^{'}=V-2\omega\left(\omega\cdot V\right)\\
\nu^{'}=-\nu+2\omega\left(\omega\cdot\nu\right)
\end{cases}
\end{equation}
From the definition of $\nu^{'}$ and $V^{'}$ we have that $\nu\cdot V=-\nu^{'}\cdot V^{'}$.
It follows that $I$ sends incoming configuration in outgoing configuration.
The main property of $I$ is given by the following lemma, proved
in \cite{Pulvirenti2014}.
\begin{lem}
\label{lem:PresLeb}$I$ is an invertible transformation that preserves
the Lebesgue measure.
\end{lem}
We conclude this section with an estimate for the angle $\theta$,
for which a complete proof can be found in \cite{Desvillettes2001}
\begin{lem}
\label{lem:theta}Let $\Phi$ be a potential satisfying our assumption
and let $\theta(\rho,\alpha)$ be the scattering angle in function
of the impact parameter $\rho$. Then the following estimate holds
true:
\begin{equation}
\theta(\rho,\alpha)\leq\frac{-2}{\vert V\vert^{2}\sqrt{\alpha}}\gamma(\rho)+\frac{1}{\vert V\vert^{4}\alpha}M(\rho,\alpha)\label{eq:Thetaest}
\end{equation}
where 
\begin{equation}
\gamma(\rho)=\intop_{\vert\rho\vert}^{1}\frac{\rho}{u}\Phi^{'}\left(\frac{\vert\rho\vert}{u}\right)\frac{du}{\sqrt{1-u^{2}}}
\end{equation}
 and $M(\rho,\alpha)$ is positive bounded functions.
\end{lem}
\begin{rem}
Formula (\ref{eq:Thetaest}) points out that when $\alpha\rightarrow\infty$
the collision becomes grazing. 
\end{rem}

\subsection{Statistical description}

Now we want to describe our system from a statistical point of view.
We will denote the phase space as 
\begin{equation}
\Lambda_{N}=\left\{ \boldsymbol{z}_{N}\in\left(\Gamma\times\mathbb{R}^{3}\right)^{N}\right\} 
\end{equation}
where $\boldsymbol{z}_{N}=\left(z_{1},z_{2},..,z_{N}\right)$, $z_{i}=(x_{i},v_{i})$
and $\Gamma$ is the torus of unitary side.

We consider a probability density function $W_{0,N}$ defined on $\Lambda_{N}$.
The time evolution of $W_{0,N}$ is given by the solution $W_{N}$
of the following Liouville equation
\begin{equation}
\begin{cases}
\partial_{t}W_{N}+\mathcal{L}_{N}W_{N}=0\\
W_{N}(0)=W_{0,N}
\end{cases}\label{eq:Liouvilleeq}
\end{equation}
where $\mathcal{L}_{N}=\mathcal{L}_{N}^{0}+\mathcal{L}_{N}^{I}$ with
\begin{equation}
\mathcal{L}_{N}^{0}=\sum_{i=1}^{N}v_{i}\cdot\nabla_{x_{i}}
\end{equation}
\begin{equation}
\mathcal{L}_{N}^{I}=\frac{1}{\epsilon}\underset{i\neq j}{\sum_{i,j=1}^{N}}F_{i,j}\cdot\nabla_{v_{i}}
\end{equation}
and $F_{i,j}=-\frac{1}{\sqrt{\alpha}}\nabla\Phi\left(\frac{x_{i}(t)-x_{j}(t)}{\epsilon}\right)$.
We suppose that $W_{0,N}$ is symmetric in the exchange of particles,
and hence $W_{N}(t)$ is still symmetric for any positive times. 

The marginals distribution of the measure $W_{N}(t)$ are defined
as 
\begin{equation}
\overline{f_{j}^{N}}(\boldsymbol{z}_{j},t)=\intop dz_{j+1}...dz_{N}\,W_{N}(\boldsymbol{z}_{N},t)
\end{equation}
Nevertheless, it is more convenient to work with the reduced marginals
$\widetilde{f_{j}^{N}}(\boldsymbol{z}_{j},t)$ that read as follow
\begin{equation}
\widetilde{f_{j}^{N}}(\boldsymbol{z}_{j},t)=\intop_{S(\boldsymbol{x}_{j})^{N-j}}dz_{j+1}...dz_{N}\,W_{N}(\boldsymbol{z}_{N},t)
\end{equation}
where 
\begin{equation}
S(\boldsymbol{x}_{j})^{N-j}=\left\{ z=(x,v)\in\Gamma\times\mathbb{R}^{3}\,|\,\vert x-x_{k}\vert>\epsilon\,\,\forall\,1\leq k\leq j\right\} 
\end{equation}
As can be easily seen the reduced marginals are asymptotically equivalent
(for $\epsilon\rightarrow0$) to the standard marginals.

For the reduced marginals it is possible to derive from the Liouville
equation the following hierarchy of equations, called the Grad hierarchy
(GH),
\begin{equation}
\left(\partial_{t}+\mathcal{L}_{j}\right)f_{j}^{N}=\sum_{m=0}^{N-j-1}A_{j+1+m}^{\epsilon}f_{j+1+m}^{N}\,\,\,0\leq j\leq N
\end{equation}
where 

\[
A_{j+1+m}^{\epsilon}f_{j+1+m}^{N}(z_{j},t)=\binom{N-j-1}{m}(N-j)\epsilon^{2}\sum_{i=1}^{j}\intop_{S^{2}}d\nu\chi_{\left\{ min_{l=1,...,j;l\neq i}\vert x_{i}+\nu\epsilon-x_{l}>\epsilon\vert\right\} }\left(\nu\right)
\]
\begin{equation}
\intop_{\mathbb{R}^{3}}dv_{j+1}\left(v_{j+1}-v_{j}\right)\cdot\nu\intop_{\Delta_{m}(\boldsymbol{x}_{j+1})}dz_{j+1}...dz_{j+1+m}f_{j+1+m}^{N}(\boldsymbol{z}_{j},x_{i}+\nu\epsilon,v_{j+1},\boldsymbol{z}_{j+1,m},t)
\end{equation}
and $\boldsymbol{z}_{j+1,m}=\left(z_{j+1},...,z_{j+1+m}\right)$.
The set $\Delta_{m}(x_{j+1})$ is defined as follows
\begin{alignat}{1}
\Delta_{m}(\boldsymbol{x}_{j+1})= & \left\{ \boldsymbol{z}_{j+1,m}\subset\right.S(x_{j})^{m}\,\text{such that}\,\,\,\forall\,l=j+2,...,j+1+m,\,\,\text{there exists}\nonumber \\
 & \text{a choice of index}\,\,\,h_{1},...,h_{r}\in\{j+2,...,j+1+m\}\,\nonumber \\
 & \text{such that}\,\vert x_{l}-x_{h_{1}}\vert\leq\epsilon,\,\vert x_{h_{k}-1}-x_{h_{k}}\vert\leq\epsilon\,\,\,\text{for}\,k=2,...,r\nonumber \\
 & \left.\text{and}\,\min_{i\in\left\{ l,h_{1},...,h_{r}\right\} }\vert x_{i}-x_{j+1}\vert\leq\epsilon\right\} 
\end{alignat}
This hierarchy was first introduced by Grad \cite{Grad1958}. Actually
in views of the Boltzmann-Grad limit only the first equation of this
hierarchy was considered. The full hierarchy was introduced and derived
by King in \cite{King1975}. A complete derivation of this hierarchy
can also be found in \cite{Gallagher} adn \cite{Pulvirenti2014}. 

It is possible to represent the solution of the Grad hierarchy as
a series obtained by iterating the Duhamel formula. It results that

\begin{equation}
\widetilde{f_{j}^{N}}(t)=\sum_{n=0}^{\infty}G_{j,n}^{\epsilon}(t)f_{0,j}^{N}\label{eq:GSS}
\end{equation}
where
\begin{equation}
f_{0,j}^{N}=\intop_{S(\boldsymbol{x}_{j})^{N-j}}dz_{j+1}...dz_{N}\,W_{0,N}(\boldsymbol{z}_{N})
\end{equation}
and $G_{j,n}^{\epsilon}(t)$ is defined for $n\leq N-j$ as
\[
G_{j,n}^{\epsilon}(t)=\underset{j+n+\sum_{i=1}^{n}m_{i}\leq N}{\sum_{m_{1},...,m_{n}\geq0}}\intop_{0}^{t}dt_{1}...\intop_{0}^{t_{n-1}}dt_{n}
\]
\begin{equation}
S_{j}^{\epsilon}(t-t_{1})A_{j+1+m_{1}}^{\epsilon}S_{j+1+m_{1}}^{\epsilon}(t_{1}-t_{2})...A_{j+n+\sum_{i=1}^{n}m_{i}}^{\epsilon}S_{j+n+\sum_{i=1}^{n}m_{i}}^{\epsilon}(t_{n})f_{0,j+n+\sum_{i=1}^{n}m_{i}}^{N}\label{eq:GSSOp}
\end{equation}
 and it is identically equal to zero for $n>N-j$. The operator $S_{j}^{\epsilon}(t)$
is the interacting flow operator: 
\begin{equation}
S_{j}^{\epsilon}(t)g(\boldsymbol{z}_{j})=g\left(T_{j}^{\epsilon}(-t)\boldsymbol{z}_{j}\right),
\end{equation}
where $T_{j}^{\epsilon}(t)$ is the solution of the Newton equation
(\ref{eq:Motioneq}). We call this series the Grad series solution
(GSS).

Next we introduce the following hierarchy of equations, called the
intermediate hierarchy (IH)
\begin{equation}
\left(\partial_{t}+\mathcal{L}_{j}\right)f_{j}^{N}=(N-j)\epsilon^{2}C_{j+1}^{\epsilon}(f_{j+1}^{N})
\end{equation}
\[
C_{j+1}^{\epsilon}(f_{j+1}^{N})=\sum_{k=1}^{j}\intop_{\mathbb{R}^{3}}dv_{j+1}\intop_{\nu\cdot(v_{k}-v_{j+1})\geq0}d\nu\vert\nu\cdot(v_{k}-v_{j+1})\vert
\]
\begin{equation}
\left[f_{j+1}^{N}(x_{1},v_{1},...,x_{k},v_{k}^{'},...,x_{j},v_{j},x_{k}-\eta\epsilon,v_{j+1}^{'})-f_{j+1}^{N}(x_{1},v_{1},...,x_{j},v_{j},x_{k}+\eta\epsilon,v_{j+1})\right]\label{eq:2.28}
\end{equation}
This hierarchy is formally similar to the BBGKY hierarchy for hard
spheres but the collision operator appearing in IH is different. Indeed,
in the IH we have that the trasfered momentum is 
\begin{equation}
p=\left(\mathrm{V}\cdot\omega\right)\omega
\end{equation}
while in hard spheres it is 
\begin{equation}
p=\left(\mathrm{V}\cdot\nu\right)\nu.
\end{equation}
Note that it may be convenient to express $\nu$ in terms of $\omega$,
which is the parameter appearing in the expression of the outgoing
velocities. However, as described in \cite{Pulvirenti2014}, this
is a delicate point and we prefer to avoid it, working as much as
possible with formula (\ref{eq:2.28}). We want to notice also that
$A_{j+1}^{\epsilon}f_{j+1}^{N}=C_{j+1}^{\epsilon}(f_{j+1}^{N})$,
i.e. the first term in the sum on the right hand side of equation
(\ref{eq:GSS}) is the collision term that arise in the IH case. As
we will see this will be the only $O(1)$ term as $\epsilon\rightarrow0$.

Also for IH we can write the following formal series for the solution,
that we will call intermediate series solution (ISS)

\begin{equation}
f_{j}^{N}(t)=\sum_{n=0}^{\infty}Q_{j,n}^{\epsilon}(t)f_{0,j}^{N}\label{eq:ISS}
\end{equation}
where the operator $Q_{j,n}^{\epsilon}(t)$ is defined for $n\leq N-j$
as 
\begin{alignat}{1}
Q_{j,n}^{\epsilon}(t)= & \left(N-j\right)...\left(N-j-n+1\right)\epsilon^{2n}\nonumber \\
 & \intop_{0}^{t}dt_{1}...\intop_{0}^{t_{n-1}}dt_{n}S_{j}^{\epsilon}(t-t_{1})C_{j+1}^{\epsilon}S_{j+1+m_{1}}^{\epsilon}(t_{1}-t_{2})...C_{j+n}^{\epsilon}S_{j+n}^{\epsilon}(t_{n})f_{0,j+n}^{N}\label{eq:ISSOp}
\end{alignat}
and it is identically equal to zero for $n>N-j$.

Finally we observe that by sending $\epsilon\rightarrow0$, $N\rightarrow\infty$,
$N\epsilon^{2}\rightarrow\alpha$ in the IH we obtain, formally, the
following hierarchy, called the Boltzmann hierarchy (BH) 

\begin{equation}
\left(\partial_{t}+v\cdot\nabla_{x}\right)f_{j}=\alpha C_{j+1}(f_{j})\,\,\,0\leq j
\end{equation}
\[
C_{j+1}(f^{j})=\sum_{k=1}^{j}\intop_{\mathbb{R}^{3}}dv_{j+1}\intop_{\nu\cdot(v_{k}-v_{j+1})\geq0}d\nu\vert\nu\cdot(v_{k}-v_{j+1})\vert
\]
\begin{equation}
\left[f_{j+1}(x_{1},v_{1},...,x_{k},v_{k}^{'},...,x_{j},v_{j},x_{k},v_{j+1})-f_{j+1}(x_{1},v_{1},...,x_{j},v_{j},x_{k},v_{j+1})\right]\label{eq:BH}
\end{equation}
If we assume the propagation of chaos, i.e. that $f_{j}=f_{1}^{\otimes j}$,
the first equation of this infinite hierarchy becomes the Boltzmann
equation.

The series solution for the Boltzmann hierarchy (BSS) is the following
\begin{equation}
f_{j}^{\alpha}(t)=\sum_{n=0}^{\infty}Q_{j,n}^{\alpha}(t)f_{0,j+n}\label{eq:BSS}
\end{equation}
where $f_{0,j+n}$ is the $j+n$ particles initial datum and $Q_{j,n}^{\alpha}(t)$
is defined as follows

\begin{equation}
Q_{j,n}^{\alpha}(t)=\alpha^{n}\intop_{0}^{t}dt_{1}...\intop_{0}^{t_{n-1}}dt_{n}S_{j}(t-t_{1})C_{j+1}S_{j+1+m_{1}}(t_{1}-t_{2})...C_{j+n}S_{j+n}(t_{n})f_{0,j+n}\label{eq:BSSOp}
\end{equation}
where $S_{j}(t)$ is the free flow operator, i.e. 
\begin{equation}
S_{j}(t)g^{j}(\boldsymbol{z}_{j})=g(\boldsymbol{x}_{j}-\boldsymbol{v}_{j}t).
\end{equation}

\pagebreak{}

\section{Linear regime}

In this section we formally derive the linear Boltzmann and Landau
equations. First we define the Gibbs measure defined by
\begin{equation}
M_{N,\beta}(\boldsymbol{z}_{n})=C_{N,\beta}e^{-\beta H_{N}(\boldsymbol{z}_{n})}\label{eq:GIbbs}
\end{equation}
where $\beta>0$ and $C_{N,\beta}$ is chosen so that 
\begin{equation}
\intop_{\Lambda_{N}}M_{N,\beta}(\boldsymbol{z}_{n})d\boldsymbol{z}_{n}=1\label{eq:cbet1}
\end{equation}
The Gibbs measure is an invariant measure for the gas dynamics and
(\ref{eq:GIbbs}) is a stationary solution of the Liouville equation. 

In case of the Boltzmann and Landau equations, a stationary solution
is given by the Maxwellian distribution (free gas)
\begin{equation}
M_{\beta}(v)=C_{\beta}e^{-\frac{\beta}{2}\vert v\vert^{2}},\label{eq:cbet2}
\end{equation}
where $\beta>0$ and $C_{\beta}$ is such that 
\begin{equation}
\intop_{\Gamma\times\mathbb{R}^{3}}M_{\beta}(v)dxdv=1.
\end{equation}
Moreover a stationary solution of the Boltzmann hierarchy is 
\begin{equation}
M_{\beta}^{\otimes j}(\boldsymbol{v}_{j})=\prod_{i=1}^{j}M_{\beta}(v_{i}).
\end{equation}

Now we consider the Liouville equation (\ref{eq:Liouvilleeq}) with
initial datum given by
\begin{equation}
W_{0,N}(\boldsymbol{z}_{N})=M_{N,\beta}(\boldsymbol{z}_{N})g_{0}(x_{1},v_{1})\label{eq:Initialpertu}
\end{equation}
where $g_{0}\in L^{\infty}\left(\Gamma\times\mathbb{R}^{3}\right)$
is a perturbation on the first particle such that $\intop d\boldsymbol{z}_{1}M_{N,\beta}(\boldsymbol{z}_{1})g_{0}(x_{1},v_{1})=1$. 
\begin{thm}
Let $W^{N}$ be the solution of the Liouville equation (\ref{eq:Liouvilleeq})
with initial datum (\ref{eq:Initialpertu}) and let $f_{j}^{N}$ be
the j-particles reduced marginal. Then for any $1\leq j\leq N$ the
following bound holds
\begin{equation}
\sup_{t}f_{j}^{N}(\boldsymbol{z}_{j},t)\leq M_{N,\beta}(\boldsymbol{z}_{j})\Vert g_{0}\Vert_{\infty}\leq M_{\beta}^{\otimes j}(\boldsymbol{z}_{j})\Vert g_{0}\Vert_{\infty}\label{eq:Apriori}
\end{equation}
\end{thm}
\begin{proof}
From the choice of the initial datum we have that
\begin{equation}
f_{0}^{N}(\boldsymbol{z}_{N})\leq M_{N,\beta}(\boldsymbol{z}_{N})\Vert g_{0}\Vert_{\infty}
\end{equation}
Since the maximum principle holds for the Liouville equation and $M_{N,\beta}(\boldsymbol{z}_{N})$
is a stationary solution we have that
\begin{equation}
W^{N}(\boldsymbol{z}_{N},t)\leq M_{N,\beta}(\boldsymbol{z}_{N})\Vert g_{0}\Vert_{\infty}
\end{equation}
This implies the (\ref{eq:Apriori}) since $M_{N,\beta}(\boldsymbol{z}_{j})\leq M_{\beta}^{\otimes j}(\boldsymbol{z}_{j})$
by the positivity of the interaction. 
\end{proof}

\subsection{Linear Boltzmann equation and asymptotics}

In this section we derive the linear Boltzmann equation from the non
linear one and study its asymptotic behavior for $\alpha\rightarrow\infty$.
Suppose that the initial datum of the Boltzmann hierarchy (\ref{eq:BH})
is 
\begin{equation}
f_{0,j}(x_{1},v_{1},...,x_{j},v_{j})=M_{\beta}(v_{1})...M_{\beta}(v_{j})g_{0}(x_{1},v_{1})\label{eq:initialBH}
\end{equation}
with $g_{0}(x_{1},v_{1})\in L^{\infty}\left(\Gamma\right)$. Since
the Maxwellian distribution is a stationary solution of the equations
we look for a solution at time $t$ given by
\begin{equation}
f_{j}^{\alpha}(\boldsymbol{z}_{j},t)=M_{\beta}(v_{1})...M_{\beta}(v_{j})g^{\alpha}(x_{1},v_{1},t).\label{eq:BHsol}
\end{equation}
From (\ref{eq:initialBH}) and (\ref{eq:BH}) we have that (\ref{eq:BHsol})
is a solution of the Boltzmann hierarchy if $g^{\alpha}$ satisfies
the following equation
\begin{equation}
M_{\beta}(v)\left(\partial_{t}g^{\alpha}+v\cdot\nabla_{x}g^{\alpha}\right)=\alpha\intop dv_{1}\intop_{\nu\cdot\mathrm{V}>0}d\nu\vert\nu\cdot V\vert\left[M_{\beta}(v^{'})M_{\beta}(v_{1}^{'})g^{\alpha}(x,v^{'})-M_{\beta}(v)M_{\beta}(v_{1})g^{\alpha}(x,v)\right]\label{eq:lin}
\end{equation}
Since $M_{\beta}(v^{'})M_{\beta}(v_{1}^{'})=M_{\beta}(v)M_{\beta}(v_{1})$
the equation (\ref{eq:lin}) becomes the Linear Boltzmann equation
\begin{equation}
\partial_{t}g^{\alpha}+v\cdot\nabla_{x}g^{\alpha}=Q_{B}(g^{\alpha}),
\end{equation}
where
\begin{equation}
Q_{B}(g^{\alpha})=\alpha\intop dv_{1}M_{\beta}(v_{1})\intop_{\nu\cdot V>0}d\nu\vert\nu\cdot V\vert\left[g^{\alpha}(x,v^{'})-g^{\alpha}(x,v)\right].\label{eq:LinCollOp}
\end{equation}

We are interested to investigate the behavior of $Q_{B}$ when $\alpha\rightarrow\infty$.
We denote with $\left(\hat{e}_{1},\hat{e}_{2},\hat{e}_{3}\right)$
an orthonormal base of $\mathbb{R}^{3}$ such that $\hat{e}_{1}=\frac{\mathrm{V}}{\vert\mathrm{V}\vert}$.
Now we consider the semispehere $S_{+}^{2}=\left\{ \nu\in s^{2}\,|\,\nu\cdot\mathrm{V}>0\right\} $.
For a fixed $\nu$ in this semispehere the scattering takes place
in the plane generated by $\hat{e}_{1}$ and $\nu$. An orthonormal
base of the scattering plane is given by the vectors $\hat{e}_{1}$
and $\hat{e}\left(\psi\right)=\hat{e}_{2}\cos\psi+\hat{e}_{3}\sin\psi$,
calling with $\psi$ the angle between $\hat{e}_{2}$ and $\hat{e}$
. We also denote with $\varphi$ the angle between $\hat{e}_{1}$
and $\nu$. 

\begin{figure}[h]
\includegraphics[scale=0.2]{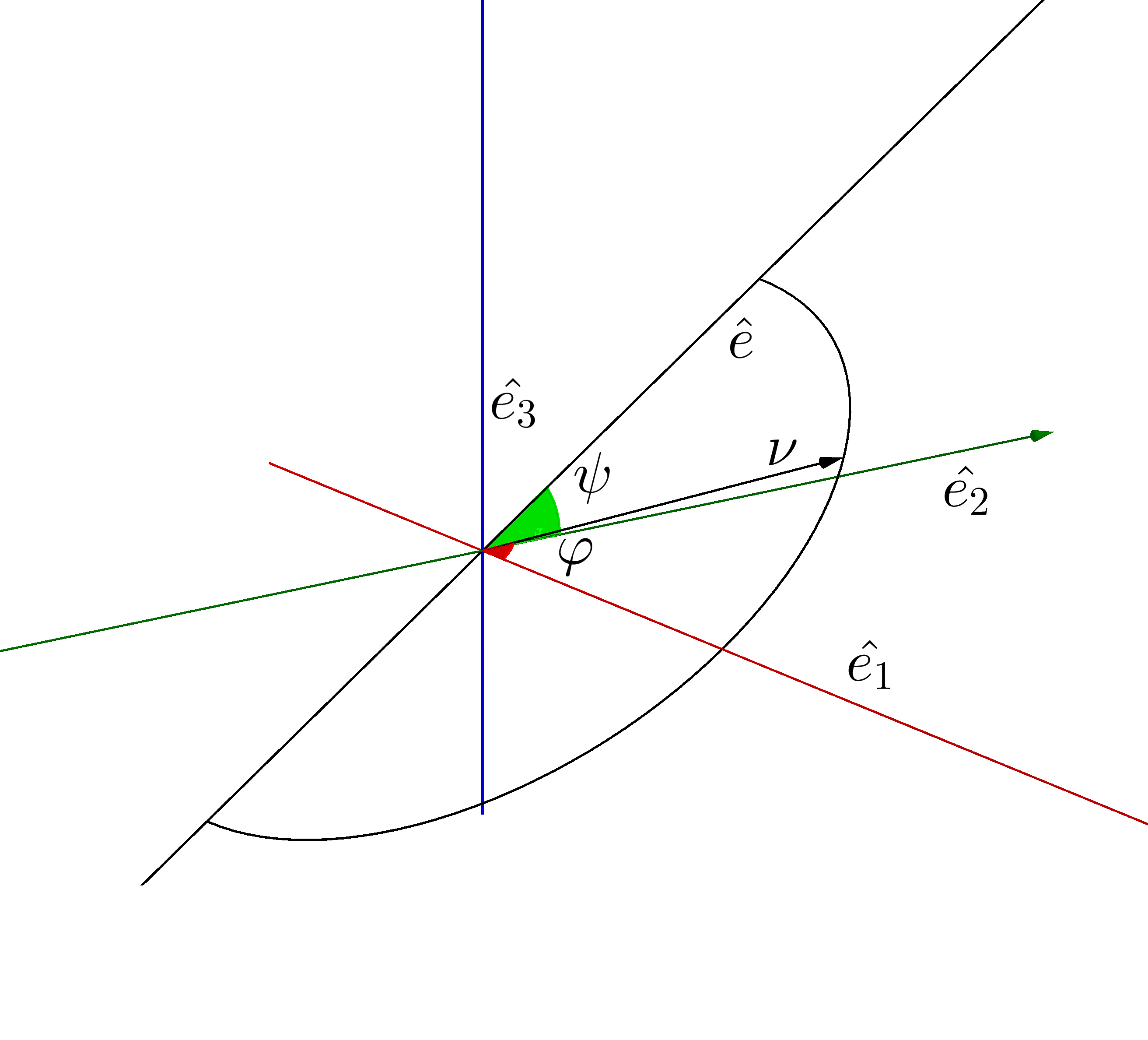}\caption{A representation of a three dimensional scattering.}
\end{figure}

From the $\sigma-representation$ (\ref{eq:sigmarapp}) we have that
\begin{equation}
v^{'}=c+r\sigma
\end{equation}
where $r=\frac{\vert\mathrm{V}\vert}{2}$ and $c=\frac{v+v_{1}}{2}$.
Notice that in our coordinates it results that
\begin{equation}
\sigma=\cos\theta\hat{e_{1}}-\sin\theta\hat{e}(\psi)
\end{equation}

\begin{figure}[tb]
\definecolor{ffqqtt}{rgb}{1.,0.,0.2}
\definecolor{ffqqqq}{rgb}{1.,0.,0.}
\begin{tikzpicture}[line cap=round,line join=round,>=triangle 45,x=1.0cm,y=1.0cm]
\clip(-1.6033333333333386,-3.133333333333341) rectangle (13.516666666666698,4.413333333333337);
\draw [shift={(5.5,-0.06)},color=ffqqtt,fill=ffqqtt,fill opacity=1.0] (0,0) -- (-0.4375436499996177:0.6666666666666683) arc (-0.4375436499996177:33.43172347935919:0.6666666666666683) -- cycle;
\draw [shift={(5.5,-0.06)},fill=black,fill opacity=0.1] (0,0) -- (128.5550770047361:0.4) arc (128.5550770047361:180.:0.4) -- cycle;
\draw(5.5,0.14) circle (2.673275144836386cm);
\draw (5.636666666666679,3.3866666666666685) node[anchor=north west] {$\hat{e}$};
\draw (8.316666666666686,-0.013333333333336386) node[anchor=north west] {$\hat{e_1}$};
\draw [->,color=ffqqqq] (-0.36,-0.06) -- (11.4,-0.06);
\draw [->] (5.5,-0.06) -- (7.818567339480824,1.4706560383107328);
\draw [->] (5.5,-0.06) -- (3.7381690391141147,2.150560037716834);
\draw (7.156666666666682,1.8) node[anchor=north west] {$\nu$};
\draw (4.583333333333343,0.6666666666666646) node[anchor=north west] {$\theta$};
\draw (4.036666666666675,2.373333333333334) node[anchor=north west] {$\sigma$};
\draw [->] (5.5,-3.133333333333341) -- (5.5,4.226666666666671);
\end{tikzpicture}

\caption{}
\end{figure}

We denote with $v^{'}(\theta)$ the post collisional velocity in function
of the scattering angle $\theta$
\begin{equation}
v^{'}(\theta)=c+r\cos\theta\hat{e}_{1}-r\sin\theta\hat{e}(\psi)
\end{equation}
This implies that 
\begin{equation}
g^{\alpha}(v^{'}(\theta))=g^{\alpha}(c+r\cos\theta\hat{e}_{1}-r\sin\theta\hat{e}(\psi))\label{eq:coord}
\end{equation}
For sake of brevity we will not take care of the dependence of $g$
from the spatial variable. Let us consider the Taylor expansion of
g with respect to $\theta$ up to the second order. We have

\begin{alignat}{1}
g^{\alpha}(v^{'})-g^{\alpha}(v)= & g^{\alpha}(v^{'}(\theta))-g^{\alpha}(v^{'}(0))\nonumber \\
= & \theta\nabla_{v}g^{\alpha}(v)\cdot\frac{dv^{'}}{d\theta}(0)+\frac{\theta^{2}}{2}\left[\nabla_{v}g^{\alpha}\cdot\frac{d^{2}v^{'}}{d\theta^{2}}(0)+\left(\frac{dv^{'}}{d\theta}(0),D_{v}^{2}(g^{\alpha})\frac{dv^{'}}{d\theta}(0)\right)\right]+o(\theta^{2})\label{eq:tayexp}
\end{alignat}
where $D_{v}^{2}(g^{\alpha})$ is the hessian matrix of $g^{\alpha}$
with respect to the velocity. A simple calculation gives us that
\begin{equation}
\frac{dv^{'}}{d\theta}(0)=-r\hat{e}(\psi)
\end{equation}
\begin{equation}
\frac{d^{2}v^{'}}{d\theta^{2}}(0)=-r\hat{e}_{1}
\end{equation}
It can be easily seen that the integration of the first term is zero
by symmetry. Moreover from Lemma \ref{lem:theta} we have that 
\begin{equation}
\theta^{2}(\rho,\alpha)\leq\frac{4}{\vert\mathrm{V}\vert^{4}\alpha}\gamma^{2}(\rho)+o(\alpha^{-1})
\end{equation}
From this remark and by equations (\ref{eq:tayexp}) and (\ref{eq:LinCollOp})
we have that 
\begin{alignat}{1}
Q_{B}= & \intop dv_{1}M_{\beta}(v_{1})\intop_{\nu\cdot V>0}d\nu\vert\nu\cdot V\vert\frac{2}{\vert V\vert^{4}}\gamma(\rho)^{2}\nonumber \\
 & \left[-\frac{1}{2}V\cdot\nabla_{v}g^{\alpha}(v)+\frac{\vert V\vert^{2}}{4}\left(\hat{e}(\psi),D^{2}\hat{e}(\psi)\right)\right]+o(\alpha^{-1})
\end{alignat}
From the change of variables $\nu\rightarrow\psi,\varphi$ , since
$d\nu=\sin\varphi d\varphi d\psi$, we have that

\begin{alignat}{1}
Q_{B}= & \intop dv_{1}M_{\beta}(v_{1})\intop_{-\frac{\pi}{2}}^{\frac{\pi}{2}}d\varphi\intop_{-\frac{\pi}{2}}^{\frac{\pi}{2}}d\psi\vert\nu\cdot V\vert\sin\varphi\frac{2}{\vert\mathrm{V}\vert^{4}}\gamma(\rho)^{2}\nonumber \\
 & \left[-\frac{1}{2}V\cdot\nabla_{v}g^{\alpha}(v)+\frac{\vert V\vert^{2}}{4}\left(\hat{e}(\psi),D^{2}\hat{e}(\psi)\right)\right]+o(\alpha^{-1})
\end{alignat}
Since $\vert\nu\cdot V\vert=\vert V\vert\cos\varphi$ and $\rho=\sin\varphi$,
it results that $\cos\varphi d\varphi=d\rho$ and so 
\[
Q_{B}=\intop dv_{1}M_{\beta}(v_{1})\intop_{-1}^{1}d\rho\frac{1}{\vert V\vert^{3}}\mbox{\ensuremath{\rho}}\gamma(\rho)^{2}\intop_{-\frac{\pi}{2}}^{\frac{\pi}{2}}d\psi\left[-V\cdot\nabla_{v}g^{\alpha}(v)+\frac{\vert V\vert^{2}}{2}\left(\hat{e}(\psi),D^{2}\hat{e}(\psi)\right)\right]+o(\alpha^{-1})=
\]
\begin{equation}
\intop dv_{1}M_{\beta}(v_{1})\frac{1}{\vert V\vert^{3}}\intop_{-\frac{\pi}{2}}^{\frac{\pi}{2}}d\psi\left[\frac{\vert V\vert^{2}}{2}\left(\hat{e}(\psi),D^{2}\hat{e}(\psi)\right)-V\cdot\nabla_{v}g^{\alpha}(v)\right]\intop_{-1}^{1}d\rho\rho\gamma(\rho)^{2}+o(\alpha^{-1})
\end{equation}
From the definition of $\hat{e}(\psi)$, and since $\intop_{-\frac{\pi}{2}}^{\frac{\pi}{2}}\sin^{2}\psi d\psi=\intop_{-\frac{\pi}{2}}^{\frac{\pi}{2}}\cos^{2}\psi d\psi=\frac{\pi}{2}$
and $\intop_{-\frac{\pi}{2}}^{\frac{\pi}{2}}\sin\psi\cos\psi d\psi=0$,
we have that
\begin{equation}
Q_{B}=\intop dv_{1}M_{\beta}(v_{1})\frac{1}{\vert V\vert^{3}}\left[\vert V\vert^{2}\left(\hat{e}_{2},D^{2}\hat{e}_{2}\right)+\vert V\vert^{2}\left(\hat{e}_{3},D^{2}\hat{e}_{3}\right)-4V\cdot\nabla_{v}g^{\alpha}(v)\right]\frac{\pi}{4}\intop_{-1}^{1}d\rho\rho\gamma(\rho)^{2}+o(\alpha^{-1})
\end{equation}
Now since the laplacian is the trace of the Hessian matrix and it
is invariant under changes of coordinates we have that 
\begin{equation}
\triangle g(v)=\left(\hat{e}_{1},D^{2}\hat{e}_{1}\right)+\left(\hat{e}_{2},D^{2}\hat{e}_{2}\right)+\left(\hat{e}_{3},D^{2}\hat{e}_{3}\right)
\end{equation}
and so 
\begin{equation}
\vert V\vert^{2}\left(\hat{e}_{2},D^{2}\hat{e}_{2}\right)+\vert V\vert^{2}\left(\hat{e}_{3},D^{2}\hat{e}_{3}\right)=\vert V\vert^{2}\triangle g^{\alpha}(v)-\left(V,D^{2}V\right)\label{eq:laptrace}
\end{equation}
Thanks to (\ref{eq:laptrace}) we finally arrive to 
\begin{alignat}{1}
Q_{B}(g)= & B\intop dv_{1}M_{\beta}(v_{1})\frac{1}{\vert V\vert^{3}}\left[\vert V\vert^{2}\triangle g^{\alpha}(v)-\left(V,D^{2}V\right)-4V\cdot\nabla_{v}g^{\alpha}(v)\right]+o(\alpha)\nonumber \\
= & Q_{L}(g)+o(\alpha)\label{eq:AsymLinBolt}
\end{alignat}
where 
\begin{equation}
B=\frac{\pi}{4}\intop_{-1}^{1}d\rho\rho\gamma(\rho)^{2}
\end{equation}

\subsection{Linear Landau equation}

In this subsection we will show that the linear operator $Q_{L}(g)$
is indeed the linear Landau operator obtained by the full nonlinear
equatios.. Consider
\[
\begin{cases}
\partial_{t}f+v\cdot\nabla_{x}f=C_{L}(f)\\
f(x,v,0)=f_{0}(x,v)
\end{cases}
\]
whit 
\begin{equation}
C_{L}(f)=A\intop dv_{1}\nabla_{v}\cdot\left[\frac{1}{\vert v-v_{1}\vert}P_{(v-v_{1})}^{\perp}\left(\nabla_{v}-\nabla_{v_{1}}\right)f(v)f(v_{1})\right]\label{eq:LandOp}
\end{equation}
where $A>0$ is a suitable constant and $P_{(v-v_{1})}^{\perp}$ is
the projector on the orthogonal subspace to $v-v_{1}$. 

Also in this case we consider a perturbation of the stationary state.
We set $f(v)=M_{\beta}(v)g(v)$ and $f(v_{1})=M_{\beta}(v_{1})$ in
(\ref{eq:LandOp}). This represents a single particle perturbed in
a stationary background. With this choice equation (\ref{eq:LandOp})
becomes
\[
M_{\beta}(v)\left(\partial_{t}g+v\cdot\nabla_{x}g\right)=K(g)
\]
\[
K(g)=A\intop dv_{1}\nabla_{v}\cdot\left[\frac{1}{\vert v-v_{1}\vert}P_{(v-v_{1})}^{\perp}\left(\nabla_{v}-\nabla_{v_{1}}\right)M_{\beta}(v)M_{\beta}(v_{1})g(v)\right]
\]
We suppose to have all the necessary regularity to give sense to the
following calculations. We start from the gradient term which leads
to 

\[
K(g)=A\intop dv_{1}\nabla_{v}\cdot\left[\frac{1}{\vert\mathrm{V}\vert}P_{\mathrm{V}}^{\perp}\left(M_{\beta}(v)M_{\beta}(v_{1})\nabla_{v}g(v)-2v\beta M_{\beta}(v)M_{\beta}(v_{1})h(v)+2v_{1}\beta M_{\beta}(v)M_{\beta}(v_{1})g(v)\right)\right]
\]
\begin{equation}
=A\intop dv_{1}M_{\beta}(v_{1})\nabla_{v}\cdot\left[\frac{1}{\vert\mathrm{V}\vert}P_{\mathrm{V}}^{\perp}\left(M_{\beta}(v)\nabla_{v}g(v)-2\beta M_{\beta}(v)g(v)\left(V\right)\right)\right]
\end{equation}
Notice that $P_{V}^{\perp}\left(2\beta M_{\beta}(v)g(v)\left(V\right)\right)=0$,
this yields 
\begin{equation}
K(g)=A\intop dv_{1}M_{\beta}(v_{1})\nabla_{v}\cdot\left[\frac{1}{\vert V\vert}P_{\mathrm{V}}^{\perp}\left(M_{\beta}(v)\nabla_{v}g(v)\right)\right]\label{eq:ut1}
\end{equation}
We also notice that $\nabla_{v}\frac{1}{\vert V\vert}$ is parallel
to $V$ , we calculate the divergence and obtain 
\[
\nabla_{v}\cdot\left[\frac{1}{\vert V\vert}P_{(\mathrm{V})}^{\perp}\left(M_{\beta}(v)\nabla_{v}g(v)\right)\right]=\nabla_{v}\frac{1}{\vert V\vert}\cdot P_{(\mathrm{V})}^{\perp}\left(M_{\beta}(v)\nabla_{v}g(v)\right)+\frac{1}{\vert V\vert}\nabla_{v}\cdot P_{(\mathrm{V})}^{\perp}\left(M_{\beta}(v)\nabla_{v}g(v)\right)=
\]
\[
\frac{1}{\vert V\vert}\nabla_{v}\cdot P_{V}^{\perp}\left(M_{\beta}(v)\nabla_{v}g(v)\right)
\]
Therefore by (\ref{eq:ut1}) we have 
\begin{equation}
K(g)=A\intop dv_{1}M_{\beta}(v_{1})\frac{1}{\vert V\vert}\nabla_{v}\cdot\left[M_{\beta}(v)P_{V}^{\perp}\left(\nabla_{v}g(v)\right)\right]\label{eq:ut2}
\end{equation}
We calculate again the divergence
\[
\nabla_{v}\cdot\left[M_{\beta}(v)P_{V}^{\perp}\left(\nabla_{v}g(v)\right)\right]=-2\beta vM_{\beta}(v)\cdot P_{V}^{\perp}\left(\nabla_{v}g(v)\right)+M_{\beta}(v)\nabla_{v}\cdot\left[P_{V}^{\perp}\left(\nabla_{v}g(v)\right)\right]=
\]
\[
-2\beta\left(v-v_{1}\right)M_{\beta}(v)\cdot P_{V}^{\perp}\left(\nabla_{v}g(v)\right)-2\beta v_{1}M_{\beta}(v)\cdot P_{V}^{\perp}\left(\nabla_{v}g(v)\right)+M_{\beta}(v)\nabla_{v}\cdot\left[P_{V}^{\perp}\left(\nabla_{v}g(v)\right)\right]=
\]
\begin{equation}
-2\beta v_{1}M_{\beta}(v)\cdot P_{V}^{\perp}\left(\nabla_{v}g(v)\right)+M_{\beta}(v)\nabla_{v}\cdot\left[P_{V}^{\perp}\left(\nabla_{v}g(v)\right)\right]\label{eq:ut3}
\end{equation}
From (\ref{eq:ut3}) and (\ref{eq:ut2}) we arrive to
\[
K(g)=A\intop dv_{1}M_{\beta}(v_{1})\frac{1}{\vert V\vert}\left\{ -2\beta v_{1}M_{\beta}(v)\cdot P_{V}^{\perp}\left(\nabla_{v}g(v)\right)+M_{\beta}(v)\nabla_{v}\cdot\left[P_{\mathrm{V}}^{\perp}\left(\nabla_{v}g(v)\right)\right]\right\} =
\]
\[
A\intop dv_{1}M_{\beta}(v_{1})\frac{1}{\vert V\vert}\left[-2\beta v_{1}M_{\beta}(v)\cdot P_{V}^{\perp}\left(\nabla_{v}g(v)\right)\right]+
\]
\begin{equation}
A\intop dv_{1}M_{\beta}(v_{1})\frac{1}{\vert V\vert}M_{\beta}(v)\nabla_{v}\cdot\left[P_{V}^{\perp}\left(\nabla_{v}g(v)\right)\right]\label{eq:ut4}
\end{equation}
Now we work on the first term of the right hand side of (\ref{eq:ut4}).
Since $-2\beta v_{1}M_{\beta}(v_{1})=\nabla_{v_{1}}M_{\beta}(v_{1})$,
by means of the divergence Theorem we have that 
\[
A\intop dv_{1}M_{\beta}(v_{1})\frac{1}{\vert V\vert}\left[-2\beta v_{1}M_{\beta}(v)\cdot P_{V}^{\perp}\left(\nabla_{v}g(v)\right)\right]=M_{\beta}(v)A\intop dv_{1}\left(-2\beta v_{1}\right)M_{\beta}(v_{1})\cdot\frac{P_{V}^{\perp}\left(\nabla_{v}g(v)\right)}{\vert V\vert}=
\]
\begin{equation}
-M_{\beta}(v)A\intop dv_{1}M_{\beta}(v_{1})\nabla_{v_{1}}\cdot\left[\frac{P_{V}^{\perp}\left(\nabla_{v}h(v)\right)}{\vert V\vert}\right]=-M_{\beta}(v)A\intop dv_{1}M_{\beta}(v_{1})\frac{1}{\vert V\vert}\nabla_{v_{1}}\cdot\left[P_{V}^{\perp}\left(\nabla_{v}g(v)\right)\right]\label{eq:ut5}
\end{equation}
From (\ref{eq:ut5}) and (\ref{eq:ut4}) we arrive to
\begin{equation}
K(g)=M_{\beta}(v)A\intop dv_{1}M_{\beta}(v_{1})\frac{1}{\vert V\vert}\left(\nabla_{v}-\nabla_{v_{1}}\right)\cdot\left[P_{V}^{\perp}\left(\nabla_{v}g(v)\right)\right]\label{eq:ut6}
\end{equation}
Now we want to calculate $\nabla_{v}\cdot\left[P_{V}^{\perp}\left(\nabla_{v}g(v)\right)\right]$
and $\nabla_{v_{1}}\cdot\left[P_{V}^{\perp}\left(\nabla_{v}g(v)\right)\right]$.
First we observe that 
\begin{equation}
P_{V}^{\perp}\left(\nabla_{v}g(v)\right)=\nabla_{v}g(v)-\frac{\left(V,\nabla_{v}g(v)\right)V}{\vert V\vert^{2}}
\end{equation}
and so 
\[
\nabla_{v_{1}}\cdot\left[P_{V}^{\perp}\left(\nabla_{v}g(v)\right)\right]=\nabla_{v_{1}}\cdot\left[\nabla_{v}g(v)-\frac{\left(V,\nabla_{v}g(v)\right)V}{\vert V\vert^{2}}\right]=-\nabla_{v_{1}}\cdot\left[\frac{\left(V,\nabla_{v}g(v)\right)V}{\vert V\vert^{2}}\right]=
\]
\begin{equation}
2\frac{\left(V,\nabla_{v}g(v)\right)}{\vert V\vert^{2}}\label{eq:ut7}
\end{equation}
For the other term we have that 
\[
\nabla_{v}\cdot\left[P_{V}^{\perp}\left(\nabla_{v}g(v)\right)\right]=\nabla_{v}\cdot\left[\nabla_{v}g(v)-\frac{\left(V,\nabla_{v}g(v)\right)V}{\vert V\vert^{2}}\right]=\triangle g(v)-\nabla_{v}\cdot\left[\frac{\left(V,\nabla_{v}g(v)\right)V}{\vert V\vert^{2}}\right]=
\]
\begin{equation}
\triangle g(v)-\left[2\frac{\left(V,\nabla_{v}g(v)\right)}{\vert V\vert^{2}}+\frac{\left(V,D^{2}V\right)}{\vert V\vert^{2}}\right]\label{eq:ut8}
\end{equation}
We now use (\ref{eq:ut7}) and (\ref{eq:ut8}) together with (\ref{eq:ut6})
to get
\begin{equation}
K(g)=M_{\beta}(v)A\intop dv_{1}M_{\beta}(v_{1})\frac{1}{\vert V\vert^{3}}\left[\vert V\vert^{2}\triangle g(v)-\left(V,D^{2}V\right)-4V\cdot\nabla_{v}g(v)\right]
\end{equation}
Finally we can define the linear Landau equation
\begin{equation}
\partial_{t}g+v\cdot\nabla_{x}g=\tilde{Q}_{L}(g)
\end{equation}
where $\tilde{Q}_{L}$ is the linear Landau operator defined as
\begin{equation}
\tilde{Q}_{L}(g)=A\intop dv_{1}M_{\beta}(v_{1})\frac{1}{\vert V\vert^{3}}\left[\vert V\vert^{2}\triangle g(v)-\left(V,D^{2}V\right)-4V\cdot\nabla_{v}g(v)\right]\label{eq:LinLandOp}
\end{equation}
Notice that $\tilde{Q}_{L}$ and $Q_{L}$ are the same operator if
$A=B$. The constant $A$ is precisally characterized by the formal
derivation of the Landau equation from a system of particles and it
has the following value 
\begin{equation}
A=\frac{1}{8\pi}\intop_{0}^{+\infty}dr\,r^{3}\hat{\Phi}(r)^{2}
\end{equation}
where $\hat{\Phi}(\vert k\vert)=\intop dx\,\Phi(\vert x\vert)e^{-ik\cdot x}$.
It can be easily proved that $A=B$ by following the calculations
made in \cite{Kirkpatrick2009} and, therefore, that $\tilde{Q}_{L}=Q_{L}$.\pagebreak{}

\section{Continuity estimates}

In this section we will prove some useful estimates for the operators
arising in the series solution of the hierarchies. Observe that in
the case of $\alpha=1$ these estimates are enough to prove the convergence
of the series solution for a small time. In our case since $\alpha\rightarrow\infty$
the time of the convergence of the series is going to zero. As we
will see in the next section we can still use these estimates in the
linear case thanks to the a priori estimate .

We define the following norm

\begin{equation}
\Vert f_{j}(\boldsymbol{z}_{j})\Vert_{\beta}=\sup_{\boldsymbol{z}_{j}\in\Lambda_{j}}\left(e^{\beta H(\boldsymbol{z}_{j})}f_{j}(\boldsymbol{z}_{j})\right)
\end{equation}
where the hamiltonian $H(\boldsymbol{z}_{j})$ in macroscopic variables
reads as
\begin{equation}
H(\boldsymbol{z}_{j})=\frac{1}{2}\sum_{i=1}^{j}\vert v_{i}\vert^{2}+\frac{1}{2\sqrt{\alpha}}\sum_{i,k=1,\,i\neq k}^{j}\Phi(\frac{x_{i}-x_{k}}{\epsilon})
\end{equation}

For sake of simplicity we don't indicate the dependence from j in
the definition of $\Vert\cdot\Vert_{\beta}$. Notice also that the
norm depends on $\alpha$ but not in a harmful way. 

Since we are interested in the linear regime we will take as initial
datum a perturbation of the stationary state, as we have seen in section
3.1 and 3.2. We assume that the initial datum of GH and IH has the
form
\begin{equation}
f_{j,0}^{N}(\boldsymbol{z}_{j})=M_{N,\beta}(\boldsymbol{z}_{j})g_{0}(x_{1},v_{1})
\end{equation}
We assume also that the initial data for the Boltzmann hierarchy is
\begin{equation}
f_{0}^{\alpha}(\boldsymbol{z}_{j})=M_{\beta}^{\otimes j}(v_{j})g_{0}(x_{1},v_{1})
\end{equation}
Notice that the estimates that we will prove work also in case of
a general $f_{0}$ with $\Vert f_{0}\Vert_{\beta}<\infty$ for a $\beta>0$. 

\subsection{Estimates of the operators}

We start by estimating the operator appearing in GSS
\begin{lem}
\label{lem:ConEst}Let $g_{j}^{N}$ be a sequence of continuous functions
with $g_{j}^{N}=0$ for $j>N$ and suppose that 
\begin{equation}
\Vert g_{j}^{N}\Vert_{\beta}\leq C^{j}
\end{equation}
 Then for $\beta^{'}<\beta$ there exist a constant $C_{1}=C_{1}(\beta,\beta^{'},g_{j}^{N})$
such that for $\epsilon$ small enough and $\forall j\geq0$ 
\begin{equation}
\Vert G_{j,n}^{\epsilon}(t)g_{j}^{N}(z_{j})\Vert_{\beta^{'}}\leq\left(C_{1}\alpha t\right)^{n}
\end{equation}
\end{lem}
\begin{proof}
From the definition of the operator $A_{j+1+m}^{\epsilon}g_{j+1+m}^{N}$
we have that
\begin{equation}
e^{\beta^{'}H(z_{j})}\vert A_{j+1+m}^{\epsilon}g_{j+1+m}^{N}(\boldsymbol{z}_{j})\vert\leq C^{j+1}C^{m}\epsilon^{3m}\epsilon^{2}N^{m+1}\sum_{i=1}^{j}\intop dv_{j+1}\left(\vert v_{i}\vert+\vert v_{j+1}\vert\right)e^{-\frac{\beta-\beta^{'}}{2}\sum_{i=1}^{j}v_{i}^{2}}e^{-\frac{\beta}{2}v_{j+1}^{2}}
\end{equation}
since 
\begin{equation}
\intop_{\Delta_{m}(\boldsymbol{x}_{j+1})}d\boldsymbol{z}_{j+1,m}f_{j+1+m}^{N}(\boldsymbol{z}_{j},x_{i}+\nu\epsilon,v_{j+1},\boldsymbol{z}_{j+1,m},t)\leq C^{m}\epsilon^{3m}
\end{equation}
and
\begin{equation}
\Vert g_{j+1+m}^{N}\Vert_{\beta}\leq C^{j+1}C^{m}
\end{equation}
Now since $\epsilon^{2}N\cong\alpha$ we have that
\begin{equation}
\epsilon^{3m}\epsilon^{2}N^{m+1}\leq\alpha(C\epsilon\alpha)^{m}
\end{equation}
and so
\begin{equation}
e^{\beta^{'}H(z_{j})}\vert A_{j+1+m}^{\epsilon}g_{j+1+m}^{N}(\boldsymbol{z}_{j})\vert\leq n\alpha(C\epsilon\alpha)^{m}
\end{equation}
We can choose $\epsilon$ small enough, since $\alpha\cong\sqrt{\log\log N}$,
to have that $C\epsilon\alpha<1$. We performe the sum over m to obtain
\begin{equation}
\sum_{m\geq0}\left(C\epsilon\alpha\right)^{m}\leq C
\end{equation}
that leads us to
\begin{equation}
\Vert A_{j+1+m}^{\epsilon}g_{j+1+m}^{N}(\boldsymbol{z}_{j})\Vert_{\beta^{'}}\leq n\alpha C\label{eq:4ut1}
\end{equation}
Now since for any $\beta>0$ it results that
\begin{equation}
\Vert S_{j}^{\epsilon}(t)g_{j}^{N}\Vert_{\beta}=\Vert g_{j}^{N}\Vert_{\beta}\label{eq:4ut2}
\end{equation}
we can alternate estimate (\ref{eq:4ut1}) and (\ref{eq:4ut2}) and
performe the time integrals in (\ref{eq:GSSOp}). This gives us that
\begin{equation}
\Vert G_{j,n}^{\epsilon}(t)g_{j}^{N}(\boldsymbol{z}{}_{j})\Vert_{\beta^{'}}\leq\left(C_{1}\alpha t\right)^{n}
\end{equation}
\end{proof}
In the same way we can estimate the operators $Q_{j,n}^{\epsilon}(t)$
and $Q_{j,n}^{\alpha}(t)$ and prove the following lemma
\begin{lem}
Let $Q_{j,n}^{\epsilon}(t)$ and $Q_{j,n}^{\alpha}(t)$ be defined
respectively as in (\ref{eq:ISSOp}) and in (\ref{eq:BSSOp}). Let
also $g_{j}^{N},g_{j}$ be sequence of continuous functions with $g_{j}^{N}=0$
for $j>N$ suppose that 
\begin{equation}
\Vert g_{j}^{N}\Vert_{\beta}\leq C^{j}
\end{equation}
\begin{equation}
\Vert g_{j}\Vert_{\beta}\leq C^{j}
\end{equation}
 then there exist constants $C_{2}$ and $C_{3}$ such that for $\epsilon$
small enough and $\beta^{'}<\beta$ 
\begin{equation}
\Vert Q_{j,n}^{\epsilon}(t)g_{j}^{N}\Vert_{\beta^{'}}\leq\left(C_{2}\alpha t\right)^{n}
\end{equation}
\begin{equation}
\Vert Q_{j,n}^{\alpha}(t)g_{j}\Vert_{\beta^{'}}\leq\left(C_{3}\alpha t\right)^{n}
\end{equation}
\end{lem}

\subsection{Estimates for an arbitrary time}

Now we want to use the a priori estimate to prove the convergence
of the series solution for an arbitrary time. The main idea is to
separate the interval $\left[0,t\right]$ in $s\in\mathbb{N}$ parts
of length $h$ such that

\begin{equation}
t=sh
\end{equation}
and write $\widetilde{f_{1}^{N}}(t)$, $f_{1,s}^{\alpha}(t)$ and
$f_{1,s}^{N}(t)$ in terms of a finite sum plus a remainder. We use
the technique used by Bodineau, Gallagher and Saint-Raymond \cite{Bodineau2015}.
It consists in bounding the number of interactions in an interval
$[ih,(i+1)h]\,\,0\leq i<s$ by $2^{i}-1$ and send the time $h$ to
zero in a suitable way.

In literature there is another method, which is employed by Colangeli,
Pezzoti and Pulvirenti in \cite{Colangeli1975}, that consists in
taking $h$ smaller than the Lanford time of the convergence of the
series solutions and then bounding in a suitable way the number of
creations in each interval. We cannot use this method since in our
case the time of the convergence of the series is going to zero.

We can write the solution at time t of the GH as the evolution of
a time $h$ of the solution at time $t-h$
\begin{equation}
\widetilde{f_{1}^{N}}(t)=\sum_{j_{1}=0}^{\infty}G_{1,j_{1}}^{\epsilon}(h)\widetilde{f_{j_{1}+1}^{N}}(t-h)\label{eq:imp1}
\end{equation}
We introduce the Grad truncated series solution (GTS) by truncating
the series (\ref{eq:imp1}) at $j_{1}=2^{1}-1=1$. We obtain 
\begin{equation}
\widetilde{f_{1}^{N}}(t)=\sum_{j_{1}=0}^{1}G_{1,j_{1}}^{\epsilon}(h)\widetilde{f_{j_{1}+1}^{N}}(t-h)+\widetilde{R_{1,1}}(t-h,t)\label{eq:imp2}
\end{equation}
\begin{equation}
\widetilde{R_{1,1}}(t-h,t)=\sum_{j_{1}=2}^{\infty}G_{1,j_{1}}^{\epsilon}(h)\widetilde{f_{j_{1}+1}^{N}}(t-h)
\end{equation}
Now we can iterate this procedure on $\widetilde{f_{j_{1}+1}^{N}}(t-h)$.
We have that
\begin{equation}
\widetilde{f_{j_{1}+1}^{N}}(t-h)=\sum_{j_{2}=0}^{\infty}G_{j_{1}+1,j_{2}}^{\epsilon}(h)\widetilde{f_{j_{2}+1}^{N}}(t-h)
\end{equation}
 We truncate again the series at $j_{2}=2^{2}-1$ and we arrive to

\begin{equation}
\widetilde{f_{j_{1}+1}^{N}}(t-h)=\sum_{j_{2}=0}^{2^{2}-1}G_{j_{1}+1,j_{2}}^{\epsilon}(h)\widetilde{f_{j_{2}+1}^{N}}(t-2h)+\widetilde{R_{j_{1}+1,2}}(t-2h,t-h)\label{eq:imp3}
\end{equation}
where 
\begin{equation}
\widetilde{R_{j_{1}+1,2}}(t-2h,t-h)=\sum_{p=4}^{N-j_{1}-1}G_{j_{1}+1,2}^{\epsilon}(h)\widetilde{f_{j_{1+1}+p}^{N}}(t-2h)
\end{equation}
From (\ref{eq:imp3}) and (\ref{eq:imp2}) we have 
\begin{equation}
\widetilde{f_{1}^{N}}(t)=\sum_{j_{1}=0}^{1}\sum_{j_{2}=0}^{2^{2}-1}G_{1,j_{1}}^{\epsilon}(h)G_{j_{1}+1,j_{2}}^{\epsilon}(h)\widetilde{f_{j_{2}+1}^{N}}(t-2h)+\widetilde{R_{N}^{2}}(t)
\end{equation}
where $\widetilde{R_{N}^{2}}(t)$ takes into account the evolution
of the remainders of each truncation and reads as follows 
\begin{equation}
\widetilde{R_{N}^{2}}(t)=\widetilde{R_{1,1}}(t-h,t)+\sum_{j_{1}=0}^{1}G_{1,j_{1}}^{\epsilon}(h)\widetilde{R_{j_{1}+1,2}}(t-2h,t-h)
\end{equation}
We iterate this procedure with a sequence of cutoffs $2^{i}-1$, this
leads to
\begin{equation}
\widetilde{f_{1}^{N}}(t)=\widetilde{f_{1,s}^{N}}(t)+\widetilde{R_{N}^{s}}(t)
\end{equation}
where, denoting with $P_{i}=1+\sum_{k=1}^{i}j_{k}$ the number of
particles after i iterations, 
\begin{equation}
\widetilde{f_{1,s}^{N}}(t)=\sum_{j_{1}=0}^{1}...\sum_{j_{s}=0}^{2^{s}-1}G_{1,j_{1}}^{\epsilon}(h)G_{P_{1},j_{2}}^{\epsilon}(h)...G_{P_{s-1},j_{S}}^{\epsilon}(h)f_{0}^{N}
\end{equation}

\begin{equation}
\widetilde{R_{N}^{s}}(t)=\sum_{i=1}^{s}\sum_{j_{1}=0}^{1}...\sum_{j_{i-1}=0}^{2^{i-1}-1}G_{1,j_{1}}^{\epsilon}(h)G_{P_{1},j_{2}}^{\epsilon}(h)...G_{P_{i-2},j_{i-1}}^{\epsilon}(h)\tilde{R}_{P_{i-1},i}\left(t-ih,t-(i-1)h\right)\label{eq:imp4}
\end{equation}
\begin{equation}
\tilde{R}_{P_{i-1},i}\left(t-ih,t-(i-1)h\right)=\sum_{p=2^{i}}^{N-P_{i-1}}G_{P_{i-1},p}^{\epsilon}(h)\tilde{f}_{P_{i-1}+p}^{N}
\end{equation}

We use the same procedure for the series solution of the Boltzmann
hierarchy and we obtain the truncated Boltzmann solution (BTS)
\begin{equation}
f_{1,s}^{\alpha}(t)=\sum_{j_{1}=0}^{1}...\sum_{j_{s}=0}^{2^{s}-1}Q_{1,j_{1}}^{\alpha}(h)Q_{P_{1},j_{2}}^{\alpha}(h)...Q_{P_{s-1},j_{s}}^{\alpha}(h)f_{0}^{N}
\end{equation}

\begin{equation}
R^{s}(t)=\sum_{i=1}^{s}\sum_{j_{1}=0}^{1}...\sum_{j_{i-1}=0}^{2^{i-1}-1}Q_{1,P_{1}}^{\alpha}(h)Q_{P_{1},j_{2}}^{\alpha}(h)...Q_{P_{i-2},j_{i-1}}^{\alpha}(h)R_{P_{i-1},i}\left(t-ih,t-(i-1)h\right)\label{eq:imp5}
\end{equation}
\begin{equation}
R_{P_{i-1},i}\left(t-ih,t-(i-1)h\right)=\sum_{p=2^{i}}^{N-P_{i-1}}Q_{P_{i-1},p}^{\alpha}(h)f_{J_{i-1}+p}^{N}
\end{equation}
We also define the intermediate truncated solution (ITS) 
\begin{equation}
f_{1,s}^{N}(t)=\sum_{j_{1}=0}^{1}...\sum_{j_{s}=0}^{2^{s}-1}Q_{1,j_{1}}^{\epsilon}(h)Q_{P_{1},j_{2}}^{\epsilon}(h)...Q_{P_{s-1},j_{s}}^{\epsilon}(h)f_{0}^{N}
\end{equation}

Now we want to prove an estimate for the remainder term. 
\begin{thm}
\label{thm:Rem}Let $\widetilde{R_{N}^{s}}(t)$,$R^{s}(t)$ be defined
respectively as in (\ref{eq:imp4}) and (\ref{eq:imp5}). Then the
following estimate holds 
\begin{equation}
\Vert\widetilde{R_{N}^{s}}(t)\Vert_{\infty}+\Vert R^{s}(t)\Vert_{\infty}\leq\Vert g_{0}\Vert_{\infty}\left(\frac{C\left(\alpha t\right)^{2}}{s}\right)^{2}
\end{equation}
\end{thm}
\begin{proof}
Thanks to the semigroup property we have that
\begin{equation}
\widetilde{R_{s}^{N}}(t)=\sum_{i=1}^{s}\sum_{j_{1}=0}^{1}...\sum_{j_{i-1}=0}^{2^{i-1}-1}G_{1,P_{i-1}-1}^{\epsilon}((i-1)h)\tilde{R}_{P_{i-1},i}\left(t-ih,t-(i-1)h\right)
\end{equation}
Now from the steps of Lemma \ref{lem:ConEst} it follows that 
\[
\Vert G_{1,P_{i-1}-1}^{\epsilon}((i-1)h)\tilde{R}_{P_{i-1},i}\left(t-ih,t-(i-1)h\right)\Vert_{\infty}\leq
\]
\begin{equation}
\left(C\alpha(i-1)h\right)^{P_{i-1}-1}\Vert\tilde{R}_{P_{i-1},i}\left(t-ih,t-(i-1)h\right)\Vert_{\frac{\beta}{2}}
\end{equation}
Furthermore we have that 
\begin{equation}
\Vert\tilde{R}_{P_{i-1},i}\left(t-ih,t-(i-1)h\right)\Vert_{\frac{\beta}{2}}\leq\sum_{p=2^{i}}^{N-P_{i-1}}\left(C\alpha h\right)^{p}\Vert\widetilde{f_{P_{i-1}+p}^{N}}\Vert_{\frac{\beta}{2}}\leq\Vert g_{0}\Vert_{\infty}\sum_{p=2^{i}}^{N-P_{i-1}}\left(C\alpha h\right)^{p}
\end{equation}
We use together the last two estimates and that
\[
C\alpha h<\frac{1}{2}
\]
and we arrive to

\[
\Vert\widetilde{R_{s}^{N}}(t)\Vert_{\infty}\leq\Vert g_{0}\Vert_{\infty}\sum_{i=1}^{s}\sum_{j_{1}=0}^{1}...\sum_{j_{i-1}=0}^{2^{i-1}-1}\left(C\alpha t\right)^{P_{i-1}-1}\left(C\alpha h\right)^{2^{i}}\leq
\]
\[
\Vert g_{0}\Vert_{\infty}\sum_{i=1}^{s}\sum_{j_{1}=0}^{1}...\sum_{j_{i-1}=0}^{2^{i-1}-1}\left(C\alpha t\right)^{2^{i}}\left(C\alpha h\right)^{2^{i}}\leq\Vert g_{0}\Vert_{\infty}\sum_{i=1}^{s}2^{i(i-1)}\left(C\left(\alpha t\right)\alpha h\right)^{2^{i}}
\]
\begin{equation}
\leq\Vert g_{0}\Vert_{\infty}\sum_{i=1}^{s}\left(\frac{C\left(\alpha t\right)^{2}}{s}\right)^{2^{i}}
\end{equation}
In the last step we used that $h=\frac{t}{s}$ and that $i(i-1)\leq2^{i}$.
Now we assume also that 
\begin{equation}
\frac{C\left(\alpha t\right)^{2}}{s}<\frac{1}{2}
\end{equation}
and we finally arrive to
\begin{equation}
\Vert\widetilde{R_{N}^{s}}(t)\Vert_{\infty}\leq\Vert g_{0}\Vert_{\infty}\left(\frac{C\left(\alpha t\right)^{2}}{s}\right)^{2}
\end{equation}
The estimate for $R^{s}(t)$ can be obtained in the same way.
\end{proof}
Thanks to Theorem \ref{thm:Rem} we can work directly on the truncated
series since we have an estimate on the remainders. We want to prove
that the GTS is close to the ITS as $\epsilon\rightarrow0$. We have 
\begin{thm}
\label{thm:ghih}Let $\widetilde{f_{1,s}^{N}}(t)$, $f_{1,s}^{N}(t)$
be respectively the solution of the first equation of GH and IH. Then
the following estimate holds for all $t\geq0$ 
\begin{equation}
\Vert\widetilde{f_{1,s}^{N}}(t)-f_{1,s}^{N}(t)\Vert_{\infty}\leq\Vert g_{0}\Vert_{\infty}2^{s(s+1)}\left(C\alpha t\right)^{2^{s+1}}\epsilon
\end{equation}
\end{thm}
\begin{proof}
The definition of the truncated solution series leads to 
\begin{equation}
\widetilde{f_{j,s}^{N}}(t)-f_{j,s}^{N}(t)=\sum_{j_{1}=0}^{2}...\sum_{j_{s}=0}^{2^{s}-1}\left[G_{1,j_{1}}^{\epsilon}(h)G_{P_{1},j_{2}}^{\epsilon}(h)...G_{P_{s-1},j_{S}}^{\epsilon}(h)-Q_{1,j_{1}}^{\epsilon}(h)Q_{P_{1},j_{2}}^{\epsilon}(h)...Q_{P_{s-1},j_{s}}^{\epsilon}(h)\right]f_{0}^{N}
\end{equation}
Now from the semigroup property and the identity 
\begin{equation}
a^{n}-b^{n}=\sum_{i=1}^{n}a^{i-1}(a-b)b^{n-i}
\end{equation}
we have
\[
G_{1,j_{1}}^{\epsilon}(h)G_{P_{1},j_{2}}^{\epsilon}(h)...G_{P_{s-1},j_{S}}^{\epsilon}(h)-Q_{1,j_{1}}^{\epsilon}(h)Q_{P_{1},j_{2}}^{\epsilon}(h)...Q_{P_{s-1},j_{s}}^{\epsilon}(h)=
\]
\begin{equation}
\sum_{l=1}^{s}G_{1,P_{l-1}-1}^{\epsilon}((l-1)h)\left[G_{P_{l-1},j_{l}}^{\epsilon}(h)-Q_{P_{l-1},j_{l}}^{\epsilon}(h)\right]Q_{P_{l},P_{s}-P_{l}}^{\epsilon}((s-l)h)\label{eq:Asybbg1}
\end{equation}
Since the operator $Q_{P_{l-1},j_{l}}^{\epsilon}(h)$ is the first
term not equal to zero in the asymptothic of the operator $G_{P_{l-1},j_{l}}^{\epsilon}(h)$
we obtain that

\[
G_{P_{l-1},j_{l}}^{\epsilon}(h)-Q_{P_{l-1},j_{l}}^{\epsilon}(h)f^{N}(0)=\underset{P_{l-1}+j_{l}+\sum_{i=1}^{j_{l}}m_{i}\leq N}{\sum_{m_{1},...,m_{j_{l}}\geq0,\,\sum_{i=1}^{j_{l}}m_{i}\neq0}}\intop_{0}^{h}dt_{1}...\intop_{0}^{t_{j_{l}-1}}dt_{j_{l}}
\]
\begin{equation}
S_{P_{l-1}}^{\epsilon}(h-t_{1})A_{P_{l-1}+1+m_{1}}^{\epsilon}S_{P_{l-1}+1+m_{1}}^{\epsilon}(t_{1}-t_{2})...A_{P_{l-1}+j_{l}+\sum_{i=1}^{j_{l}}m_{i}}^{\epsilon}S_{P_{l-1}+j_{l}+\sum_{i=1}^{j_{l}}m_{i}}^{\epsilon}(t_{j_{l}})f_{0,P_{l-1}+j_{l}+\sum_{i=1}^{j_{l}}m_{i}}^{N}
\end{equation}
 The same steps of Lemma \ref{lem:ConEst} lead to
\begin{equation}
\Vert G_{P_{l-1},j_{l}}^{\epsilon}(h)-Q_{P_{l-1},j_{l}}^{\epsilon}(h)f^{N}(0)\Vert_{\beta^{'}}\leq\left(C\alpha h\right)^{j_{l}}\Vert g_{0}\Vert_{\infty}\epsilon\label{eq:asybbg2}
\end{equation}
From (\ref{eq:asybbg2}) and (\ref{eq:Asybbg1}) we arrive to 

\begin{equation}
\sum_{l=1}^{s}\Vert G_{1,P_{l-1}-1}^{\epsilon}((l-1)h)\left[G_{P_{l-1},j_{l}}^{\epsilon}(h)-Q_{P_{l-1},j_{l}}^{\epsilon}(h)\right]Q_{P_{l},P_{s}-P_{l}}^{\epsilon}((s-l)h)f_{0}^{N,\epsilon}\Vert_{\infty}\leq\Vert g_{0}\Vert_{\infty}\epsilon\left(C\alpha t\right)^{P_{s}-1}
\end{equation}
We perform the sum over $j_{1},...,j_{s}$ and we finally have that
\begin{equation}
\Vert\widetilde{f_{1,s}^{N}}(t)-f_{1,s}^{N}(t)\Vert_{\infty}\leq\Vert g_{0}\Vert_{\infty}2^{s(s+1)}\epsilon\left(C\alpha t\right)^{2^{s+1}}
\end{equation}
\end{proof}
Thanks to this theorem we can reduce us to study only the convergence
of the ITS to the BTS. 

\pagebreak{}

\section{Convergence to Linear Boltzmann equation}

\subsection{The Boltzmann backward flow and the Interacting backward flow}

In this section we will represent in a convenient way the series (\ref{eq:ISS})
and (\ref{eq:BSS}) for the first-particle marginal. These series
solutions can be represented graphically as a trees expansion. We
define a n-collision tree graph as the following collection of integer
\begin{equation}
\Gamma(n)=\left\{ \left(i_{1},...,i_{n}\right)\in\mathbb{N}^{n}\,|\,i_{k}\leq k\right\} 
\end{equation}
Roughly speaking, this integer represent the label of the particle
that creates a new particle in a creation term. In Figure (\ref{fig:alberi})
we give a picture of the tree $(1,1,2)$. 

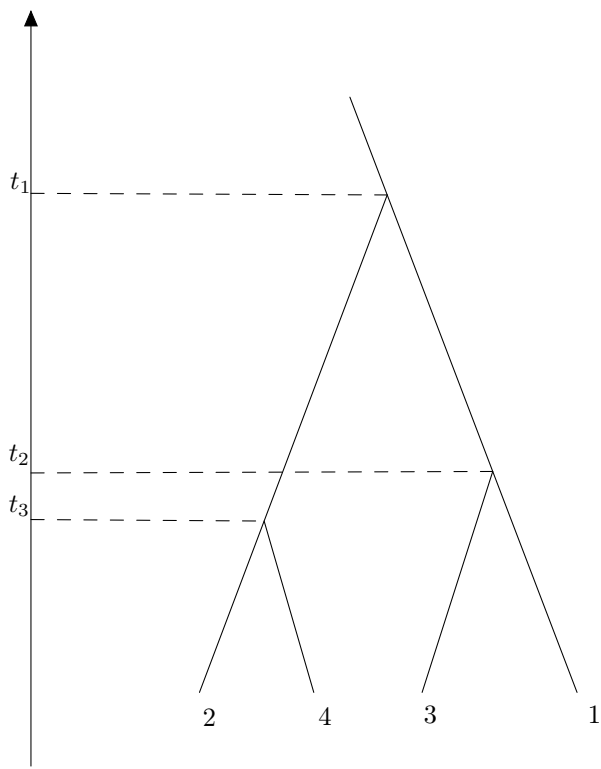
\begin{figure}[H]
\begin{tikzpicture}[line cap=round,line join=round,>=triangle 45,x=1.0cm,y=1.0cm]
\clip(-3.62,-4.2) rectangle (19.06,7.12);
\draw [->] (0.,-3.7) -- (0.,6.36);
\draw (4.24,5.2)-- (7.26,-2.72);
\draw (4.736252025142386,3.898570847970961)-- (2.24,-2.72);
\draw (6.13907402974106,0.21964691538106074)-- (5.2,-2.72);
\draw (3.0999317052606106,-0.43997424204024904)-- (3.76,-2.72);
\draw [dash pattern=on 5pt off 5pt] (0.,3.92)-- (4.736252025142386,3.898570847970961);
\draw [dash pattern=on 5pt off 5pt] (0.,0.2)-- (6.13907402974106,0.21964691538106074);
\draw [dash pattern=on 5pt off 5pt] (0.,-0.42)-- (3.0999317052606106,-0.43997424204024904);
\draw (7.28,-2.78) node[anchor=north west] {1};
\draw (2.16,-2.82) node[anchor=north west] {2};
\draw (5.1,-2.78) node[anchor=north west] {3};
\draw (3.7,-2.8) node[anchor=north west] {4};
\draw (-0.4,4.32) node[anchor=north west] {$t_1$};
\draw (-0.42,0.7) node[anchor=north west] {$t_2$};
\draw (-0.42,0.02) node[anchor=north west] {$t_3$};
\end{tikzpicture}

\caption{\label{fig:alberi}A representation of the tree graph $(1,1,2)$.
At the time $t_{1}$ we create the particle $2$ on the particle $1$.
Then at time $t_{2}$ we create the particle$3$ on the particle $1$.
Finally at time $t_{3}$ the particle $4$ is created on the particle
$2$.}
\end{figure}

We define the following collections of variables for the ITS

\begin{equation}
\boldsymbol{\sigma}_{n}=\left(\sigma_{1},...,\sigma_{n}\right)\,\,\,\,\,\,\sigma_{i}=\pm1\,\,\,\,\,\,\boldsymbol{\sigma}_{n}=\prod_{i=1}^{n}\sigma_{i}
\end{equation}

\begin{equation}
\boldsymbol{t}_{n}=\left(t_{1},...,t_{n}\right)
\end{equation}
\begin{equation}
\boldsymbol{w}_{n}=\left(w_{1},...,w_{n}\right)
\end{equation}
\begin{equation}
\boldsymbol{\nu}_{n}=\left(\nu_{1},...,\nu_{n}\right)
\end{equation}
Here $t_{1},...,t_{n}$ are the time variables appearing in the time
integrals, while $w_{i}$ and $\nu_{i}$ are the velocity and the
impact parameter that appears in the creation of the (i+1)-th particle.
Fixed these variables we can construct the interacting backwards flow
(IBF). We define the IBF at time $s\in\left(t_{k},t_{k+1}\right)$
as
\begin{equation}
\zeta^{\epsilon}(s)=\left(r_{1}^{\epsilon}(s),\xi_{1}^{\epsilon}(s),...,r_{1+k}^{\epsilon}(s),\xi_{1+k}^{\epsilon}(s)\right)
\end{equation}
where $r_{i}(s),\xi_{i}(s)$ are respectively position and velocity
of the i-Th particle at time $s$. At time $t$ we have that $\zeta^{\epsilon}(t)=\left(x_{1},v_{1}\right)$,
then we go back in time with the interacting flow defined as the solution
of equation (\ref{eq:Motioneq}). Between time $t$ and time $t-t_{1}$
we set $\zeta^{\epsilon}(s)=T_{1}^{\epsilon}(-s)\left(r_{1}^{\epsilon}(t),\xi_{1}^{\epsilon}(t)\right)$.
Then at time $t-t_{1}$ we create a new particle in position $r_{2}(t-t_{1})=r_{i_{1}}^{\epsilon}(t-t_{1})+\epsilon\nu_{1}$
with velocity $\xi_{2}(t-t_{1})=w_{1}$ in a pre-collisional state
if $\sigma_{1}=+1$ or in post collisional one if $\sigma_{1}=-1$.
Between time $t-t_{1}$ and $t-t_{1}-t_{2}$ we set the IBF as $\zeta^{\epsilon}(s)=T_{2}^{\epsilon}(-t+t_{1}+s)\left(r_{1}^{\epsilon}(t-t_{1}),\xi_{1}^{\epsilon}(t-t_{1}),r_{2}^{\epsilon}(t-t_{1}),\xi_{2}^{\epsilon}(t-t_{1})\right)$.
In this way we create a new particle at time $t-t_{1}-t_{2}$ in position
$r_{i_{2}}^{\epsilon}(t-t_{1}-t_{2})+\epsilon\nu_{2}$ with velocity
$w_{2}$ in pre-collisional or post-collisional configuration that
depends on $\sigma_{2}$. We iterate this procedure and we define
the IBF up to time $0$ by alternating the creation of new particles
with the interacting flow $T_{j}^{\epsilon}$ . For sake of simplicity
we define the following time variables 
\begin{equation}
\tau_{k}=t-\sum_{i=1}^{k}t_{i}
\end{equation}
With this definition $\tau_{k}$ are the backward times of a creation.

We can write the one particle marginal in a more manageable way thanks
to the IBF

\begin{equation}
f_{1,s}^{N}(t)=\sum_{j_{1}=0}^{1}...\sum_{j_{s}=0}^{2^{s}-1}\left(N-1\right)...\left(N-P_{s}-1\right)\left(\epsilon^{2}\right)^{P_{s}-1}\sum_{\Gamma(P_{s}-1)}\sum_{\sigma_{P_{s}-1}}\boldsymbol{\sigma}_{P_{s}-1}I_{\sigma_{P_{s}-1}}^{\epsilon}(z_{j},t)
\end{equation}
with
\begin{equation}
I_{\sigma_{P_{s}-1}}^{\epsilon}(z_{j},t)=\intop d\boldsymbol{t}_{P_{s-1}}d\boldsymbol{w}_{P_{s-1}}d\boldsymbol{\nu}_{P_{s-1}}\prod_{k=1}^{P_{s}-1}B^{\epsilon}\left(\nu_{k},w_{1+k}-\xi_{i_{k}}^{\epsilon}(\tau_{k})\right)\,f_{0,P_{s}}^{N}(\zeta^{\epsilon}(0))
\end{equation}
and 
\[
B^{\epsilon}\left(\nu_{k},w_{1+k}-\xi_{i_{k}}^{\epsilon}(\tau_{k})\right)=\vert\nu_{k}\cdot\left(w_{1+k}-\xi_{i_{k}}^{\epsilon}(\tau_{k})\right)\vert\chi\left\{ \vert r_{k+1}(\tau_{k})-r_{i_{k}}(\tau_{k})\vert>\epsilon\right\} \cdot
\]
\begin{equation}
\chi\left\{ \sigma_{k}\nu_{k}\cdot\left(w_{1+k}-\xi_{i_{k}}^{\epsilon}(\tau_{k})\right)\geq0\right\} 
\end{equation}

With a similar procedure we can build the Boltzmann backward flow
(BBF) but we have to take into account the following difference:
\begin{itemize}
\item The flow between two creation is the free flow and not the interacting
flow;
\item The new particle in each creation is created in the position of his
progenitor, i.e. $r_{i_{k}}(\tau_{k})=r_{k+1}(\tau_{k})$;
\item There is no constraint on $\nu_{k}$ other than the one implied by
the value of $\sigma_{k}$;
\item if $\sigma_{k}=+1$ before going back in time we have to change the
velocities from post collisional in pre-collisional according to the
scattering rules.
\end{itemize}
Taking into account these differences, we define the BBF at time $s\in\left(t_{k},t_{k+1}\right)$
as
\begin{equation}
\zeta(s)=\left(r_{1}(s),\xi_{1}(s),...,r_{1+k}(s),\xi_{1+k}(s)\right)
\end{equation}
We use the BBF to write the one particle marginal of the Boltzmann
equation as

\begin{equation}
f_{1,s}(t)=\sum_{j_{1}=0}^{1}...\sum_{j_{s}=0}^{2^{s}-1}\left(\alpha\right)^{P_{s}-1}\sum_{\Gamma(P_{s}-1)}\sum_{\sigma_{P_{s}-1}}\boldsymbol{\sigma}_{P_{s}-1}I_{\sigma_{P_{s}-1}}(z_{j},t)
\end{equation}
 
\begin{equation}
I_{\sigma_{P_{s}-1}}(z_{j},t)=\intop d\boldsymbol{t}_{P_{s-1}}d\boldsymbol{w}_{P_{s-1}}d\boldsymbol{\nu}_{P_{s-1}}\prod_{k=1}^{P_{s}-1}B\,f_{0,P_{s}}(\zeta(0))
\end{equation}
We also define the vectors of the only velocities at time $s\in\left(t_{k},t_{k+1}\right)$
as 
\begin{equation}
\xi^{\epsilon}(s)=\left(\xi_{1}^{\epsilon}(s),...,,\xi_{1+k}^{\epsilon}(s)\right)
\end{equation}
\begin{equation}
\xi(s)=\left(\xi_{1}(s),...,,\xi_{1+k}(s)\right)
\end{equation}

\begin{figure}[H]

\begin{tikzpicture}[line cap=round,line join=round,>=triangle 45,x=1.0cm,y=1.0cm]
\clip(-3.62,-4.2) rectangle (19.06,7.12);
\draw [->] (0.,-3.7) -- (0.,6.36);
\draw (4.24,5.2)-- (7.26,-2.72);
\draw (4.736252025142386,3.898570847970961)-- (2.24,-2.72);
\draw (6.13907402974106,0.21964691538106074)-- (5.2,-2.72);
\draw (3.0999317052606106,-0.43997424204024904)-- (3.76,-2.72);
\draw [dash pattern=on 5pt off 5pt] (0.,3.92)-- (4.736252025142386,3.898570847970961);
\draw [dash pattern=on 5pt off 5pt] (0.,0.2)-- (6.13907402974106,0.21964691538106074);
\draw [dash pattern=on 5pt off 5pt] (0.,-0.42)-- (3.0999317052606106,-0.43997424204024904);
\draw (7.28,-2.78) node[anchor=north west] {1};
\draw (2.16,-2.82) node[anchor=north west] {2};
\draw (5.1,-2.78) node[anchor=north west] {3};
\draw (3.7,-2.8) node[anchor=north west] {4};
\draw (-0.4,4.32) node[anchor=north west] {$t_1$};
\draw (-0.42,0.7) node[anchor=north west] {$t_2$};
\draw (-0.42,0.02) node[anchor=north west] {$t_3$};
\draw [dash pattern=on 1pt off 3pt on 5pt off 4pt] (3.653618915731883,-2.6960573188912624)-- (3.0130796020547397,-0.4834890955586059);
\draw [dash pattern=on 1pt off 3pt on 5pt off 4pt] (3.0130796020547397,-0.4834890955586059)-- (4.630644850540663,3.805328736330449);
\draw [dash pattern=on 1pt off 3pt on 5pt off 4pt] (4.594846449062731,3.8992106368950887)-- (4.138934844697329,5.094846367548729);
\end{tikzpicture}

\caption{We used a dashed line to evidence the virtual trajectory of the first
particle.}
\end{figure}
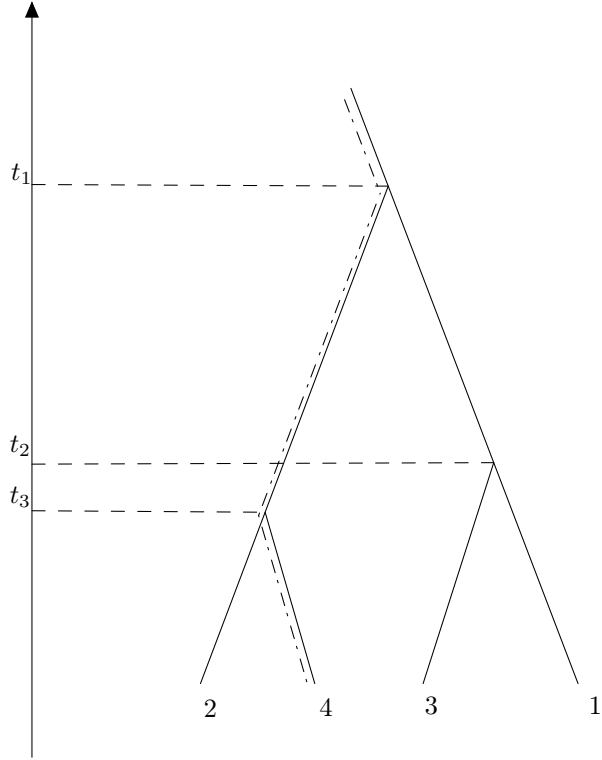

Now we define the virtual trajectory of the particle $i$ in the BBF
in an inductive way. We set $\zeta^{1}(s)=\left(r_{1}(s),\xi_{1}(s)\right)$,
then we define the inductive step 
\begin{equation}
\zeta^{i}(s)=\left(r^{i}(s),\xi^{i}(s)\right)=\begin{cases}
\left(r_{i}(s),\xi_{i}(s)\right) & s\in[0,t_{i-1}]\\
\left(r^{j}(s),\xi^{j}(s)\right) & s\in(t_{i-1},t]
\end{cases}
\end{equation}
where $j\in\{1,...,i-1\}$ is the progenitor of the particle $i$,
i.e. the particle where we create the particle $i$. With this definition
the virtual trajectory of a particle $i$ is its backward trajectory
until his creation, before of his creation it is the virtual trajectory
of its progenitors. 

\subsection{Estimate of the recolission }

We want to take advantage of the tree expansion to estimate the difference
between the intermediate truncated solution and the Boltzmann truncated
solution by estimating the difference between the IBF and the BBF.
The main difference between the IBF and the BBF are the recolission,
i.e. an interaction between particles which is not a creation. This
can happen only in the IBF and creates correlations. 

First we consider some cutoffs on the integration variables and estimate
the complementary term, denoting the various cutoff with an apex $Err_{i}$.
We establish some obvious estimates useful in the following. 
\begin{equation}
\sum_{j_{1}=0}^{1}...\sum_{j_{s}=0}^{2^{s}-1}1\leq2^{s(s+1)}
\end{equation}
\begin{equation}
P_{s}-1\leq\sum_{i=1}^{s}2^{i}\leq2^{s+1}
\end{equation}
\begin{equation}
\vert f_{0,P_{s}}^{N}(\zeta^{\epsilon}(0))\vert\leq Ce^{-\frac{\beta}{2}\vert\xi^{\epsilon}(0)\vert^{2}}
\end{equation}
\begin{equation}
\vert\vert f_{0,P_{s}}(\zeta(0))\vert\leq Ce^{-\frac{\beta}{2}\vert\xi(0)\vert^{2}}
\end{equation}
We also denote with $d\Lambda_{P_{s-1}}=d\boldsymbol{t}_{P_{s-1}}d\boldsymbol{w}_{P_{s-1}}d\boldsymbol{\nu}_{P_{s-1}}$.

First we estimate the error coming from the difference $\vert\alpha-N\epsilon^{2}\vert$.
\begin{lem}
\label{lem:Err1}Suppose that $\vert\alpha-N\epsilon^{2}\vert\leq\epsilon$
and let 
\begin{equation}
f_{1,s}^{N,Err_{1}}=\sum_{j_{1}=0}^{1}...\sum_{j_{s}=0}^{2^{s}-1}\left(\alpha\right)^{P_{s}-1}\sum_{\Gamma(P_{s}-1)}\sum_{\sigma_{P_{s}-1}}\boldsymbol{\sigma}_{P_{s}-1}\intop d\Lambda_{P_{s-1}}\prod_{k=1}^{P_{s}-1}B^{\epsilon}\,f_{0,P_{s}}^{N}(\zeta^{\epsilon}(0))
\end{equation}
Then 
\begin{equation}
\Vert f_{1,s}^{N}-f_{1,s}^{N,Err_{1}}\Vert_{\infty}\leq\Vert g_{0}\Vert_{\infty}\left(Ct\right)^{2^{s+1}}2^{s(s+1)}\epsilon
\end{equation}
\end{lem}
\begin{proof}
We recall that
\[
f_{1,s}^{N}=\sum_{j_{1}=0}^{1}...\sum_{j_{s}=0}^{2^{s}-1}\left(N-1\right)...\left(N-P_{s}-1\right)\left(\epsilon^{2}\right)^{P_{s}-1}\sum_{\Gamma(P_{s}-1)}\sum_{\sigma_{P_{s}-1}}\boldsymbol{\sigma}_{P_{s}-1}\intop d\Lambda_{P_{s-1}}\prod_{k=1}^{P_{s}-1}B^{\epsilon}\,f_{0,P_{s}}^{N}(\zeta^{\epsilon}(0))
\]
and since 
\[
\vert\left(N-1\right)...\left(N-P_{s}-1\right)\left(\epsilon^{2}\right)^{P_{s}-1}-(\alpha)^{P_{s}-1}\vert\leq2^{s+1}\alpha^{2^{s+1}}\vert\alpha-N\epsilon^{2}\vert
\]
it results that
\begin{equation}
\Vert f_{1,s}^{N}-f_{1,s}^{N,Err_{1}}\Vert_{\infty}\leq\Vert g_{0}\Vert_{\infty}2^{s(s+1)^{2}}\left(C\alpha t\right)^{2^{s+1}}\epsilon
\end{equation}
\end{proof}
Next we control the terms $\prod_{k=1}^{P_{s}-1}B^{\epsilon}$ and
$\prod_{k=1}^{P_{s}-1}B$. For $\lambda\in\left(0,1\right)$ we define
the indicator function 
\begin{equation}
\chi_{\lambda}^{\epsilon}=\chi\left\{ \prod_{k=1}^{P_{s}-1}B^{\epsilon}\leq\epsilon^{-\lambda}\right\} 
\end{equation}
and 
\begin{equation}
\chi_{\lambda}=\chi\left\{ \prod_{k=1}^{P_{s}-1}B\leq\epsilon^{-\lambda}\right\} .
\end{equation}
The following lemma gives us an estimate for the complementary terms.
\begin{lem}
Let 
\begin{equation}
f_{1,s}^{N,Err_{2}}=\sum_{j_{1}=0}^{1}...\sum_{j_{s}=0}^{2^{s}-1}\left(\alpha\right)^{P_{s}-1}\sum_{\Gamma(P_{s}-1)}\sum_{\sigma_{P_{s}-1}}\boldsymbol{\sigma}_{P_{s}-1}\intop d\Lambda_{P_{s-1}}\chi_{\lambda}^{\epsilon}\prod_{k=1}^{P_{s}-1}B^{\epsilon}\,f_{0,P_{s}}^{N}(\zeta^{\epsilon}(0))
\end{equation}
 and 
\begin{equation}
f_{1,s}^{\alpha,Err_{2}}=\sum_{j_{1}=0}^{1}...\sum_{j_{s}=0}^{2^{s}-1}\left(\alpha\right)^{P_{s}-1}\sum_{\Gamma(P_{s}-1)}\sum_{\sigma_{P_{s}-1}}\boldsymbol{\sigma}_{P_{s}-1}\intop d\Lambda_{P_{s-1}}\chi_{\lambda}^{\epsilon}\prod_{k=1}^{P_{s}-1}B\,f_{0,P_{s}}(\zeta(0))
\end{equation}
Then
\begin{equation}
\Vert f_{1,s}^{N,Err_{1}}-f_{1,s}^{N,Err_{2}}\Vert_{\infty}+\Vert f_{1,s}^{\alpha}-f_{1,s}^{\alpha,Err_{2}}\Vert_{\infty}\leq\epsilon^{\lambda}\Vert g_{0}\Vert_{\infty}(C\alpha t)^{2^{(s+1)}}2^{s(s+1)}
\end{equation}
\end{lem}
\begin{proof}
We prove the estimate only for $B^{\epsilon}$, the one for $B$ can
be obtained along the same lines. We have that 

\[
\vert f_{1,s}^{N}-f_{1,s}^{N,\lambda}\vert\leq\sum_{j_{1}=0}^{1}...\sum_{j_{s}=0}^{2^{s}-1}\left(C\alpha\right)^{P_{s}-1}\sum_{\Gamma(P_{s}-1)}\intop d\Lambda_{P_{s-1}}\chi\left\{ \prod_{k=1}^{P_{s}-1}B^{\epsilon}>\epsilon^{-\lambda}\right\} \vert\prod_{k=1}^{P_{s}-1}B^{\epsilon}\,f_{0,P_{s}}^{N}(\zeta^{\epsilon}(0))\vert\leq
\]
\begin{equation}
\epsilon^{\lambda}\Vert g_{0}\Vert_{\infty}\sum_{j_{1}=0}^{1}...\sum_{j_{s}=0}^{2^{s}-1}\left(C\alpha\right)^{P_{s}-1}\sum_{\Gamma(P_{s}-1)}\intop d\Lambda_{P_{s-1}}\vert\prod_{k=1}^{P_{s}-1}B^{\epsilon}\vert^{2}\,e^{-\frac{\beta}{2}\vert\xi^{\epsilon}(0)\vert^{2}}
\end{equation}
where we used that $1=\epsilon^{-\lambda}\epsilon^{\lambda}\leq\epsilon^{\lambda}\prod_{k=1}^{P_{s}-1}B^{\epsilon}$.
Now we observe that 

\begin{equation}
\sum_{\Gamma(P_{s}-1)}\vert\prod_{k=1}^{P_{s}-1}B^{\epsilon}\vert^{2}\leq2^{P_{s}-1}\prod_{k=1}^{P_{s}-1}\left(P_{s}\right)v_{k+1}^{2}+\sum_{i=1}^{P_{s}}v_{i}^{2}
\end{equation}
Therefore:
\begin{alignat}{1}
\vert f_{1,s}^{N}-f_{1,s}^{N,\lambda}\vert\leq & \epsilon^{\lambda}\Vert g_{0}\Vert_{\infty}(C\alpha)^{2^{(s+1)}}\sum_{j_{1}=0}^{1}...\sum_{j_{s}=0}^{2^{s}-1}\intop d\Lambda_{P_{s-1}}e^{-\frac{\beta}{4}\vert\xi^{\epsilon}(0)\vert^{2}}\nonumber \\
 & \prod_{k=1}^{P_{s}-1}\left[\left(P_{s}\right)v_{k+1}^{2}e^{-\frac{\beta}{4}\vert v_{k+1}\vert^{2}}+\sum_{i=1}^{P_{s}}v_{i}^{2}e^{-\frac{\beta}{4P_{s}}\vert v_{i}\vert^{2}}\right]\nonumber \\
\leq & \epsilon^{\lambda}\Vert g_{0}\Vert_{\infty}(C\alpha t)^{2^{(s+1)}}2^{s(s+1)}.
\end{alignat}
\end{proof}
The next step is to consider an energy cutoff. We define 
\begin{equation}
\chi_{\lambda,E}^{\epsilon}=\chi\left\{ \prod_{k=1}^{P_{s}-1}B^{\epsilon}\leq\epsilon^{-\lambda}\right\} \chi\left\{ \vert\xi^{\epsilon}(0)\vert\leq2E\right\} 
\end{equation}
 and 
\begin{equation}
\chi_{\lambda,E}=\chi\left\{ \prod_{k=1}^{P_{s}-1}B\leq\epsilon^{-\lambda}\right\} \chi\left\{ \vert\xi(0)\vert\leq2E\right\} 
\end{equation}
The following estimate holds true
\begin{lem}
Let 
\begin{equation}
f_{1,s}^{N,Err_{3}}=\sum_{j_{1}=0}^{1}...\sum_{j_{s}=0}^{2^{s}-1}\left(\alpha\right)^{P_{s}-1}\sum_{\Gamma(P_{s}-1)}\sum_{\sigma_{P_{s}-1}}\boldsymbol{\sigma}_{P_{s}-1}\intop d\Lambda_{P_{s-1}}\chi_{\lambda,E}^{\epsilon}\prod_{k=1}^{P_{s}-1}B^{\epsilon}\,f_{0,P_{s}}^{N}(\zeta^{\epsilon}(0))
\end{equation}
 and let 
\begin{equation}
f_{1,s}^{\alpha,Err_{3}}=\sum_{j_{1}=0}^{1}...\sum_{j_{s}=0}^{2^{s}-1}\left(\alpha\right)^{P_{s}-1}\sum_{\Gamma(P_{s}-1)}\sum_{\sigma_{P_{s}-1}}\boldsymbol{\sigma}_{P_{s}-1}\intop d\Lambda_{P_{s-1}}\chi_{\lambda,E}\prod_{k=1}^{P_{s}-1}B\,f_{0,P_{s}}(\zeta(0))
\end{equation}
Then it results:
\begin{equation}
\Vert f_{1,s}^{N,Err_{2}}-f_{1,s}^{N,Err_{3}}\Vert_{\infty}+\Vert f_{1,s}^{\alpha,Err_{2}}-f_{1,s}^{\alpha,Err_{3}}\Vert_{\infty}\leq\Vert g_{0}\Vert_{\infty}e^{-\beta E^{2}}\left(C\alpha t\right)^{2^{s+1}}2^{s(s+1)}
\end{equation}
\end{lem}
\begin{proof}
We give a proof only for $\Vert f_{1,s}^{N,Err_{2}}-f_{1,s}^{N,Err_{3}}\Vert_{\infty}$,
the other one can be proved in the same way. We have that

\[
\vert f_{1,s}^{N,Err_{2}}-f_{1,s}^{N,Err_{3}}\vert\leq\Vert g_{0}\Vert_{\infty}\sum_{j_{1}=0}^{1}...\sum_{j_{s}=0}^{2^{s}-1}\left(C\alpha\right)^{P_{s}-1}\sum_{\Gamma(P_{s}-1)}\intop d\Lambda_{P_{s-1}}\,\prod_{k=1}^{P_{s}-1}B^{\epsilon}\,\chi\left\{ \vert\xi^{\epsilon}(0)\vert>2E\right\} e^{-\frac{\beta}{2}\vert\xi^{\epsilon}(0)\vert^{2}}
\]
\begin{equation}
\leq\Vert g_{0}\Vert_{\infty}e^{-\beta E^{2}}\left(C\alpha t\right)^{2^{s+1}}2^{s(s+1)}
\end{equation}
where we used that $e^{-\frac{\beta}{2}\vert\xi^{\epsilon}(0)\vert^{2}}\le e^{-\frac{\beta}{4}\vert\xi^{\epsilon}(0)\vert^{2}}e^{-\beta E^{2}}$.
\end{proof}
The next cutoff regards the time variables. We want to separate enough
the time between two creation, i.e. we want that $t_{i}-t_{i-1}>\delta$
$\forall\,0<i\leq P_{s}-1$. We define 
\begin{equation}
\chi_{\lambda,E,\delta}^{\epsilon}=\chi\left\{ \prod_{k=1}^{P_{s}-1}B^{\epsilon}\leq\epsilon^{-\lambda}\right\} \chi\left\{ \vert\xi^{\epsilon}(0)\vert\leq2E\right\} \chi\left\{ t_{i}-t_{i-1}>\delta,\,0<i\leq P_{s}-1\right\} 
\end{equation}
 and 
\begin{equation}
\chi_{\lambda,E,\delta}=\chi\left\{ \prod_{k=1}^{P_{s}-1}B\leq\epsilon^{-\lambda}\right\} \chi\left\{ \vert\xi(0)\vert\leq2E\right\} \chi\left\{ t_{i}-t_{i-1}>\delta,\,0<i\leq P_{s}-1\right\} 
\end{equation}
For the complementary set we have the following lemma
\begin{lem}
Let 
\begin{equation}
f_{1,s}^{N,Err_{4}}=\sum_{j_{1}=0}^{1}...\sum_{j_{s}=0}^{2^{s}-1}\left(\alpha\right)^{P_{s}-1}\sum_{\Gamma(P_{s}-1)}\sum_{\sigma_{P_{s}-1}}\boldsymbol{\sigma}_{P_{s}-1}\intop d\Lambda_{P_{s-1}}\chi_{\lambda,E,\delta}^{\epsilon}\prod_{k=1}^{P_{s}-1}B^{\epsilon}\,f_{0,P_{s}}^{N}(\zeta^{\epsilon}(0))
\end{equation}
 and let 
\begin{equation}
f_{1,s}^{\alpha,Err_{4}}=\sum_{j_{1}=0}^{1}...\sum_{j_{s}=0}^{2^{s}-1}\left(\alpha\right)^{P_{s}-1}\sum_{\Gamma(P_{s}-1)}\sum_{\sigma_{P_{s}-1}}\boldsymbol{\sigma}_{P_{s}-1}\intop d\Lambda_{P_{s-1}}\chi_{\lambda,E,\delta}\prod_{k=1}^{P_{s}-1}B\,f_{0,P_{s}}(\zeta(0))
\end{equation}
Then the following estimate holds

\begin{equation}
\Vert f_{1,s}^{N,Err_{3}}-f_{1,s}^{N,Err_{4}}\Vert_{\infty}+\Vert f_{1,s}^{\alpha,Err_{3}}-f_{1,s}^{\alpha,Err_{5}}\Vert_{\infty}\leq\epsilon^{-\lambda}\Vert g_{0}\Vert_{\infty}\left(C\alpha t\right)^{2^{s+1}}2^{(s+2)(s+1)}\frac{\delta}{t}
\end{equation}
\end{lem}
\begin{proof}
As the other lemma we give a proof only for the term $\Vert f_{1,s}^{N,Err_{3}}-f_{1,s}^{N,Err_{4}}\Vert_{\infty}$
since for the other one the proof is similar. We have that
\begin{equation}
\vert f_{1,s}^{N,Err_{3}}-f_{1,s}^{N,Err_{4}}\vert\leq\epsilon^{-\lambda}\Vert g_{0}\Vert_{\infty}\sum_{j_{1}=0}^{1}...\sum_{j_{s}=0}^{2^{s}-1}\left(C\alpha\right)^{P_{s}-1}\sum_{\Gamma(P_{s}-1)}\intop d\Lambda_{P_{s-1}}\left(\chi_{\lambda,E}^{\epsilon}-\chi_{\lambda,E,\delta}^{\epsilon}\right)e^{-\frac{\beta}{2}\vert\xi^{\epsilon}(0)\vert^{2}}
\end{equation}
There are $P_{s}-1$ choices of time variables such that $t_{i}-t_{i-1}\leq\delta$,
this gives us that 
\begin{equation}
\vert f_{1,s}^{N,Err_{3}}-f_{1,s}^{N,Err_{4}}\vert\leq\epsilon^{-\lambda}\Vert g_{0}\Vert_{\infty}\left(C\alpha t\right)^{2^{s+1}}2^{(s+2)(s+1)}\frac{\delta}{t}
\end{equation}
\end{proof}
Finally we introduce the last cutoff in the integrals. We define the
indicator function
\begin{equation}
\chi_{\lambda,E,\delta,q}^{\epsilon}=\chi_{\lambda,E,\delta}^{\epsilon}\chi\left\{ \vert\omega_{k}\cdot(w_{k+1}-\xi_{i_{k}}^{\epsilon}(\tau_{k}))\vert\geq\epsilon^{q},\,\vert\rho_{k}\vert\geq\epsilon^{q},\,1\leq k\leq P_{s}-1\right\} 
\end{equation}
 and 
\begin{equation}
\chi_{\lambda,E,\delta,q}=\chi_{\lambda,E,\delta}\chi\left\{ \vert\omega_{k}\cdot(w_{k+1}-\xi_{i_{k}}(\tau_{k}))\vert\geq\epsilon^{q},\,\vert\rho_{k}\vert\geq\epsilon^{q},\,1\leq k\leq P_{s}-1\right\} \label{eq:Cutoff}
\end{equation}
With this cutoff we are neglecting the grazing and the central velocities
in the creation of new particles. We have the following estimate:
\begin{lem}
\label{lem:Err5}Let 
\begin{equation}
f_{1,s}^{N,Err_{5}}=\sum_{j_{1}=0}^{1}...\sum_{j_{s}=0}^{2^{s}-1}\left(\alpha\right)^{P_{s}-1}\sum_{\Gamma(P_{s}-1)}\sum_{\sigma_{P_{s}-1}}\boldsymbol{\sigma}_{P_{s}-1}\intop d\Lambda_{P_{s-1}}\chi_{\lambda,E,\delta,q}^{\epsilon}\prod_{k=1}^{P_{s}-1}B^{\epsilon}\,f_{0,P_{s}}^{N}(\zeta^{\epsilon}(0))
\end{equation}
 and let 
\begin{equation}
f_{1,s}^{\alpha,Err_{5}}=\sum_{j_{1}=0}^{1}...\sum_{j_{s}=0}^{2^{s}-1}\left(\alpha\right)^{P_{s}-1}\sum_{\Gamma(P_{s}-1)}\sum_{\sigma_{P_{s}-1}}\boldsymbol{\sigma}_{P_{s}-1}\intop d\Lambda_{P_{s-1}}\chi_{\lambda,E,\delta,q}\prod_{k=1}^{P_{s}-1}B\,f_{0,P_{s}}(\zeta(0))
\end{equation}
Then:

\begin{equation}
\Vert f_{1,s}^{N,Err_{4}}-f_{1,s}^{N,Err_{5}}\Vert_{\infty}+\Vert f_{1,s}^{\alpha,Err_{4}}-f_{1,s}^{\alpha,Err_{5}}\Vert_{\infty}\leq\Vert g_{0}\Vert_{\infty}\epsilon^{\frac{q}{2}-\lambda}\left(C\alpha t\right)^{2^{s+1}}2^{(s+2)(s+1)}
\end{equation}
with $0<q<1$.
\end{lem}
\begin{proof}
We have that
\begin{equation}
\vert f_{1,s}^{N,Err_{4}}-f_{1,s}^{N,Err_{5}}\vert\leq\epsilon^{-\lambda}\Vert g_{0}\Vert_{\infty}\sum_{j_{1}=0}^{1}...\sum_{j_{s}=0}^{2^{s}-1}\left(C\alpha\right)^{P_{s}-1}\sum_{\Gamma(P_{s}-1)}\intop d\Lambda_{P_{s-1}}\left(\chi_{\lambda,E,\delta}^{\epsilon}-\chi_{\lambda,E,\delta,q}^{\epsilon}\right)e^{-\frac{\beta}{2}\vert\xi^{\epsilon}(0)\vert^{2}}
\end{equation}
This means that there exist a $k$ such that $\vert\omega_{k}\cdot(v_{k+1}-\xi_{i_{k}}^{\epsilon}(\tau_{k})\vert\leq\epsilon^{q}$
. If $\vert(v_{k+1}-\xi_{i_{k}}^{\epsilon}(\tau_{k})\vert\leq\epsilon^{\frac{q}{2}}$
then we simply have that 
\begin{equation}
\vert f_{1,s}^{N,Err_{4}}-f_{1,s}^{N,Err_{5}}\vert\leq\Vert g_{0}\Vert_{\infty}\epsilon^{\frac{q}{2}-\lambda}\left(C\alpha t\right)^{2^{s+1}}2^{(s+2)(s+1)}\label{eq:err1}
\end{equation}
Otherwise if $\vert(v_{k+1}-\xi_{i_{k}}^{\epsilon}(\tau_{k})\vert>\epsilon^{\frac{q}{2}}$
it results that $\vert\cos\gamma\vert\leq\epsilon^{\frac{q}{2}}$,
where $\gamma$ is the angle between $v_{k+1}-\xi_{i_{k}}^{\epsilon}(\tau_{k})$
and $\omega_{k}$. Therefore $\vert\frac{\pi}{2}-\gamma\vert\leq C\epsilon^{\frac{q}{2}}$
and, fixed $v_{k+1}-\xi_{i_{k}}^{\epsilon}(\tau_{k})$, $\omega_{k}$
must be in a set of measure bounded by $C\epsilon^{q}$. The case
$\rho_{k}\leq\epsilon^{q}$ can be easily estimated, since $d\nu_{k}=\rho_{k}d\rho_{k}d\psi$.
We have that 

\begin{equation}
\vert f_{1,s}^{N,Err_{4}}-f_{1,s}^{N,Err_{5}}\vert\leq\Vert g_{0}\Vert_{\infty}\epsilon^{q-\lambda}\left(C\alpha t\right)^{2^{s+1}}2^{(s+2)(s+1)}\label{eq:err2}
\end{equation}
From (\ref{eq:err1}) and (\ref{eq:err2}) we arrive to
\begin{equation}
\Vert f_{1,s}^{N,Err_{4}}-f_{1,s}^{N,Err_{5}}\Vert_{\infty}\leq\Vert g_{0}\Vert_{\infty}\epsilon^{\frac{q}{2}-\lambda}\left(C\alpha t\right)^{2^{s+1}}2^{(s+2)(s+1)}
\end{equation}
\end{proof}
We are now in position to estimate the difference between the BBF
and the IBF.

We define the following set 
\begin{equation}
N^{P_{s}}(\epsilon_{0})=\left\{ (\boldsymbol{t}_{P_{s}-1},\boldsymbol{\nu}_{P_{s}-1},\boldsymbol{w}_{P_{s}-1})\in\mathbb{R}^{P_{s}-1}\times S^{2(P_{s}-1)}\times\mathbb{R}^{3(P_{s}-1)}\,|\,\min_{i<k}\min_{\tau\in[0,t_{i-1}]}\,d(r_{i}(\tau),r_{k}(\tau))<\epsilon_{0}\right\} \label{eq:bset}
\end{equation}
where $d(\cdot,\cdot)$ denotes the distance over the torus $\Gamma$.
This set is completely defined via the BBF and it is the set of variables
for which a recollision can appear. At this point we need to prove
that the measure of the set $N^{P_{s}}(\epsilon_{0})$ is small, taking
into account also the constraints given by $\chi_{\lambda,E,\delta,q}^{\epsilon}$
and $\chi_{\lambda,E,\delta,q}$. This smallness is proved in \cite{Pulvirenti2014}
in the case of particles moving in the whole $\mathbb{R}^{3}$ instead
that in a torus. In the following lemma we adapt this result to the
present context by using also some geometrical estimate proved in
\cite{Bodineau2015}.
\begin{lem}
\label{lem:Ric}Let $\chi_{\lambda,E,\delta,q}$ be defined as in
(\ref{eq:Cutoff}) and let $\chi\left\{ N^{P_{s}}(\epsilon_{0})\right\} $
be the characteristic function of the set (\ref{eq:bset}). Then it
results that
\[
\sum_{j_{1}=0}^{1}...\sum_{j_{s}=0}^{2^{s}-1}\left(\alpha\right)^{P_{s}-1}\sum_{\Gamma(P_{s}-1)}\sum_{\sigma_{P_{s}-1}}\boldsymbol{\sigma}_{P_{s}-1}\intop d\Lambda_{P_{s-1}}\chi_{\lambda,E,\delta,q}\chi\left\{ N^{P_{s}}(\epsilon_{0})\right\} \prod_{k=1}^{P_{s}-1}B\,f_{0,P_{s}}^{N}(\zeta(0))\leq
\]
\begin{equation}
\Vert g_{0}\Vert_{\infty}\left(C\alpha t\right)^{2^{s+1}}E^{8}2^{(s+4)(s+1)}\left(\epsilon_{0}^{\frac{2}{5}-\lambda}+\frac{\epsilon_{0}^{\frac{4}{5}-\lambda}}{\delta^{2}}+\epsilon_{0}^{\frac{4}{5}-\lambda}\right)
\end{equation}
\end{lem}
We leave the proof of this lemma in the appendix II.

Thanks to these estimates we can now give a proof of the convergence
of the IBF to the BBF and then of the one particle marginal of the
GH to the solution of the Boltzmann equation. First we choose the
magnitude of the parameters in the following way
\begin{equation}
\alpha\cong C\left(\log\log N\right)^{\frac{1}{2}}\,s\cong\frac{\log\log N}{2\log2}
\end{equation}
Furthermore we have that 
\begin{equation}
2^{s+1}\leq2\left(\log N\right)^{\frac{1}{2}}
\end{equation}
\begin{equation}
2^{(s+2)(s+1)}\leq2\left(\log N\right)^{\frac{\log\log N}{2\log2}}
\end{equation}
\begin{equation}
\left(C\alpha t\right)^{2^{s+1}}\leq\left(C\log\log N\right)^{\sqrt{\log N}}
\end{equation}
\begin{equation}
N\epsilon^{2}\leq C\left(\log\log N\right)^{\frac{1}{2}}
\end{equation}
\begin{equation}
\epsilon\leq C\frac{\left(\log\log N\right)^{\frac{1}{4}}}{N^{\frac{1}{2}}}
\end{equation}
We also set $\epsilon_{0}=\epsilon^{\frac{5}{6}}$, $\delta=\epsilon^{\frac{1}{8}}$,
$E=\frac{\sqrt{\log N}}{\beta}$ and we fix $q=\frac{1}{8}$ and $\lambda=\frac{1}{32}$.
We have the following theorem 
\begin{thm}
\label{thm:ricollision}Let $\widetilde{f_{1}^{N}}(t)$ be the one
particle marginal of the Grad hierarchy with initial datum as (\ref{eq:Initialpertu})
and let $f_{1}^{\alpha}(t)$ be the solution of the Boltzmann equation
with initial datum as (\ref{eq:initialBH}). Then $\forall t\in[0,T]$
it results that 
\begin{equation}
\Vert\widetilde{f_{1}^{N}}(t)-f_{1}^{\alpha}(t)\Vert_{\infty}\rightarrow0
\end{equation}
for $N\rightarrow\infty$, $\epsilon\rightarrow0$, $\vert N\epsilon^{2}-\alpha\vert\leq\epsilon$.
\end{thm}
\begin{proof}
We have
\begin{equation}
\Vert\widetilde{f_{1}^{N}}(t)-f_{1}^{\alpha}(t)\Vert_{\infty}\leq\Vert\widetilde{f_{1,s}^{N}}(t)-f_{1,s}^{N}(t)\Vert_{\infty}+\Vert f_{1,s}^{N}(t)-f_{1,s}^{\alpha}(t)\Vert_{\infty}+\Vert\widetilde{R_{N}^{s}}(t)\Vert_{\infty}+\Vert R^{s}(t)\Vert_{\infty}
\end{equation}
 From Theorems \ref{thm:Rem} and \ref{thm:ghih} it results that
\begin{equation}
\Vert\widetilde{f_{1,s}^{N}}(t)-f_{1,s}^{N}(t)\Vert_{\infty}\leq\Vert g_{0}\Vert_{\infty}2^{s(s+1)}\epsilon\left(C\alpha t\right)^{2^{s+1}}\leq\Vert g_{0}\Vert_{\infty}\left(C\log\log N\right)^{\sqrt{\log N}}\left(\log N\right)^{\frac{\log\log N}{4\log2}}\frac{\left(\log\log N\right)^{\frac{1}{4}}}{N^{\frac{1}{2}}}
\end{equation}
\begin{equation}
\Vert\widetilde{R_{N}^{s}}(t)\Vert_{\infty}+\Vert R^{s}(t)\Vert_{\infty}\leq\Vert g_{0}\Vert_{\infty}\left(\frac{C\left(\alpha t\right)^{2}}{s}\right)^{2}\leq\frac{C\Vert g_{0}\Vert_{\infty}}{\log\log N}
\end{equation}
We have to work on the term $\Vert f_{1,s}^{N}(t)-f_{1,s}^{\alpha}(t)\Vert_{\infty}$.
First it results that 

\begin{alignat}{1}
\Vert f_{1,s}^{N}(t)-f_{1,s}^{\alpha}(t)\Vert_{\infty}\leq & \sum_{l=1}^{5}\Vert f_{1,s}^{N,Err_{l-1}}-f_{1,s}^{N,Err_{l}}\Vert_{\infty}+\sum_{l=0}^{5}\Vert f_{1,s}^{\alpha,Err_{l-1}}(t)-f_{1,s}^{\alpha,Err_{l}}\Vert_{\infty}\nonumber \\
+ & \Vert f_{1,s}^{N,Err_{5}}(t)-f_{1,s}^{\alpha,Err_{5}}(t)\Vert_{\infty}
\end{alignat}
where $f_{1,s}^{N,Err_{0}}(t)=f_{1,s}^{N}(t)$ and $f_{1,s}^{\alpha,Err_{0}}(t)=f_{1,s}^{\alpha}(t)$.
We focus on the last term, it results that 
\begin{alignat}{1}
\vert f_{1,s}^{N,Err_{5}}(t)-f_{1,s}^{\alpha,Err_{5}}(t)\vert\leq & \sum_{j_{1}=0}^{1}...\sum_{j_{s}=0}^{2^{s}-1}\left(C\alpha\right)^{P_{s}-1}\sum_{\sigma_{P_{s}-1}}\sum_{\Gamma(P_{s}-1)}\nonumber \\
 & \sum_{\Gamma(P_{s}-1)}\intop d\Lambda_{P_{s-1}}\,\vert\chi_{\lambda,E,\delta,q}^{\epsilon}\,f_{0,P_{s}}^{N}(\zeta^{\epsilon}(0))-\chi_{\lambda,E,\delta,q}\,f_{0,P_{s}}(\zeta(0))\vert
\end{alignat}
Now we split the integrals by using the indicator functions $1-\chi\left\{ N^{P_{s}}(\epsilon_{0})\right\} $
and $\chi\left\{ N^{P_{s}}(\epsilon_{0})\right\} $. In the first
case since we are outside the set $N^{P_{s}}(\epsilon_{0})$ the particles
must be at a distance greater than $\epsilon_{0}$, this implies that
$M_{N,\beta}(\boldsymbol{z}_{N})=C_{N,\beta}e^{-\frac{\beta}{2}\vert\boldsymbol{v}_{N}\vert^{2}}$
and that $\chi_{\lambda,E,\delta,q}^{\epsilon}=\chi_{\lambda,E,\delta,q}$.
Then we have that 
\[
\intop d\Lambda_{P_{s-1}}\,\left(1-\chi\left\{ N^{P_{s}}(\epsilon_{0})\right\} \right)\chi_{\lambda,E,\delta,q}\,\vert f_{0,P_{s}}^{N}(\zeta^{\epsilon}(0))-f_{0,P_{s}}(\zeta(0))\vert\leq
\]
\begin{equation}
\intop d\Lambda_{P_{s-1}}\,\left(1-\chi\left\{ N^{P_{s}}(\epsilon_{0})\right\} \right)\chi_{\lambda,E,\delta,q}\,\left[\vert f_{0,P_{s}}^{N}(\zeta^{\epsilon}(0))-f_{0,P_{s}}(\zeta^{\epsilon}(0))\vert+\vert f_{0,P_{s}}(\zeta^{\epsilon}(0))-f_{0,P_{s}}(\zeta(0))\vert\right]
\end{equation}
From the definition of the initial datum it turns out that 
\begin{equation}
\vert f_{0,P_{s}}^{N}(\zeta^{\epsilon}(0))-f_{0,P_{s}}(\zeta^{\epsilon}(0))\vert\leq\Vert g_{0}\Vert_{\infty}\vert C_{P_{s},\beta}-C_{\beta}^{P_{s}}\vert
\end{equation}
A straightforward calculation from the definition (\ref{eq:cbet1})
and (\ref{eq:cbet2}) gives us that 
\begin{equation}
\vert C_{P_{s},\beta}-C_{\beta}^{P_{s}}\vert\leq2^{2(s+1)}\epsilon^{3}
\end{equation}
Moreover outside the set $N^{P_{s}}(\epsilon_{0})$ the velocities
of the BBF and of the IBF are the same and also $p_{1}^{\epsilon}(s)=p_{1}(s)\,0\leq s\leq t$,
it follows that 
\begin{equation}
\vert f_{0,P_{s}}(\zeta^{\epsilon}(0))-f_{0,P_{s}}(\zeta(0))\vert=\vert C_{N,\beta}e^{-\frac{\beta}{2}\vert\xi^{\epsilon}(0)\vert^{2}}g_{0}(p_{1}^{\epsilon}(0),\xi_{1}^{\epsilon}(0))-C_{N,\beta}e^{-\frac{\beta}{2}\vert\xi(0)\vert^{2}}g_{0}(r_{1}^{\epsilon}(0),\xi_{1}^{\epsilon}(0))\vert=0
\end{equation}
Finally we have that
\[
\sum_{j_{1}=0}^{1}...\sum_{j_{s}=0}^{2^{s}-1}\left(C\alpha\right)^{P_{s}-1}\sum_{\sigma_{P_{s}-1}}\sum_{\Gamma(P_{s}-1)}\intop d\Lambda_{P_{s-1}}\,\left(1-\chi\left\{ N^{P_{s}}(\epsilon_{0})\right\} \right)\chi_{\lambda,E,\delta,q}\,\vert\,f_{0,P_{s}}^{N}(\zeta^{\epsilon}(0))-\,f_{0,P_{s}}(\zeta(0))\vert\leq
\]

\begin{equation}
\left(C\alpha t\right)^{2^{s+1}}2^{2s(s+1)}\epsilon^{3-\lambda}
\end{equation}
In the second case we use the estimates of Lemma \ref{lem:Ric} to
obtain that
\[
\sum_{j_{1}=0}^{1}...\sum_{j_{s}=0}^{2^{s}-1}\left(C\alpha\right)^{P_{s}-1}\sum_{\sigma_{P_{s}-1}}\sum_{\Gamma(P_{s}-1)}\intop d\Lambda_{P_{s-1}}\,\chi\left\{ N^{P_{s}}(\epsilon_{0})\right\} \,\vert\chi_{\lambda,E,\delta,q}^{\epsilon}\,f_{0,P_{s}}^{N}(\zeta^{\epsilon}(0))-\chi_{\lambda,E,\delta,q}\,f_{0,P_{s}}^{\alpha}(\zeta(0))\vert
\]
\[
\leq\epsilon^{-\lambda}\Vert g_{0}\Vert_{\infty}\left(C\alpha t\right)^{2^{s+1}}E^{8}2^{(s+4)(s+1)}\left(\epsilon_{0}^{\frac{2}{5}}+\frac{\epsilon_{0}^{\frac{4}{5}}}{\delta^{2}}+\epsilon_{0}^{\frac{4}{5}}\right)\leq
\]
\begin{equation}
\Vert g_{0}\Vert_{\infty}\epsilon^{\frac{1}{20}}\left(C\log\log N\right)^{\sqrt{\log N}}\left(\log N\right)^{4\log\log N}
\end{equation}
We have proved that 
\begin{equation}
\Vert f_{1,s}^{N,Err_{5}}(t)-f_{1,s}^{\alpha,Err_{5}}(t)\Vert_{\infty}\leq\Vert g_{0}\Vert_{\infty}\epsilon^{\frac{1}{20}}\left(C\log\log N\right)^{\sqrt{\log N}}\left(\log N\right)^{4\log\log N}
\end{equation}
The remainders can be easily handled with the estimates proved in
Lemmas \ref{lem:Err1}-\ref{lem:Err5}. It follows that
\[
\sum_{l=1}^{5}\Vert f_{1,s}^{N,Err_{l-1}}-f_{1,s}^{N,Err_{l}}\Vert_{\infty}+\sum_{l=0}^{5}\Vert f_{1,s}^{\alpha,Err_{l-1}}(t)-f_{1,s}^{\alpha,Err_{l}}\Vert_{\infty}\leq\Vert g_{0}\Vert_{\infty}\left(C\alpha t\right)^{2^{s+1}}2^{(s+2)(s+1)}
\]
\begin{equation}
\left(\epsilon^{\frac{q}{2}-\lambda}+\epsilon^{\lambda}+\epsilon^{-\lambda}\delta+e^{-\beta E^{2}}\right)\leq\Vert g_{0}\Vert_{\infty}\left(C\log\log N\right)^{\sqrt{\log N}}\left(\log N\right)^{\frac{\log\log N}{\log2}}\left(\epsilon^{\frac{1}{32}}+\frac{1}{N}\right)
\end{equation}
Summarizing, we have that 
\[
\Vert\widetilde{f_{1}^{N}}(t)-f_{1}^{\alpha}(t)\Vert_{\infty}\leq\Vert\widetilde{f_{1,s}^{N}}(t)-f_{1,s}^{N}(t)\Vert_{\infty}+\Vert f_{1,s}^{N}(t)-f_{1,s}^{\alpha}(t)\Vert_{\infty}+\Vert\widetilde{R_{N}^{s}}(t)\Vert_{\infty}+\Vert R^{s}(t)\Vert_{\infty}\leq
\]
\[
\frac{C\Vert g_{0}\Vert_{\infty}}{\log\log N}+\Vert g_{0}\Vert_{\infty}\left(C\log\log N\right)^{\sqrt{\log N}}\left(\log N\right)^{\left(\frac{\log\log N}{\log2}\right)^{3}}\left[\epsilon^{\frac{1}{20}}+\frac{1}{N}\right]
\]
If we send $N\rightarrow\infty$, $\epsilon\rightarrow0$ with $N\epsilon^{2}\cong C\left(\log\log N\right)^{\frac{1}{2}}$
we obtain the proof of the theorem. \pagebreak{}
\end{proof}

\section{From Linear Boltzmann to Linear Landau}

\subsection{Existence of semigroups}

In this section we want to prove that the solution of the Linear Boltzmann
equation converges as $\alpha\rightarrow\infty$ to the solution of
the Linear Landau equation. For this purpose we rewrite in the following
way the linear Boltzmann and Landau equations 

\begin{equation}
\begin{cases}
\partial_{t}f=G_{\alpha}(f)\\
f(x,v,0)=f_{0}(x,v)
\end{cases}
\end{equation}

\begin{equation}
\begin{cases}
\partial_{t}f=G(f)\\
f(x,v,0)=f_{0}(x,v)
\end{cases}
\end{equation}
where 
\begin{equation}
G_{\alpha}(f)=Q_{B}(f)-v\cdot\nabla_{x}f\label{eq:groupB}
\end{equation}
 and 
\begin{equation}
G(f)=Q_{L}(f)-v\cdot\nabla_{x}f.\label{eq:groupL}
\end{equation}

Now we want to set the problem in the Hilbert space $\mathbf{H}=L^{2}\left(\Gamma\times\mathbb{R}^{3},dxd\mu\right)$
where $d\mu=M_{\beta}(v)dv$. This space arises naturally from the
definition of the operators $G$ and $G_{\alpha}$. Indeed, we have
that $G_{\alpha}$ and $G$ are unbounded linear operators densely
defined respectively on $D(G_{\alpha})=H^{1}(\Gamma,dx)\times L^{2}(\mathbb{R}^{3},d\mu)$
and $D(G)=H^{1}(\Gamma,dx)\times H^{2}(\mathbb{R}^{3},d\mu)$, where
$H^{1}$ and $H^{2}$ denote the usual Sobolev spaces.

The main motivation to introduce $\mathbf{H}$ is the following lemma:
\begin{lem}
\label{lem:OperPorp}The operators \textup{$Q_{B}(f)$ and $Q_{L}(f)$
are well defined as }self-adjoint operators\textup{ on $L^{2}(\mathbb{R}^{3},d\mu)$
and $H^{2}(\mathbb{R}^{3},d\mu)$ respectively. }Moreover for the
operators $G$ and $G_{\alpha}$, defined in (\ref{eq:groupB}) and
(\ref{eq:groupL}), we have that $\forall f\in\mathbf{H}$ and $\forall g\in D(G)$
\begin{equation}
\left(G_{\alpha}^{*}f,f\right)=(f,G_{\alpha}f)\leq0\label{eq:Aa1}
\end{equation}
and 
\begin{equation}
(G^{*}g,g)=(g,Gg)\leq0\label{eq:Aa2}
\end{equation}
i.e. $G_{\alpha}$ and $G$ are dissipative operators. Furthermore
$G_{\alpha}$ and $G$ are closed operators.
\end{lem}
We gives the proof of this lemma in the appendix I. 

Thanks to these properties of the operators we can use the following
theorem.
\begin{thm}
(\cite{Engel}) Let $A$ be a linear operator densely defined on a
linear subspace $D(A)$ of the Hilbert space $\mathbf{H}$. If both
$A$ and $A^{*}$ are dissipative operators then $\overline{A}$ generate
a contraction semigroup on $\mathbf{H}$.
\end{thm}
This theorem ensures the existence of $T_{\alpha}(t)$ and $T(t)$,
the semigroups with infinitesimal generator given by $G_{\alpha}$
and $G$ respectively. Indeed, from Lemma \ref{lem:OperPorp} we have
that $G_{\alpha}$ and $G$ are closed operators and that $G_{\alpha}^{*}$
and $G^{*}$ are dissipative operators, then we have the existence
of $T_{\alpha}(t)$ and $T(t)$. 

\subsection{Convergence of the semigroups}

The last step of our proof is to show that the semigroup generated
by $G_{\alpha}(f)$ strongly converges to the semigroup generated
by $G(f)$ in the limit $\alpha\rightarrow0$. We use the following
theorem, that gives necessary and sufficient conditions for the convergence.
\begin{thm}
\label{thm:(Trotter-Kato)}(Trotter-Kato). Let $A$ and $A_{n}$ be
the generators of the contraction semigroups $T(t)$ and $T_{n}(t)$
respectively. Let $D$ be a core for $A$. Suppose that $D\subseteq D(A_{n})\,\,\forall n$
and that $\forall f\in D$ $A_{n}f\rightarrow Af$. Then 
\begin{equation}
\Vert T_{n}f-Tf\Vert_{H}\rightarrow0\,\,\,\,as\,\,n\rightarrow+\infty
\end{equation}
$\forall f\in\mathbf{H}$ and uniformly for $t\in[0,T]$ for any $T>0$.
\end{thm}
A proof of this theorem can be found in \cite{Engel}. 

We want to apply this theorem to prove that $T_{\alpha}f\rightarrow Tf$.
We note that $D=C_{p}^{\infty}(\Gamma)\times C_{0}^{\infty}(\mathbb{R}^{3})$
is a core for $G_{\alpha}$ and $G$ as follows by a direct insepction.
Then we use the steps of section 3.2 to prove the strong convergence
of the operators on this set.
\begin{thm}
\label{thm:conv}Let $G_{\alpha}$ and $G$ be defined as in (\ref{eq:groupB})
and (\ref{eq:groupL}). Then $\forall f\in D$ it results that
\[
\Vert\left(G_{\alpha}-G\right)f\Vert_{\mathbf{H}}\underset{\alpha\rightarrow\infty}{\longrightarrow}0
\]
\end{thm}
\begin{proof}
First we define the following operator 
\begin{equation}
Q_{B}^{c}(f)=\alpha\intop dv_{1}M_{\beta}(v_{1})\intop_{\nu\cdot V>0}d\nu\vert\nu\cdot V\vert\chi\left\{ \vert V\vert\geq\alpha^{-\frac{4}{15}}\right\} \left[f(v^{'})-f(v)\right]
\end{equation}
This is the Linear Boltzmann operator with a $\alpha\text{-depending}$
cutoff on the small relative velocities. Observe that $Q_{B}^{c}$
and $Q_{B}$ are asymptotically equivalent as $\alpha\rightarrow\infty$.
Indeed, we have that $\forall f\in D_{0}$

\begin{alignat}{1}
\Vert\left(Q_{B}^{c}-Q_{B}\right)f\Vert_{\mathbf{H}}^{2}= & \intop dx\intop d\mu(v)\left|\alpha\intop d\mu(v_{1})\intop_{\nu\cdot V>0}\chi\left\{ \vert V\vert<\alpha^{-\frac{4}{15}}\right\} \vert\nu\cdot V\vert\left[f(x,v^{'})-f(x,v)\right]d\nu\right|^{2}\nonumber \\
\leq & C\Vert f\Vert_{\infty}^{2}\intop dxd\mu(v)\left|\intop dv_{1}\alpha^{\frac{11}{15}}M_{\beta}(v_{1})\chi\left\{ \vert V\vert<\alpha^{-\frac{4}{15}}\right\} \right|^{2}\nonumber \\
\leq & C\Vert f\Vert_{\infty}^{2}\intop dxd\mu(v)\left(\alpha^{\frac{11}{15}}\intop\chi\left\{ \vert V\vert<\alpha^{-\frac{4}{15}}\right\} dv_{1}\right)^{2}
\end{alignat}
Since 
\begin{equation}
\intop\chi\left\{ \vert V\vert<\alpha^{-\frac{4}{15}}\right\} dv_{1}\leq C\alpha^{-\frac{12}{15}},
\end{equation}
we arrive to
\begin{equation}
\Vert\left(Q_{B}^{c}-Q_{B}\right)f\Vert^{2}\leq C\alpha^{-\frac{1}{15}}.
\end{equation}
We put the same cutoff on the operator $Q_{L}$ and we define 
\begin{equation}
Q_{L}^{c}=A\intop dv_{1}M_{\beta}(v_{1})\frac{1}{\vert V\vert^{3}}\left[\vert V\vert^{2}\triangle f(v)-\left(V,D^{2}V\right)-4V\cdot\nabla_{v}f(v)\right]\left\{ \vert V\vert\geq\alpha^{-\frac{4}{15}}\right\} 
\end{equation}
Then $\forall f\in D$ we have
\begin{alignat}{1}
\Vert\left(Q_{L}^{c}-Q_{L}\right)f\Vert_{\mathbf{H}}^{2}= & \intop dxd\mu\left|\intop dv_{1}\frac{M_{\beta}(v_{1})A}{\vert V\vert^{3}}\left[\vert V\vert^{2}\triangle f(v)-\left(V,D^{2}V\right)-4V\cdot\nabla_{v}f(v)\right]\left\{ \vert V\vert\leq\alpha^{-\frac{4}{15}}\right\} \right|^{2}\nonumber \\
\leq & C(A,f)\alpha^{-\frac{8}{15}}
\end{alignat}
Now we want to prove that $Q_{B}^{c}$ converges strongly to $Q_{L}^{c}$
when $\alpha\rightarrow+\infty$. We have that for all $f\in D$ 
\[
\Vert\left(Q_{B}^{c}-Q_{L}^{C}\right)f\Vert_{\mathbf{H}}^{2}=\intop dx\intop d\mu(v)
\]
\[
\left|\intop d\mu(v_{1})\left\{ \intop_{\nu\cdot\mathrm{V}>0}\alpha\,\chi\left\{ \vert V\vert\geq\alpha^{-\frac{4}{15}}\right\} \vert v\cdot V\vert\left[f(x,v^{'})-f(x,v)\right]d\nu-\right.\right.
\]
\begin{equation}
\left.\left.\frac{A}{\vert V\vert^{3}}\left[\vert V\vert^{2}\triangle_{v}f(x,v)-\left(V,D_{v}^{2}(f(x,v))V\right)-4\mathrm{V}\cdot\nabla_{v}f(x,v)\right]\right\} \chi\left\{ \vert V\vert\geq\alpha^{-\frac{4}{15}}\right\} \right|^{2}
\end{equation}
We perform the same steps of section 3 to obtain: 
\[
\Vert\left(Q_{B}^{c}-Q_{L}^{c}\right)f\Vert_{\mathbf{H}}^{2}\leq
\]
\begin{equation}
C\intop dx\intop d\mu(v)\left|\intop d\mu(v_{1})\intop_{\nu\cdot V>0}\alpha\,\chi\left\{ \vert V\vert\geq\alpha^{-\frac{4}{15}}\right\} \vert\nu\cdot V\vert o(\alpha^{-1})\right|^{2}\label{eq:group1}
\end{equation}
For the second term we have to go further in the Taylor expansion
and use the Lagrange form for the remainder term. From Lemma (\ref{lem:theta})
it results that 
\begin{equation}
o(\alpha^{-1})=\frac{M^{2}(\rho,\alpha)}{\vert V\vert^{8}\alpha^{2}}+\frac{\theta^{3}}{3!}f^{'''}(\xi)
\end{equation}
 for a certain $\xi\in\left[0,\theta\right]$ . Therefore 
\[
\intop dx\intop d\mu(v)\left|\intop d\mu(v_{1})\intop_{\nu\cdot V>0}\alpha\,\chi\left\{ \vert V\vert\geq\alpha^{-\frac{1}{3}}\right\} \vert\nu\cdot V\vert o(\alpha^{-1})\right|^{2}\leq
\]
\begin{equation}
\intop dx\intop d\mu(v)\left|\intop d\mu(v_{1})\intop_{\nu\cdot\mathrm{V}>0}\alpha\,\chi\left\{ \vert V\vert\geq\alpha^{-\frac{4}{15}}\right\} \vert\nu\cdot V\vert\left[\frac{M^{2}(\rho,\alpha)}{\vert V\vert^{8}\alpha^{2}}+\frac{\theta^{3}}{3!}f^{'''}(\xi)\right]\right|^{2}
\end{equation}
Thanks to formula (\ref{eq:Thetaest}) we have that 
\begin{equation}
\vert\theta^{3}(\rho,\alpha)\vert\leq C\left(\alpha^{-\frac{3}{2}}\frac{\gamma^{3}(\rho)}{\vert V\vert^{6}}+\alpha^{-3}\frac{M^{3}(\rho,\alpha)}{\vert V\vert^{12}}\right)
\end{equation}
Furthermore from (\ref{eq:coord}) it follows that 
\begin{equation}
\vert f^{'''}(\xi)\vert\leq C(f)\vert V\vert
\end{equation}
and then we can write 
\begin{align}
\Vert\left(Q_{B}^{c}-Q_{L}\right)f\Vert_{\mathbf{H}}^{2}\leq & C_{1}(f,\gamma,M)\intop dx\intop d\mu(v)\nonumber \\
 & \left[\intop d\mu(v_{1})\left(\frac{\alpha^{-1}}{\vert V\vert^{6}}+\frac{\alpha^{-\frac{1}{2}}}{\vert V\vert^{4}}+\frac{\alpha^{-2}}{\vert V\vert^{10}}\right)\chi\left\{ \vert V\vert\geq\alpha^{-\frac{4}{15}}\right\} \right]^{2}\label{eq:group2}
\end{align}
A change of variables on the right hand side of (\ref{eq:group2})
gives us 
\begin{equation}
\intop d\mu(v_{1})\left(\frac{\alpha^{-1}}{\vert V\vert^{6}}+\frac{\alpha^{-\frac{1}{2}}}{\vert V\vert^{4}}+\frac{\alpha^{-2}}{\vert V\vert^{10}}\right)\chi\left\{ \vert V\vert\geq\alpha^{-\frac{1}{3}}\right\} \leq C\intop_{\alpha^{-\frac{4}{15}}}^{\infty}dr\left(\frac{\alpha^{-1}}{r^{4}}+\frac{\alpha^{-\frac{1}{2}}}{r^{2}}+\frac{\alpha^{-2}}{r^{8}}\right)\leq C\alpha^{\frac{1}{3}}\label{eq:group3}
\end{equation}
From formula (\ref{eq:group2}) and (\ref{eq:group3}) we have 
\begin{equation}
\Vert\left(Q_{B}^{c}-Q_{L}\right)f\Vert_{\mathbf{H}}^{2}\leq C_{2}(\gamma M)\alpha^{-\frac{2}{15}}
\end{equation}
Then we have 
\begin{equation}
\Vert\left(Q_{B}-Q_{L}\right)f\Vert_{\mathbf{H}}^{2}\leq C\alpha^{-\frac{1}{15}}
\end{equation}
and this proves our theorem.
\end{proof}
Finally we use Theorem \ref{thm:(Trotter-Kato)} and Theorem \ref{thm:conv}
to prove that the solution of the linear Boltzmann equation converge
to the solution of the linear Landau equation.
\begin{thm}
\label{thm:hconv}Let $g^{\alpha}(x,v,t)$ be the solution of the
linear Boltzmann equation and let $g(x,v,t)$ be the solution of the
linear Landau equation. Suppose that the initial datum of both equations
is given by $g_{0}(x,v)$. Then it results that 
\begin{equation}
\Vert g^{\alpha}(x,v)-g(x,v)\Vert_{\mathbf{H}}\rightarrow0
\end{equation}
when $\alpha\rightarrow0$.
\end{thm}
\begin{proof}
Since $g^{\alpha}(t)=T_{\alpha}(t)g_{0}(x,v)$ and $g(t)=T(t)g_{0}(x,v)$
we have that 
\begin{equation}
\Vert g^{\alpha}(x,v)-g(x,v)\Vert_{\mathbf{H}}=\Vert T_{\alpha}(t)g_{0}(x,v)-T(t)g_{0}(x,v)\Vert_{\mathbf{H}}\label{eq:group5}
\end{equation}
From Theorem \ref{thm:(Trotter-Kato)} and Theorem \ref{thm:conv}
we have that the right hand side of (\ref{eq:group5}) goes to zero
when $\alpha\rightarrow0$ and the theorem is proved. 
\end{proof}
\pagebreak{}

\section{Proof of the main theorem}

In this section we summarize all the estimates obtained and we finally
give a proof that the solution of the first equation of the Grad hierarchy
converge to the solution of the linear Landau equation in the scaling
$N\epsilon^{2}\rightarrow\alpha$ with $\alpha\cong C\left(\log\log N\right)^{\frac{1}{2}}$.
\begin{thm}
Let $\overline{f_{1}^{N}}(t)$ be the first-particle marginal of the
solution of the Liouville equation with initial datum given by $W_{0,N}(\boldsymbol{z}_{N})=M_{N,\beta}(\boldsymbol{z}_{N})g_{0}(x_{1},v_{1})$,
and let $g(t)$ be the solution of the linear Landau equation with
initial datum given by $g(x,v,0)=g_{0}(x,v)$. Then $\forall t>0$
\[
\Vert\overline{f_{1}^{N}}(x,v,t)-M_{\beta}(v)g(x,v,t)\Vert_{\mathbf{H}}\rightarrow0
\]
when $N\rightarrow\infty$, with $N\epsilon^{2}\cong\left(\log\log N\right)^{\frac{1}{2}}$.
\end{thm}
\begin{proof}
First we want to estimate the following difference 
\begin{equation}
\Vert\widetilde{f_{1}^{N}}(t)-f_{1}^{\alpha}(t)\Vert_{\mathbf{H}}
\end{equation}
Since $\intop dx\intop d\mu(v)\vert\widetilde{f_{1}^{N}}(t)-f_{1}^{\alpha}(t)\vert^{2}\leq C\Vert\widetilde{f_{1}^{N}}(t)-f_{1}^{\alpha}(t)\Vert_{\infty}^{2}$
it results that 
\begin{equation}
\Vert\widetilde{f_{1}^{N}}(t)-f_{1}^{\alpha}(t)\Vert_{\mathbf{H}}\leq C\Vert\widetilde{f_{1}^{N}}(t)-f_{1}^{\alpha}(t)\Vert_{\infty}
\end{equation}
From Theorem \ref{thm:ricollision} we have that 
\begin{equation}
\Vert\widetilde{f_{1}^{N}}(t)-f_{1}^{\alpha}(t)\Vert_{\infty}\rightarrow0
\end{equation}
and then 
\begin{equation}
\Vert\widetilde{f_{1}^{N}}(t)-f_{1}^{\alpha}(t)\Vert_{\mathbf{H}}\rightarrow0\label{eq:au1}
\end{equation}
As we have seen the solution of the first equation of the BH with
initial data given by (\ref{eq:initialBH}) has the form
\begin{equation}
f_{1}^{\alpha}(x,v,t)=M_{\beta}(v)g^{\alpha}(x,v,t)
\end{equation}
 where $g^{\alpha}(x,v)$ is the solution of the Linear Boltzmann
equation. Then we have proved that 
\begin{equation}
\Vert\widetilde{f_{1}^{N}}(t)-M_{\beta}(v)g^{\alpha}(x,v,t)\Vert_{\mathbf{H}}\rightarrow0
\end{equation}
Since 
\begin{equation}
\Vert M_{\beta}(v)g^{\alpha}(x,v,t)-M_{\beta}(v)g(x,v,t)\Vert_{\mathbf{H}}\leq C\Vert g^{\alpha}(x,v)-g(x,v,t)\Vert_{\mathbf{H}}\label{eq:au2}
\end{equation}
Thanks to Theorem (\ref{thm:conv}) we have that 
\begin{equation}
\Vert g^{\alpha}(x,v,t)-g(x,v,t)\Vert_{\mathbf{H}}\rightarrow0
\end{equation}
From (\ref{eq:au2}) and (\ref{eq:au1}) we arrive to 
\begin{equation}
\Vert\widetilde{f_{1}^{N}}(x,v,t)-M_{\beta}(v)g(x,v,t)\Vert_{\mathbf{H}}\leq\Vert\widetilde{f_{1}^{N}}(t)-f_{1}^{\alpha}(t)\Vert_{\mathbf{H}}+C\Vert g^{\alpha}(x,v,t)-g(x,v,t)\Vert_{\mathbf{H}}
\end{equation}
Finally we estimate the difference between the reduced marginal and
the standard marginal. We have
\begin{alignat}{1}
\vert\widetilde{f_{1}^{N}}(x,v,t)-\overline{f_{1}^{N}}(x,v,t)\vert= & \vert\intop dz_{j+1}...dz_{N}\,W_{N}(\boldsymbol{z}_{N},t)\left(1-\chi\left\{ S(x_{1})^{N-1}\right\} \right)\vert\nonumber \\
\leq & N\vert\intop_{\vert x-x_{1}\vert\leq\epsilon}dx_{1}dv_{1}\overline{f_{2}^{N}}(x,v,x_{1},v_{1},t)\vert\nonumber \\
\leq & CN\epsilon^{3}
\end{alignat}
Then it results
\begin{equation}
\Vert\widetilde{f_{1}^{N}}(x,v,t)-\overline{f_{1}^{N}}(x,v,t)\Vert_{\infty}
\end{equation}

We send $\alpha\rightarrow\infty$, $N\rightarrow\infty$ with $N\epsilon^{2}\cong\alpha\cong\left(\log\log N\right)^{\frac{1}{2}}$
and we obtain the proof of the theorem.
\end{proof}
\medskip{}

\begin{center}
\textit{Acknowledgements. }
\par\end{center}

I thank M.Pulvirenti for having suggested the problem and for many
useful discussions.

\pagebreak{}

\section{Appendix I, Proof of Lemma \ref{lem:OperPorp}}

Here we gives a proof of the Lemma \ref{lem:OperPorp}.
\begin{proof}
First we want to prove that the operator $Q_{L}$ is self-adjoint.
It results that 
\begin{equation}
\left(f,Q_{L}(g)\right)_{L^{2}(d\mu)}=\intop dv\intop dwM_{\beta}(v)M_{\beta}(w)\frac{1}{\vert V\vert^{3}}\left[\vert V\vert^{2}\triangle g-\left(V,D_{v}^{2}(g)V\right)-4V\cdot\nabla_{v}g\right]f(v)
\end{equation}
We integrate by parts the first term. We have 
\[
\intop dv\intop dwM_{\beta}(v)M_{\beta}(w)\frac{1}{\vert V\vert^{3}}\vert V\vert^{2}f\triangle g=
\]
\begin{equation}
\intop dv\intop dwM_{\beta}(v)M_{\beta}(w)\left[-\frac{\nabla f\cdot\nabla g}{\vert V\vert}+2\beta\frac{v\cdot\nabla g}{\vert V\vert}f+\frac{V\cdot\nabla g}{\vert V\vert^{3}}f\right]
\end{equation}
For the second term it results that
\[
-\intop dv\intop dwM_{\beta}(v)M_{\beta}(w)\frac{1}{\vert V\vert^{3}}\left(V,H_{v}(g)V\right)=
\]
\begin{equation}
\intop dv\intop dwM_{\beta}(v)M_{\beta}(w)\left[\frac{V\cdot\nabla g}{\vert V\vert^{3}}f+\frac{\left(V\cdot\nabla g\right)\left(V\cdot\nabla f\right)}{\vert V\vert^{3}}-2\beta\frac{\left(v\cdot V\right)\left(V\cdot\nabla g\right)}{\vert V\vert^{3}}f\right]
\end{equation}
We put together these two terms with the last one, this gives us
\[
\left(f,Q_{L}(g)\right)_{L^{2}(d\mu)}=\intop dv\intop dwM_{\beta}(v)M_{\beta}(w)\left[-\frac{\nabla f\cdot\nabla g}{\vert V\vert}+\frac{\left(V\cdot\nabla g\right)\left(V\cdot\nabla f\right)}{\vert V\vert^{3}}-2\frac{V\cdot\nabla g}{\vert V\vert^{3}}f\right]+
\]
\begin{equation}
\intop dv\intop dwM_{\beta}(v)M_{\beta}(w)\left[2\beta\frac{v\cdot\nabla g}{\vert V\vert}f-2\beta\frac{\left(v\cdot V\right)\left(V\cdot\nabla g\right)}{\vert V\vert^{3}}f\right]
\end{equation}
Now we observe that $v=w+V$ and that $2\beta wM_{\beta}(w)=-\nabla_{w}M_{\beta}(w)$,
we also integrate by parts with respect to the variable $w$ in the
second terms of (\ref{eq:au1}) and we arrive to 
\[
\intop dv\intop dwM_{\beta}(v)M_{\beta}(w)\left[2\beta\frac{v\cdot\nabla g}{\vert V\vert}f-2\beta\frac{\left(v\cdot V\right)\left(V\cdot\nabla g\right)}{\vert V\vert^{3}}f\right]=
\]
\[
\intop dv\intop dwM_{\beta}(v)\nabla_{w}M_{\beta}(w)\cdot\left[-\frac{\nabla g}{\vert V\vert}f+\frac{V\left(V\cdot\nabla g\right)}{\vert V\vert^{3}}f\right]=
\]
\begin{equation}
\intop dv\intop dwM_{\beta}(v)M_{\beta}(w)2\frac{V\cdot\nabla g}{\vert V\vert^{3}}f
\end{equation}
This yields 
\begin{equation}
\left(f,Q_{L}(g)\right)_{L^{2}(d\mu)}=\intop dv\intop dwM_{\beta}(v)M_{\beta}(w)\left[-\frac{\nabla f\cdot\nabla g}{\vert V\vert}+\frac{\left(V\cdot\nabla g\right)\left(V\cdot\nabla f\right)}{\vert V\vert^{3}}\right]\label{eq:au5}
\end{equation}
Another integration by parts leads to

\[
\left(f,Q_{L}(g)\right)_{L^{2}(d\mu)}=\intop dv\intop dwM_{\beta}(v)M_{\beta}(w)
\]
\[
\left[\frac{g\triangle f}{\vert V\vert}-2\beta\frac{\left(v\cdot\nabla f\right)g}{\vert V\vert}-\frac{\left(V\cdot\nabla f\right)g}{\vert V\vert^{3}}-\frac{\left(V\cdot\nabla f\right)g}{\vert V\vert^{3}}-\frac{\left(V,D^{2}(f)V\right)g}{\vert V\vert^{3}}+2\beta\frac{\left(v\cdot V\right)\left(V\cdot\nabla f\right)g}{\vert V\vert^{3}}\right]=
\]
\[
\intop dv\intop dwM_{\beta}(v)M_{\beta}(w)\left[\frac{g\triangle f}{\vert V\vert}-\frac{\left(V,D^{2}(f)V\right)g}{\vert V\vert^{3}}-2\frac{\left(V\cdot\nabla f\right)g}{\vert V\vert^{3}}\right]+
\]
\begin{equation}
\intop dv\intop dwM_{\beta}(v)M_{\beta}(w)\left[-2\beta\frac{\left(v\cdot\nabla f\right)g}{\vert V\vert}+2\beta\frac{\left(v\cdot V\right)\left(V\cdot\nabla f\right)g}{\vert V\vert^{3}}\right]\label{eq:au3}
\end{equation}
We integrate by parts the last term with respect to $w$, it gives
us 
\[
\intop dv\intop dwM_{\beta}(v)M_{\beta}(w)\left[-2\beta\frac{\left(v\cdot\nabla f\right)g}{\vert V\vert}+2\beta\frac{\left(v\cdot V\right)\left(V\cdot\nabla f\right)g}{\vert V\vert^{3}}\right]=
\]
\begin{equation}
\intop dv\intop dwM_{\beta}(v)M_{\beta}(w)-2\frac{V\cdot\nabla f}{\vert V\vert^{3}}g\label{eq:au4}
\end{equation}
 We use toghether (\ref{eq:au4}) and (\ref{eq:au3}) and we finally
arrive to
\begin{alignat}{1}
\left(f,Q_{L}(g)\right)_{L^{2}(d\mu)}= & \intop dv\intop dwM_{\beta}(v)M_{\beta}(w)\left[\frac{g\triangle f}{\vert V\vert}-\frac{\left(V,D^{2}(f)V\right)g}{\vert V\vert^{3}}-4\frac{\left(V\cdot\nabla f\right)g}{\vert V\vert^{3}}\right]\nonumber \\
= & \left(Q_{L}(f),g\right)_{L^{2}(d\mu)}
\end{alignat}
Obviously $D(Q_{L})=D(Q_{L}^{*})$ and so $Q_{L}$ is self-adjoint. 

We now prove that the linear Boltzmann operator $Q_{B}$ is self-adjoint.
We have 
\[
\left(f,Q_{B}(g)\right)_{L^{2}(d\mu)}=\intop dv\intop dw\intop d\nu M_{\beta}(v)M_{\beta}(w)\vert\nu\cdot V\vert f(v)\left[g(v^{'})-g(v)\right]=
\]
\begin{equation}
\intop dv\intop dw\intop d\nu M_{\beta}(v)M_{\beta}(w)\vert\nu\cdot V\vert f(v)g(v^{'})-\intop dv\intop dw\intop d\nu M_{\beta}(v)M_{\beta}(w)\vert\nu\cdot V\vert f(v)g(v)\label{eq:au6}
\end{equation}
In the first term of the sum we change variables in the integration
by using the map defined in formula (\ref{eq:ChVar}). This gives
us 
\begin{equation}
\intop dv\intop dw\intop d\nu M_{\beta}(v)M_{\beta}(w)\vert\nu\cdot V\vert f(v)g(v^{'})=\intop dv^{'}\intop dw^{'}\intop d\nu^{'}M_{\beta}(v)M_{\beta}(w)\vert\nu\cdot V\vert f(v^{'})g(v)
\end{equation}
Now we use Lemma \ref{lem:PresLeb} that gives us that $dv^{'}dw^{'}d\nu^{'}=dvdwd\nu$,
this with (\ref{eq:au6}) leads to
\begin{equation}
\left(f,Q_{B}(g)\right)_{L^{2}(d\mu)}=\intop dv\intop dw\intop d\nu M_{\beta}(v)M_{\beta}(w)\vert\nu\cdot V\vert g(v)\left[f(v^{'})-f(v)\right]=\left(Q_{B}(f),g\right)_{L^{2}(d\mu)}
\end{equation}
Formula (\ref{eq:Aa1}) and (\ref{eq:Aa2}) can be proved simply with
some integration by parts in the definition of the operators $Q_{B}$
and $Q_{L}$. This leads to 
\begin{equation}
(f,Q_{B}f)\leq0
\end{equation}
\begin{equation}
(f,Q_{L}f)\leq0
\end{equation}
\begin{equation}
\left(f,Q_{B}(f)\right)_{L^{2}(d\mu)}=\intop dv\intop dw\intop d\nu M_{\beta}(v)M_{\beta}(w)\vert\nu\cdot V\vert\left[f(v)f(v^{'})-f^{2}(v)\right]
\end{equation}
Another change of variables in the integration gives us 
\begin{equation}
\left(f,Q_{B}(f)\right)_{L^{2}(d\mu)}=\intop dv\intop dw\intop d\nu M_{\beta}(v)M_{\beta}(w)\vert\nu\cdot V\vert\left[f(v)f(v^{'})-f^{2}(v^{'})\right]
\end{equation}
We sum together these two equality, this leads to

\begin{equation}
2\left(f,Q_{B}(f)\right)_{L^{2}(d\mu)}\leq-\intop dv\intop dw\intop d\nu M_{\beta}(v)M_{\beta}(w)\vert\nu\cdot V\vert\left[f(v^{'})-f(v)\right]^{2}\leq0
\end{equation}
From formula (\ref{eq:au5}) we have

\begin{equation}
\left(f,Q_{L}(f)\right)_{L^{2}(d\mu)}=\intop dv\intop dwM_{\beta}(v)M_{\beta}(w)\left[\frac{\left(\hat{V}\cdot\nabla f\right)^{2}-\vert\nabla f\vert^{2}}{\vert V\vert}\right]
\end{equation}
and, since $\left(\hat{V}\cdot\nabla f\right)^{2}-\vert\nabla f\vert^{2}\leq0$,
it results that
\begin{equation}
\left(f,Q_{L}(f)\right)_{L^{2}(d\mu)}\leq0
\end{equation}
Now we observe that 
\begin{equation}
(f,-v\cdot\nabla_{x}f)=(v\cdot\nabla_{x}f,f)=(f,v\cdot\nabla_{x}f)=0
\end{equation}
and we arrive to 
\begin{equation}
\left(f,G_{\alpha}f\right)=\left(G_{\alpha}^{*}f,f\right)=\left(Lf,f\right)+\left((v\cdot\nabla_{x}f,f)\right)=\left(Lf,f\right)\leq0
\end{equation}
With similar steps it is possible to prove the (\ref{eq:Aa1}).

Since $D\left(G\right)$ and $D(G_{\alpha})$ are dense in $\mathbf{H}$
by the Von Neumann Theorem we have that
\begin{equation}
G_{\alpha}^{**}=\overline{G_{\alpha}}
\end{equation}
but $G_{\alpha}^{**}=G_{\alpha}$ and so $G_{\alpha}$ is closed.
This can be proved in the same way for $G$.
\end{proof}
\pagebreak{}

\section{Appendix II, estimate of the recolission set}
\begin{lem}
Let $\chi_{\lambda,E,\delta,q}$ be defined as in (\ref{eq:Cutoff})
and let $\chi\left\{ N^{P_{s}}(\epsilon_{0})\right\} $ be the characteristic
function of the set (\ref{eq:bset}). Then it results that
\[
\sum_{j_{1}=0}^{1}...\sum_{j_{s}=0}^{2^{s}-1}\left(\alpha\right)^{P_{s}-1}\sum_{\Gamma(P_{s}-1)}\sum_{\sigma_{P_{s}-1}}\boldsymbol{\sigma}_{P_{s}-1}\intop d\Lambda_{P_{s-1}}\chi_{\lambda,E,\delta,q}\chi\left\{ N^{P_{s}}(\epsilon_{0})\right\} \prod_{k=1}^{P_{s}-1}B\,f_{0,P_{s}}^{N}(\zeta(0))\leq
\]
\begin{equation}
\Vert g_{0}\Vert_{\infty}\left(C\alpha t\right)^{2^{s+1}}E^{8}2^{(s+4)(s+1)}\left(\epsilon_{0}^{\frac{2}{5}-\lambda}+\frac{\epsilon_{0}^{\frac{4}{5}-\lambda}}{\delta^{2}}+\epsilon_{0}^{\frac{4}{5}-\lambda}\right)
\end{equation}
\end{lem}
\begin{proof}
First we observe that 
\begin{equation}
N^{P_{s}}(\epsilon_{0})=\bigcup_{i=1}^{k}\bigcup_{k=2}^{P_{S}}N_{i,k}^{P_{s}}(\epsilon_{0})
\end{equation}
where 
\[
N_{i,k}^{P_{s}}(\epsilon_{0})=
\]
\begin{equation}
\left\{ (\boldsymbol{t}_{P_{s}-1},\boldsymbol{\nu}_{P_{s}-1},\boldsymbol{w}_{P_{s}-1})\in\mathbb{R}^{P_{s}-1}\times S^{2(P_{s}-1)}\times\mathbb{R}^{3(P_{s}-1)}\,|\,\min_{i<k}\min_{\tau\in[0,t_{i-1}]}\,d(r_{i}(\tau),r_{k}(\tau))<\epsilon_{0}\right\} 
\end{equation}
We also define a subsequence $t^{q}$ of the times $t_{1}...t_{n}$
associated to the virtual trajectory of particles $i$ and $k$. We
put $t^{0}$ as the time in which the two virtual trajectory merge,
then we consider the ordered union of the times of creations in the
virtual trajectory of particles $i$ and $k$ (Figure \ref{fig:times}).

For a point in $N_{i,k}^{P_{s}}(\epsilon_{0})$ there exist
\begin{equation}
\tau^{\star}=\max\left\{ \tau\in[0,t_{i-1}]\,|\,d(p^{i}(\tau),p^{k}(\tau))<\epsilon_{0}\right\} 
\end{equation}

It must be $\tau^{\star}\in[t^{l},t^{l+1})$ for some $l\geq0$. With
this definition $l$ represents the total number of creation after
the time $t^{0}$ in the virtual trajectory of the particles $i$
and $k$. For $q\in[0,l]$ we define

\begin{equation}
Y^{q}=r^{k}(t^{q})-r^{i}(t^{q})
\end{equation}
\begin{equation}
\xi_{i}^{q}=\xi_{i}(\tau)\,\,\,\xi_{k}^{q}=\xi_{k}(\tau)\,\,\,for\,\tau\in(t^{q+1},t^{q})
\end{equation}
\begin{equation}
W^{q}=\xi_{k}^{q}-\xi_{i}^{q}
\end{equation}
Observe that, since we are considering only one tree, it will be always
$Y^{0}=0$.

\begin{figure}[tb]
\begin{tikzpicture}[line cap=round,line join=round,>=triangle 45,x=1.0cm,y=1.0cm]
\clip(-3.62,-4.2) rectangle (18.8,7.12);
\draw [->] (0.,-3.7) -- (0.,6.36);
\draw (4.24,5.2)-- (7.26,-2.72);
\draw (4.736252025142386,3.898570847970961)-- (2.24,-2.72);
\draw (6.13907402974106,0.21964691538106074)-- (5.2,-2.72);
\draw (3.0999317052606106,-0.43997424204024904)-- (3.76,-2.72);
\draw [dash pattern=on 5pt off 5pt] (0.,3.92)-- (4.736252025142386,3.898570847970961);
\draw [dash pattern=on 5pt off 5pt] (0.,0.2)-- (6.13907402974106,0.21964691538106074);
\draw (5.42,-2.3) node[anchor=north west] {$k$};
\draw (-0.42,1.3) node[anchor=north west] {$t^1$};
\draw (-0.4,4.32) node[anchor=north west] {$t^0$};
\draw (-1.26,0.58) node[anchor=north west] {$t^l=t^2$};
\draw (3.586463928137734,0.850018897474325)-- (4.88,-2.72);
\draw(5.06,-2.72) circle (0.30463092423455623cm);
\draw [dash pattern=on 5pt off 5pt] (0.,0.86)-- (3.586463928137734,0.8500188974743249);
\draw (4.4,-2.) node[anchor=north west] {$i$};
\draw (3.586463928137734,0.8500188974743249)-- (5.895546523992718,0.8583018311184354);
\draw (4.82,4.34) node[anchor=north west] {$Y^0$};
\draw (4.74,1.58) node[anchor=north west] {$Y^1$};
\draw (4.74,0.88) node[anchor=north west] {$Y^2$};
\draw (6.13907402974106,0.21964691538106043)-- (3.82,0.22);
\end{tikzpicture}

\caption{\label{fig:times}The virtual trajectory of the praticles $i$ and
$k$ and their backward history}
\end{figure}
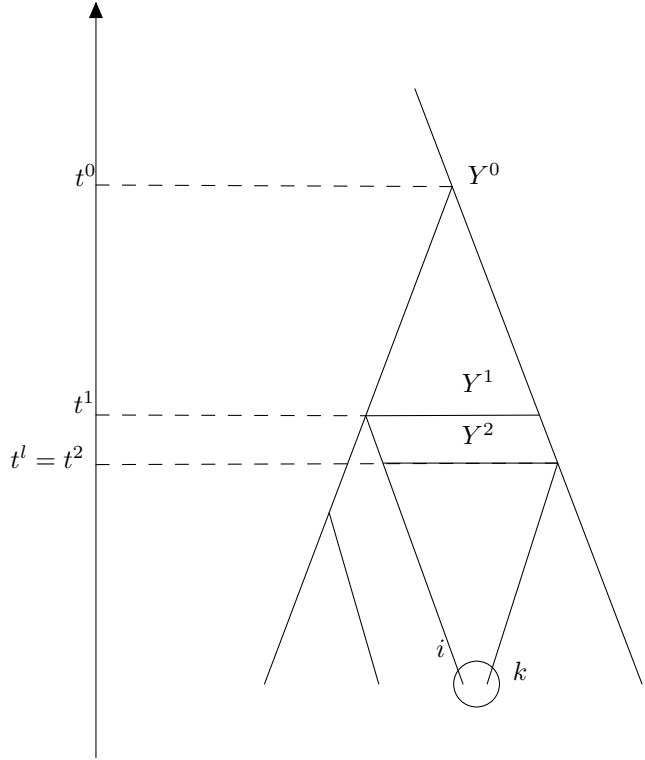

First suppose that $l=0$, this means that the particles $i$ and
$k$ have a recolission after the creation of the particle $k$. This
can happen in two cases. In the first case the particles $i$ and
$h$ do not separate enough after the creation. In the second case
the particles, after being separated enough, perform a recolission
since the trajectory on the torus have no dispersive properties.

In the first case it must be 
\begin{equation}
\vert W^{0}\vert(t^{1}-t^{0})\leq\epsilon_{0}\label{eq:recoll}
\end{equation}
We recall that the cutoff (\ref{eq:Cutoff}) implies that $(t^{1}-t^{0})\geq\delta$
and that $\vert W^{0}\vert\geq\epsilon^{q}$. Then the particles must
be separated at least by a distance of $\delta\epsilon^{q}$. We choose
the parameters in such a way that 
\begin{equation}
\epsilon_{0}\leq\delta\epsilon^{q}
\end{equation}
and this gives us that the (\ref{eq:recoll}) cannot happen.

In the second case we prove that $W^{0}$ must be in a set of small
measure. There exist a $\tau>\delta$ such that

\begin{equation}
d(r^{i}(t^{0}-\tau),r^{k}(t^{0}-\tau))\leq\epsilon_{0}
\end{equation}
We use the correspondence between the torus and the whole space with
periodic structure. We have that 
\begin{equation}
\left(r^{i}(t^{0})-\tau\xi^{i}(t^{0})\right)-\left(r^{k}(t^{0})-\tau\xi^{k}(t^{0})\right)\in\bigcup_{p\in\mathbb{Z}^{3}}B_{\epsilon_{0}}(p)
\end{equation}
Thanks to the energy cutoff we have that $\vert W^{0}\vert\leq4E$
and so
\begin{equation}
\tau W^{0}\in\left(\bigcup_{p\in\mathbb{Z}^{3}}B_{\epsilon_{0}}(r^{i}(t^{0})-r^{k}(t^{0})+p)\right)\bigcap B_{4Et}(0)
\end{equation}
Suppose that $\vert r^{i}(t^{0})-r^{k}(t^{0})+p\vert<\frac{1}{4}$.
This can happen only for a value of $p$ since the distance between
the centers of the spheres is $1$. Taking $\hat{v}$ a unit vector
normal to $r^{i}(t^{0})-r^{k}(t^{0})+p$ it results that 
\begin{equation}
\tau\vert W^{0}\cdot\hat{v}\vert\leq\epsilon_{0}
\end{equation}
and then
\begin{equation}
\vert W^{0}\cdot\hat{v}\vert\leq\frac{\epsilon_{0}}{\delta}
\end{equation}
This implies that $W^{0}$ is in the intersection of $B_{4E}(0)$
and a cylinder of radius $\frac{\epsilon_{0}}{\delta}$ and so in
a set of measure bounded by $CE\frac{\epsilon_{0}^{2}}{\delta^{2}}$.
Suppose now $\vert r^{i}(t^{0})-r^{k}(t^{0})+p\vert\geq\frac{1}{4}$
and that $\epsilon$ is small enough, then $W^{0}$ is in the intersection
of $B_{4E}(0)$ and some cone of vertex $0$ and solid angle $C\epsilon_{0}^{2}$
and these cones are at most $(8Et)^{3}$. Finally putting together
these two estimates gives us that $W^{0}$ must be in a suitable set
$B_{k,0}$ such that 
\begin{equation}
\vert B_{k,0}\vert\leq C\left(E\frac{\epsilon_{0}^{2}}{\delta^{2}}+(Et)^{3}\epsilon_{0}^{2}\right)\label{eq:rec3}
\end{equation}
We can now suppose that $l\geq1$. The $\epsilon_{0}$-overlap is
verified only if 
\begin{equation}
Y^{l}-\tau W^{l}\in\bigcup_{p\in\mathbb{Z}^{3}}B_{\epsilon_{0}}(p)\label{eq:rec1}
\end{equation}
for some $\tau\in[0,t^{l})$. Moreover it results that 

\begin{equation}
Y^{l}=\hat{p}-\sum_{q=0}^{l-1}W^{q}\left(t^{q}-t^{q+1}\right)=\hat{p}-W^{0}t^{0}+\sum_{q=1}^{l}\left(W^{q-1}-W^{q}\right)t^{q}\label{eq:rec2}
\end{equation}
where $\hat{p}\in\mathbb{Z}^{3}$ is chosen in such a way that the
right hand side of equation (\ref{eq:rec2}) is a point in the torus. 

Now we prove that it must be
\begin{equation}
\sum_{q=1}^{l}\vert W^{q}-W^{q-1}\vert>\epsilon_{0}^{\frac{2}{5}}\label{eq:rec4}
\end{equation}
Otherwise it would be $\vert W^{q}-W^{0}\vert\leq\epsilon_{0}^{\frac{2}{5}}$
for all $q$, then using (\ref{eq:rec2}) and (\ref{eq:rec1}) it
results that
\begin{equation}
W^{0}(\tau+t^{0}-t^{l})\in\bigcup_{p\in\mathbb{Z}^{3}}B_{(\epsilon_{0}+\epsilon_{0}^{\frac{2}{5}}t)}(p)
\end{equation}
Since $\tau+t^{0}-t^{l}\geq\delta$ we can perform the same steps
of estimate (\ref{eq:rec3}) to prove that in this case $W^{0}$must
be in a set $B_{k,1}$ of measure bounded by 
\begin{equation}
C\left(E\frac{\epsilon_{0}^{\frac{4}{5}}}{\delta^{2}}+(Et)^{3}\epsilon_{0}^{\frac{4}{5}}\right)
\end{equation}
Condition (\ref{eq:rec1}) implies that 
\begin{equation}
\vert\left(Y^{l}+\hat{p}\right)\wedge\hat{W^{l}}\vert\leq\epsilon_{0}\label{eq:rec5}
\end{equation}
with $\hat{W^{l}}=\frac{W^{l}}{\vert W^{l}\vert}$. Then from (\ref{eq:rec2})
we have that
\begin{equation}
\vert\left(\hat{p}-W^{0}t^{0}\right)\wedge\hat{W^{l}}-\sum_{q=1}^{l}\left[\left(W^{q}-W^{q-1}\right)\wedge\hat{W^{l}}\right]t^{q}\vert\leq\epsilon_{0}
\end{equation}
Now suppose that 
\begin{equation}
\sum_{q=1}^{l}\vert\left(W^{q}-W^{q-1}\right)\wedge\hat{W^{l}}\vert\leq\epsilon_{0}^{\frac{3}{5}}\label{eq:rec6}
\end{equation}
from (\ref{eq:rec4}) it must exist a $\bar{q}\in\left\{ 1,...,...,l\right\} $
such that 
\begin{equation}
U=U^{\bar{q}}=W^{\bar{q}}-W^{\bar{q}-1}
\end{equation}
has modulus 
\begin{equation}
\vert U\vert>\frac{\epsilon_{0}^{\frac{2}{5}}}{l}
\end{equation}
Moreover from (\ref{eq:rec5}) it results that 
\begin{equation}
\vert U\wedge\hat{W^{l}}\vert\leq\epsilon_{0}^{\frac{3}{5}}
\end{equation}
We set $\hat{U}=\frac{U}{\vert U\vert}$, this gives us 
\begin{equation}
\vert\hat{U}\wedge\hat{W^{l}}\vert\leq(P_{s}-1)\epsilon_{0}^{\frac{1}{5}}\label{eq:rec7}
\end{equation}
Thanks to cutoff (\ref{eq:Cutoff}) it results that $\vert W^{0}\vert>\epsilon^{q}$,
that with (\ref{eq:rec6}) gives us that 
\begin{equation}
\vert\hat{W^{0}}\wedge\hat{W^{l}}\vert\leq\epsilon_{0}^{\frac{3}{5}-q}
\end{equation}
This with (\ref{eq:rec7}), assuming $q=\frac{1}{8}$, finally gives
\begin{equation}
\vert\hat{W^{0}}\wedge\hat{U}\vert\leq C\epsilon_{0}^{\frac{1}{5}}(P_{s}-1)\label{eq:recest}
\end{equation}
We have two cases, if $\sum_{q=1}^{l}\vert\left(W^{q}-W^{q-1}\right)\wedge\hat{W^{l}}\vert\leq\epsilon_{0}^{\frac{3}{5}}$
then it results that $\vert\hat{W^{0}}\wedge\hat{U}\vert\leq C\epsilon_{0}^{\frac{1}{5}}(P_{s}-1)$.
Otherwise we have that $\sum_{q=1}^{l}\vert\left(W^{q}-W^{q-1}\right)\wedge\hat{W^{l}}\vert>\epsilon_{0}^{\frac{3}{5}}$.
This implies that for some $q^{\star}$ we have 
\begin{equation}
\vert\left(W^{q^{\star}}-W^{q^{\star}-1}\right)\wedge\hat{W^{l}}\vert>\frac{\epsilon_{0}^{\frac{3}{5}}}{l}
\end{equation}
From (\ref{eq:rec2}) it follows that
\begin{equation}
\vert\left(\hat{p}-W^{0}t^{0}\right)\wedge\hat{W^{l}}-\sum_{q=1}^{l}\left[\left(W^{q}-W^{q-1}\right)\wedge\hat{W^{l}}\right]t^{q}\vert\leq\epsilon_{0}
\end{equation}
and then
\begin{equation}
\vert\left(\hat{p}-W^{0}t^{0}\right)\wedge\hat{W^{l}}-\vert\left(W^{q^{\star}}-W^{q^{\star}-1}\right)\wedge\hat{W^{l}}\vert t^{q^{\star}}-\sum_{q=1,q\neq q^{\star}}^{l}\left[\left(W^{q}-W^{q-1}\right)\wedge\hat{W^{l}}\right]t^{q}\vert\leq\epsilon_{0}
\end{equation}
This last formula implies that, for a fixed $\hat{p}$, $t^{q^{\star}}$
must be in a interval of length smaller than 
\begin{equation}
\epsilon_{0}\vert\left(W^{q^{\star}}-W^{q^{\star}-1}\right)\wedge W^{l}\vert^{-1}
\end{equation}
that from (\ref{eq:recest}) is bounded by $\epsilon_{0}^{\frac{2}{5}}(P_{s}-1$).
Since the possible choices of $\hat{p}$ are at most $\left(CEt\right)^{3}$
it results that $t^{q^{\star}}$ is in a set of measure bounded by
\begin{equation}
\epsilon_{0}^{\frac{2}{5}}(P_{s}-1)\left(CEt\right)^{3}\label{eq:rec7-1}
\end{equation}
We summarize as follows. We denote with $V_{r_{1}}$ and $V_{r_{1}}^{'}$
respectively the outgoing and incoming relative velocities of the
collision at time $\tau_{r_{1}}$ in the BBF. Let $t^{\bar{q}}=t_{r_{2}}$
and $U^{\bar{q}}=U_{r_{2}}$that is a function of $V_{r_{2}},\nu_{r_{2}}$
only. We have that 
\[
\chi_{\lambda,E,\delta,q}\chi\left\{ N^{P_{s}}(\epsilon_{0})\right\} \leq\chi_{\lambda,E,\delta,q}\sum_{r=1}^{P_{s}-1}\chi\left\{ V_{r}\in B_{r,0}\bigcup B_{r,1}\right\} +
\]
\begin{equation}
\sum_{i,k}\sum_{l=1}^{n_{ik}}\sum_{q^{\star}=1}^{l}\chi_{\lambda,E,\delta,q}\chi\left\{ N_{ik}^{l,q^{\star}}(\epsilon_{0})\right\} +\sum_{r_{1}=1}^{P_{s}-1}\sum_{r_{2}=r_{1}+1}^{P_{s}-1}\chi_{\lambda,E,\delta,q}\chi\left\{ \vert\hat{V_{r_{1}}^{'}}\wedge\hat{U}_{r_{2}}\vert\leq C\epsilon_{0}^{\frac{1}{5}}(P_{s}-1)\right\} \label{eq:rec8}
\end{equation}
where $n_{ik}$ are the total number of creation in the virtual trajectory
of the particles $i$ and $k$ between the time $t^{0}$ and the time
$t$. 
\begin{equation}
N_{ik}^{l,q^{\star}}(\epsilon_{0})=\left\{ \boldsymbol{t}_{P_{s}-1},\boldsymbol{\nu}_{P_{s}-1},\boldsymbol{w}_{P_{s}-1}|\,\text{the virtual trajectories of i and k satisfies (9.16),}\right.
\end{equation}
\[
\left.\,\text{with}\,\vert W^{q^{\star}}-W^{q^{\star}-1}\wedge\hat{W^{l}}\vert\geq\frac{\epsilon_{0}^{\frac{3}{5}}}{l}\right\} 
\]
We now estimate the three terms in the right hand side of (\ref{eq:rec8}).
For the first term by a simple change of variables we have that
\[
\sum_{j_{1}=0}^{1}...\sum_{j_{s}=0}^{2^{s}-1}\left(\alpha\right)^{P_{s}-1}\sum_{\Gamma(P_{s}-1)}\sum_{\sigma_{P_{s}-1}}\boldsymbol{\sigma}_{P_{s}-1}
\]
\[
\intop d\Lambda_{P_{s-1}}\chi_{\lambda,E,\delta,q}\sum_{r=1}^{P_{s}-1}\chi\left\{ V_{r}\in B_{r,0}\bigcup B_{r,1}\right\} \prod_{k=1}^{P_{s}-1}B\,f_{0,P_{s}}^{N}(\zeta(0))\leq
\]
\begin{equation}
\epsilon^{-\lambda}\Vert g_{0}\Vert_{\infty}\left(C\alpha t\right)^{2^{s+1}}C\left(E\frac{\epsilon_{0}^{\frac{4}{5}}}{\delta^{2}}+(Et)^{3}\epsilon_{0}^{\frac{4}{5}}\right)
\end{equation}
For the second term from (\ref{eq:rec7-1}) it follows that 
\[
\sum_{j_{1}=0}^{1}...\sum_{j_{s}=0}^{2^{s}-1}\left(\alpha\right)^{P_{s}-1}\sum_{\Gamma(P_{s}-1)}\sum_{\sigma_{P_{s}-1}}
\]
\[
\boldsymbol{\sigma}_{P_{s}-1}\intop d\Lambda_{P_{s-1}}\chi_{\lambda,E,\delta,q}\sum_{i,k}\sum_{l=1}^{n_{ik}}\sum_{q^{\star}=1}^{l}\chi_{\lambda,E,\delta,q}\chi\left\{ N_{ik}^{l,q^{\star}}(\epsilon_{0})\right\} \prod_{k=1}^{P_{s}-1}B\,f_{0,P_{s}}^{N}(\zeta(0))\leq
\]
\begin{equation}
\epsilon^{-\lambda}\Vert g_{0}\Vert_{\infty}\left(C\alpha t\right)^{2^{s+1}}E^{3}2^{(s+4)(s+1)}\epsilon_{0}^{\frac{2}{5}}
\end{equation}
The last term to be estimated is 
\[
\sum_{j_{1}=0}^{1}...\sum_{j_{s}=0}^{2^{s}-1}\left(\alpha\right)^{P_{s}-1}\sum_{\Gamma(P_{s}-1)}\sum_{\sigma_{P_{s}-1}}\boldsymbol{\sigma}_{P_{s}-1}
\]
\begin{equation}
\intop d\Lambda_{P_{s-1}}\chi_{\lambda,E,\delta,q}\sum_{k_{1}=1}^{P_{s}-1}\sum_{k_{2}=k_{1}+1}^{P_{s}-1}\chi\left\{ \vert\hat{V_{k_{1}}^{'}}\wedge\hat{U}_{k_{2}}\vert\leq C\epsilon_{0}^{\frac{1}{5}}(P_{s}-1)\right\} \prod_{k=1}^{P_{s}-1}B\,f_{0,P_{s}}^{N}(\zeta(0))
\end{equation}
We first consider the set
\begin{equation}
\intop d\Lambda_{P_{s-1}}\chi_{\lambda,E,\delta,q}\chi\left\{ \vert\hat{V_{k_{1}}^{'}}\wedge\hat{U}_{k_{2}}\vert\leq C\epsilon_{0}^{\frac{1}{5}}(P_{s}-1)\right\} e^{-\frac{\beta}{2}\vert\xi^{\epsilon}(0)\vert^{2}}
\end{equation}
We change the integration variables in the following way 
\begin{equation}
\left(\nu_{k_{1}},w_{k_{1}},\nu_{k_{2}},w_{k_{2}}\right)\rightarrow\left(\nu_{k_{1}}^{'},V_{k_{1}}^{'},\nu_{k_{2}},V_{k_{2}}\right)\label{eq:recchg}
\end{equation}
where $V_{k_{1}}^{'}=w_{k_{1}}^{'}-\xi_{i_{k_{1}}}^{'}(\tau_{k_{1}})$
and $V_{k_{2}}=w_{k_{1}}-\xi_{i_{k_{1}}}(\tau_{k_{1}})$. From Lemma
\ref{lem:PresLeb} it follows that (\ref{eq:recchg}) is a change
of variables that preserve the measure. Thanks to this change of variables
a simple calculation leads to 
\begin{equation}
\intop d\Lambda_{P_{s-1}}\chi_{\lambda,E,\delta,q}\chi\left\{ \vert\hat{V_{k_{1}}^{'}}\wedge\hat{U}_{k_{2}}\vert\leq C\epsilon_{0}^{\frac{1}{5}}(P_{s}-1)\right\} e^{-\frac{\beta}{2}\vert\xi^{\epsilon}(0)\vert^{2}}\leq E^{5}\epsilon_{0}^{\frac{2}{5}}2^{s+1}
\end{equation}
This implies that 
\[
\sum_{j_{1}=0}^{1}...\sum_{j_{s}=0}^{2^{s}-1}\left(\alpha\right)^{P_{s}-1}\sum_{\Gamma(P_{s}-1)}\sum_{\sigma_{P_{s}-1}}\boldsymbol{\sigma}_{P_{s}-1}\intop d\Lambda_{P_{s-1}}\chi_{\lambda,E,\delta,q}\sum_{k_{1}=1}^{P_{s}-1}\sum_{r_{2}=k_{1}+1}^{P_{s}-1}\chi\left\{ \vert\hat{V_{k_{1}}^{'}}\wedge\hat{U}_{k_{2}}\vert\leq C\epsilon_{0}^{\frac{1}{5}}(P_{s}-1)\right\} 
\]
\begin{equation}
\prod_{k=1}^{P_{s}-1}B\,f_{0,P_{s}}^{N}(\zeta(0))\leq\epsilon^{-\lambda}\Vert g_{0}\Vert_{\infty}\left(C\alpha t\right)^{2^{s+1}}2^{(s+2)(s+1)}E^{5}\epsilon_{0}^{\frac{2}{5}}
\end{equation}
Finally from these estimates we arrive to
\[
\sum_{j_{1}=0}^{1}...\sum_{j_{s}=0}^{2^{s}-1}\left(\alpha\right)^{P_{s}-1}\sum_{\Gamma(P_{s}-1)}\sum_{\sigma_{P_{s}-1}}\boldsymbol{\sigma}_{P_{s}-1}\intop d\Lambda_{P_{s-1}}\chi_{\lambda,E,\delta,q}\chi\left\{ N^{P_{s}}(\epsilon_{0})\right\} \prod_{k=1}^{P_{s}-1}B\,f_{0,P_{s}}^{N}(\zeta(0))\leq
\]
\begin{equation}
\epsilon^{-\lambda}\Vert g_{0}\Vert_{\infty}\left(C\alpha t\right)^{2^{s+1}}E^{8}2^{(s+4)(s+1)}\left(\epsilon_{0}^{\frac{2}{5}}+\frac{\epsilon_{0}^{\frac{4}{5}}}{\delta^{2}}+\epsilon_{0}^{\frac{4}{5}}\right)
\end{equation}
\end{proof}
\newpage{}

\bibliographystyle{AIMS}
\bibliography{bib}

\end{document}